\documentclass[11pt,letterpaper]{article}
\usepackage[margin=1in]{geometry}

\usepackage{amsmath, amssymb, amsthm, mathtools}
\usepackage{enumitem}
\usepackage{aliascnt}
\usepackage{booktabs}

\newcommand{\ignore}[1]{}

\usepackage{array}
\usepackage{xcolor}
\definecolor{linkorange}{RGB}{196,93,0}
\definecolor{citeturquoise}{RGB}{0,125,145}

\usepackage{algorithm}
\usepackage[
    colorlinks=true,
    linkcolor=linkorange,
    citecolor=citeturquoise,
    urlcolor=citeturquoise
]{hyperref}
\usepackage[capitalise,nameinlink,noabbrev]{cleveref}
\usepackage[numbers,sort&compress]{natbib}

\usepackage{algpseudocode}
\algtext*{EndProcedure}

\newlist{assumptionlist}{enumerate}{1}
\setlist[assumptionlist,1]{
    label=(\arabic*),
    ref=\theassumption(\arabic*)
}

\crefname{assumptionlisti}{assumption}{assumptions}
\Crefname{assumptionlisti}{Assumption}{Assumptions}

\newcommand{\Sinplabeled}{{\bar S}_{\mathrm{inp}}}
\newcommand{\Sinpunlabeled}{{S}_{\mathrm{inp}}}
\newcommand{\Sclnlabeled}{{\bar S}_{\mathrm{cln}}}
\newcommand{\Sclnunlabeled}{{S}_{\mathrm{cln}}}
\newcommand{\Sfiltlabeled}{{\bar S}_{\mathrm{filt}}}
\newcommand{\Sfiltunlabeled}{{S}_{\mathrm{filt}}}

\newcommand{\Srefunlabeled}{{S}_{\mathrm{ref}}}

\newcommand{\C}{\mathcal{C}}
\newcommand{\err}{\mathrm{err}}
\DeclareMathOperator*{\E}{\mathbb{E}}
\newcommand{\sign}{\mathrm{sign}}

\newcommand{\R}{\mathbb{R}}

\numberwithin{equation}{section}

\newtheorem{theorem}{Theorem}[section]

\newaliascnt{lemma}{theorem}
\newtheorem{lemma}[lemma]{Lemma}
\aliascntresetthe{lemma}

\newaliascnt{corollary}{theorem}
\newtheorem{corollary}[corollary]{Corollary}
\aliascntresetthe{corollary}

\newaliascnt{fact}{theorem}
\newtheorem{fact}[fact]{Fact}
\aliascntresetthe{fact}

\newaliascnt{assumption}{theorem}

\aliascntresetthe{assumption}

\newaliascnt{definition}{theorem}
\newtheorem{definition}[definition]{Definition}
\aliascntresetthe{definition}

\newaliascnt{claim}{theorem}

\aliascntresetthe{claim}

\newaliascnt{remark}{theorem}

\aliascntresetthe{remark}

\newlist{factlist}{enumerate}{1}
\setlist[factlist]{label=\arabic*.,ref=\arabic*}

\crefname{factlisti}{Fact}{Facts}
\Crefname{factlisti}{Fact}{Facts}
\title{Efficient Robust Learning at the Information-Theoretic Limit}
\author{%
    \begin{tabular}{cc}
        \begin{tabular}{c}
            Adam R. Klivans\thanks{Supported by the NSF AI Institute for Foundations of Machine Learning (IFML) and NSF AF: MEDIUM: Efficient Algorithms for Learning with Distribution Shift.}\\ \texttt{klivans@utexas.edu} \\ UT Austin 
        \end{tabular} & 
        \begin{tabular}{c}
             Konstantinos Stavropoulos\thanks{This work was done in part while the author was at The University of Texas at Austin. The author was supported in part by the NSF AI Institute for Foundations of Machine Learning (IFML), and by the Institute for Advanced Study and an anonymous donor.} \\ \texttt{kstavrop@amias.ias.edu} \\ Institute for Advanced Study
        \end{tabular}
        \\\\
        \begin{tabular}{c}
            Sergei Tikhonov \\ \texttt{tikhonov@utexas.edu} \\ UT Austin
        \end{tabular}
         & 
         \begin{tabular}{c}
              Arsen Vasilyan\thanks{Work done while at UT Austin. Supported by the NSF AI Institute for Foundations of Machine Learning (IFML).} \\ \texttt{vasilyan@cs.aarhus.dk} \\ Aarhus University
         \end{tabular}
    \end{tabular}
}
\date{}

\begin{document}

\maketitle

\begin{abstract}
\normalsize
In an important recent work, Blanc \cite{Blanc26} gave an algorithm for robustly learning Boolean concept classes with respect to a fixed distribution that outputs a (randomized) classifier achieving the optimal error of $\eta + \varepsilon$ where $\eta$ is the noise rate.   In contrast, it is well known that deterministic hypotheses cannot achieve error less than $2\eta + \varepsilon.$

Blanc's algorithm is computationally inefficient, and the main problem left open in his work is to find a polynomial-time algorithm given access to an oracle for empirical risk minimization (ERM).  Here we resolve this problem and give such an algorithm.  Additionally, we give an efficient algorithm (no ERM oracle required) for robustly learning any function class that admits low-degree polynomial approximators with respect to hypercontractive distributions. This captures basic settings including halfspaces with respect to Gaussian marginals.  

Our techniques make new connections to various types of online convex optimization algorithms. As a surprising consequence, we show that efficient SQ learnability in the agnostic setting implies efficient learning with contamination (with respect to Gaussians or uniform over the hypercube).

\end{abstract}

\thispagestyle{empty}
\newpage
\setcounter{page}{1}

\section{Introduction}

Developing algorithms that are robust to potentially adversarially contaminated datasets remains a fundamental challenge in AI, machine learning, and statistics more generally.  Despite much impressive work over the past decade in the unsupervised setting (see e.g., \cite{diakonikolas2019robust,diakonikolas2023algorithmic}), the case of robust {\em supervised} learning, especially for Boolean concept classes, has received far less attention.

In this paper we focus on supervised learning in the bounded-contamination model~\cite{BEK02}, where an $\eta$-fraction of the training set may be replaced arbitrarily, so both features and labels may be corrupted\footnote{This model is essentially the $\eta$-nasty-noise model introduced by~\cite{BEK02}, in which the number of replaced examples is drawn from a binomial distribution with rate $\eta$. Here we refer to this as {\em bounded contamination} to align with recent work in algorithmic robust statistics. Our results apply to both models.}:

\begin{definition}[Bounded contamination, {\cite[Definition~1.1]{KSTV25}}]\label{def:bounded-contamination}
Let $\mathcal{D}$ be a distribution over $X$, let $\eta\in[0,1)$, and let
$c^\star:X\to\{\pm1\}$. We say that a labeled sample $\bar S_{\mathrm{inp}}$ is generated from $(\mathcal{D}, c^\star)$ with bounded contamination of rate $\eta$ if, for some $n\ge1$:
\begin{enumerate}
    \item First, a labeled sample $\bar S_{\mathrm{cln}}$ of $n$ i.i.d. examples of the form $(x,c^\star(x))$, with $x\sim \mathcal{D}$, is drawn.

    \item Then, an adversary receives $\bar S_{\mathrm{cln}}$, chooses at most $\eta n$ labeled examples in $\bar S_{\mathrm{cln}}$, and replaces them by an equal number of arbitrary labeled examples to form $\bar S_{\mathrm{inp}}$.
\end{enumerate}
\end{definition}

We assume that the learner receives the corrupted labeled sample $\bar S_{\mathrm{inp}}$, independently draws a reference sample $S_{\mathrm{ref}}$ of i.i.d. unlabeled examples from $\mathcal{D}$, and aims to output a classifier with approximately optimal error under $\mathcal{D}$.

\begin{definition}[Fixed-distribution learning under bounded contamination]
\label{def:fixed-distribution-learning}
Fix a distribution $\mathcal{D}$ over $X$ and a concept class $\mathcal C\subseteq\{\pm1\}^X$. An algorithm $\mathcal A$ is a fixed-distribution learner under bounded contamination for $\mathcal C$ with respect to $\mathcal{D}$ if, for every $\varepsilon,\delta\in(0,1)$, $\eta\in[0,1)$, and $c^\star\in\mathcal C$, given a sample $\bar S_{\mathrm{inp}}$ generated from $(\mathcal{D},c^\star)$ with $\eta$-bounded contamination and an independent unlabeled reference sample $S_{\mathrm{ref}}$ from $\mathcal{D}$, the algorithm outputs a randomized hypothesis $h$ such that, with probability at least $1-\delta$,
$$
\err_{\mathcal{D}}(h,c^\star) := \Pr_{x\sim \mathcal{D}, h} \!\left[h(x)\neq c^\star(x)\right] \le \eta+\varepsilon.
$$
The sample complexity of $\mathcal A$ records the sample sizes of
$\bar S_{\mathrm{inp}}$ and $S_{\mathrm{ref}}$ required to achieve this
guarantee.
\end{definition}

Recent work by \cite{Blanc26} showed that randomized hypotheses can achieve the information-theoretically optimal error rate $\eta+\varepsilon$. His algorithm for generating these hypotheses, however, is not known to be efficiently implementable, even with access to an empirical risk minimization (ERM) oracle. This leads to the central question of our work: \emph{can one efficiently achieve the information-theoretically optimal
error rate $\eta+\varepsilon$ when the concept class admits low-degree polynomial approximations or when the learner has access to an ERM oracle?}

\subsection{Our Results}

\paragraph{Learning via an ERM oracle} Our first result resolves an open question from \cite{Blanc26}, by demonstrating that efficient learning with bounded contamination up to optimum error $\eta+\varepsilon$ is possible given an ERM oracle.

\begin{theorem}[Learning with an ERM oracle, see \Cref{thm:erm-main-guarantee} for the full version]
\label{thm:erm-main-guarantee-informal}
Let $\mathcal{D}$ be a fixed distribution and let $\mathcal C$ be a concept class of finite VC dimension. There is a polynomial-time algorithm with access to an ERM oracle for $\mathcal C$ that learns $\mathcal C$ on $\mathcal{D}$ with $\eta$-bounded contamination and error at most $\eta+\varepsilon$ using
$
O\left(
\frac{
\operatorname{VCdim}(\mathcal C)\log(2/\varepsilon)
+
\log(1/\delta)
}{
\varepsilon^4
}
\right)
$
samples.
\end{theorem}

Note that an error rate of $2\eta+\varepsilon$ is straightforward to achieve by using an ERM oracle on the corrupted input dataset. However, before the work of \cite{Blanc26}, achieving the optimal rate of $\eta+\varepsilon$ had remained an open problem for more than two decades \cite{BEK02}, even without requiring an efficient algorithm relative to an ERM oracle. Our work demonstrates that the optimal rate can be achieved via an ERM-efficient reduction.

\paragraph{Efficient Algorithms via Polynomial Approximation} While our ERM reduction provides new algorithmic insights for the problem of learning with bounded contamination that are of independent interest, it does not directly lead to end-to-end efficient algorithms, even for very simple classes. For instance, it is well known that implementing an ERM oracle for the basic class of linear classifiers is an NP-hard problem in the worst case \cite{guruswami2009hardness,feldman09agnosticlearning,Dan16}.

In the context of agnostic learning, where there is only label noise, it is well known that computational intractability can be circumvented under the assumption that the ground truth classifier can be approximated by low-degree polynomials \cite{linial1993constant,kalai2008agnostically}. Recent work on learning with bounded contamination showed that polynomial approximation with respect to the clean distribution  suffices for efficient learnability even under bounded contamination up to error $2\eta+\varepsilon$. Here, we provide an efficient algorithm that achieves the optimal error guarantee of $\eta+\varepsilon$, under the same assumptions.

\begin{theorem}[Learning via ${\cal L}_2$ approximation; see \Cref{thm:main-guarantee-restated} for the full version]
\label{thm:main-guarantee}
Let $\mathcal{D}$ be an $A$-hypercontractive distribution and suppose that $\mathcal C$ admits degree-$\ell$ ${\cal L}_2$ polynomial approximators with squared error at most $c_0\varepsilon^2$, for a sufficiently small universal constant $c_0>0$. There is an algorithm that learns $\mathcal C$ on $\mathcal{D}$ with $\eta$-bounded contamination and error $\eta+\varepsilon$ and probability of failure $\delta$, using $\frac{\left(A(d+1)\log(2/\delta)\right)^{O(\ell)}}{\varepsilon^3}$ samples and $\mathrm{poly}(A,\frac{1}{\varepsilon},(d+1)^\ell,\log^\ell(1/\delta)) $ time.
\end{theorem}

Note that all prior work  obtained an error guarantee of $2\eta+\varepsilon$ \cite{KSTV25} or worse for this setting, even for the basic case of learning a single halfspace with respect to Gaussian marginals. Similarly to all prior work on efficient learning with bounded contamination (see, e.g., \cite{diakonikolas-learning-geometric-concepts-with-nasty-noise,KSTV25}), we work in the distribution-specific setting: the marginal distribution $\mathcal{D}$ is fixed and the learner receives an independent unlabeled reference sample $S_{\mathrm{ref}}$ from $\mathcal{D}$.

Our main result holds for hypercontractive clean distributions (\Cref{def:hypercontractivity}) and concept classes with low-degree ${\cal L}_2$ polynomial approximators (\Cref{def:polynomial-approximators}).
Some restriction on $\mathcal{D}$ is necessary: even learning halfspaces with label noise in the distribution-free setting is believed to be computationally intractable \cite{FGKP06,Dan16}. At the same time, our result is not limited to a particular distribution. The class of hypercontractive distributions includes many standard distributions, such as Gaussians, log-concave distributions, and product distributions over $\{\pm1\}^d$.

\begin{table*}[t]
\centering
\small
\renewcommand{\arraystretch}{1}
\setlength{\tabcolsep}{5pt}
\begin{tabular*}{\textwidth}{@{\extracolsep{\fill}}cccc@{}}
\toprule
\textbf{Concept Class}
&
\textbf{Distribution}
&
\textbf{Our Runtime}
&
\textbf{References}
\\
\midrule
Single halfspace
&
\begin{tabular}{c}
$\mathcal N(0,I_d)$ or\\
$\operatorname{Unif}(\{\pm1\}^d)$
\end{tabular}
&
$d^{\widetilde O(1/\varepsilon^2)}$
&
\begin{tabular}{c}
\cite{kalai2008agnostically}\\
\cite{KOSOS04}
\end{tabular}
\\
\midrule
Intersections of $k$ halfspaces
&
\begin{tabular}{c}
$\mathcal N(0,I_d)$ or\\
$\operatorname{Unif}(\{\pm1\}^d)$
\end{tabular}
&
$d^{\widetilde O(\log(k)/\varepsilon^2)}$
&
\begin{tabular}{c}
\cite{klivans2008learning}\\
\cite{kane2014average}
\end{tabular}
\\
\midrule
Arbitrary functions of $k$ halfspaces
&
\begin{tabular}{c}
$\mathcal N(0,I_d)$ or\\
$\operatorname{Unif}(\{\pm1\}^d)$
\end{tabular}
&
$d^{\widetilde O(k^2/\varepsilon^2)}$
&
\begin{tabular}{c}
\cite{klivans2008learning}\\
\cite{KOSOS04}
\end{tabular}
\\
\midrule
Arbitrary functions of $k$ halfspaces
&
\begin{tabular}{c}
Subgaussian\\
Log-Concave
\end{tabular}
&
$d^{\widetilde O(k^5/\varepsilon^6)}$
&
\cite{KSV26}
\\
\midrule
Convex sets of intrinsic dimension $k$
&
$\mathcal N(0,I_d)$
&
$d^{\widetilde O(\sqrt{k}/\varepsilon^2)}$
&
\cite{klivans2008learning}
\\
\midrule
Convex sets of intrinsic dimension $k$
&
\begin{tabular}{c}
Subgaussian\\
Log-Concave
\end{tabular}
&
$d^{\widetilde O(k^6/\varepsilon^6)}$
&
\begin{tabular}{c}
\cite{KSV26}\\
\cite{BGHT26}
\end{tabular}
\\
\midrule
Degree-$q$ PTFs
&
$\mathcal N(0,I_d)$
&
$d^{\widetilde O(q^2/\varepsilon^2)}$
&
\cite{kane2011gaussian}
\\
\midrule
Monotone Boolean functions
&
$\operatorname{Unif}(\{\pm1\}^d)$
&
$d^{\widetilde O(\sqrt d/\varepsilon)}$
&
\cite{bshouty1996fourier}
\\
\midrule
\begin{tabular}{c}
Size-$s$, depth-$a$
$\mathsf{AC}^0$ circuits
\end{tabular}
&
$\operatorname{Unif}(\{\pm1\}^d)$
&
$d^{O((\log s)^{a-1}\log(1/\varepsilon))}$
&
\begin{tabular}{c}
\cite{linial1993constant}\\
\cite{tal2017tight}
\end{tabular}
\\
\midrule
Size-$s$ decision trees
&
$\operatorname{Unif}(\{\pm1\}^d)$
&
$d^{O(\log(s/\varepsilon))}$
&
\cite{o2014analysis}
\\
\bottomrule
\end{tabular*}
\caption{Runtime guarantees for $\eta+\varepsilon$ error under bounded contamination. References indicate the polynomial approximation results used to obtain these guarantees. See \cite[Appendix C]{stavropoulos2026computational} for a comparison to the best known lower bounds for agnostic learning.}
\label{tab:runtime-applications}
\end{table*}

\paragraph{Efficient Algorithms via ${\cal L}_1$ Approximation} Our main theorem on efficient end-to-end algorithms for learning with bounded contamination relies on the existence of low-degree ${\cal L}_2$ approximators. On the other hand, for agnostic learning, the weaker notion of ${\cal L}_1$ polynomial approximation -- where the approximation error is measured with respect to the absolute loss -- is known to suffice for efficient learnability \cite{kalai2008agnostically}. In fact, for the special case where the marginal distribution is the uniform over the Boolean hypercube, or standard Gaussian, it is known that ${\cal L}_1$ approximation characterizes the complexity of agnostic learning for all natural classes \cite{dachman2014approximate,diakonikolas2021optimality} in the statistical query (SQ) framework of \cite{kearns98sq}.

Here, we resolve an open question from \cite{KSTV25} on whether ${\cal L}_1$ approximation suffices for efficient learning with bounded contamination, for the case where the clean distribution is Gaussian or uniform on the Boolean hypercube. 

\begin{theorem}[Learning via ${\cal L}_1$ approximation; see \Cref{cor:l1-robust-learning} for the full version]
\label{thm:L1-guarantee}
Let $\mathcal D$ be the standard Gaussian distribution over $\R^d$ or the uniform distribution on $\{\pm1\}^d$ with a universal hypercontractivity constant $A$. Suppose that $\mathcal C$ admits degree-$\ell$ ${\cal L}_1$ polynomial approximators with absolute error at most $c_0\varepsilon$, for a sufficiently small universal constant $c_0>0$. Set $\ell' = C_0\ell^2/\varepsilon^4$ for some sufficiently large universal constant $C_0$. There is an algorithm that learns $\mathcal C$ on $\mathcal{D}$ with $\eta$-bounded contamination and error $\eta+\varepsilon$ and probability of failure $\delta$, using $\frac{\left(A(d+1)\log(2/\delta)\right)^{O(\ell')}}{\varepsilon^3}$ samples and $\mathrm{poly}(A,\frac{1}{\varepsilon},(d+1)^{\ell'},\log^{\ell'}(1/\delta)) $ time.
\end{theorem}

The above result does not hold in the full generality of \Cref{thm:main-guarantee}, which captures all hypercontractive distributions. Nevertheless, combined with the SQ lower bounds of \cite{dachman2014approximate,diakonikolas2021optimality} for agnostic learning, it yields the surprising conclusion that, under Gaussian and uniform marginals, efficient agnostic learning for any reasonable class implies efficient learning with bounded contamination, under the SQ model.\footnote{Note that this conclusion has a quantitative caveat: for agnostic learning, $\varepsilon$-${\cal L}_1$ approximation degree of $\ell$ implies runtime $\mathrm{poly}(1/\varepsilon)\cdot d^{O(\ell)}$, while our result on bounded contamination gives $d^{O(\ell^2/\varepsilon^4)}$.} En route, we also show that a polynomial bound on the $\varepsilon^2$-${\cal L}_1$ approximation degree implies a polynomial bound on the $O(\varepsilon)$-${\cal L}_2$ approximation degree for any class, under Gaussian or uniform marginals (see \Cref{lemma:l1-to-l2}).

\paragraph{Matching Agnostic Error Rates} The result of \Cref{thm:L1-guarantee} demonstrates that ${\cal L}_1$ approximation suffices for efficient learning with bounded contamination under two important distributions. However, the corresponding runtime bounds do not quantitatively match the associated bounds for agnostic learning, because the exponent scales polynomially with the approximation degree, rather than linearly. We show that a further specialization, namely, considering classes with bounded noise sensitivity under the Gaussian distribution, or the uniform on the Boolean hypercube, yields runtime bounds that nearly match the optimal bounds for agnostic learning. See \Cref{sec:noise-sensitivity} for further details.

Table~\ref{tab:runtime-applications} summarizes our runtime guarantees for several concept classes under different marginal distributions. For the results that hold under the Gaussian and the uniform on the Boolean hypercube, we use the rates provided by \Cref{sec:noise-sensitivity}. For the rest, we use \Cref{thm:main-guarantee}.

\subsection{Technical Overview}\label{sec:tech-overview}

Our goal is to output a randomized\footnote{See Appendix~\ref{app:derandomization} for a conversion to a deterministic hypothesis.} hypothesis $h$ such that
\begin{equation}
    \err_{\mathcal{D}}(h,c^\star) \le \eta + O(\varepsilon)\,,\label{equation:optimal-guarantee}
\end{equation}
where $c^\star$ is the ground truth labeling function of the clean distribution. The main challenge is that we do not have direct access to the function $c^\star$, due to contamination. Instead, we have access to the contaminated input samples $\Sinplabeled$, as well as \emph{unlabeled} samples from the clean distribution $\mathcal{D}$.

A first idea would be to find an empirical risk minimizer (ERM) for the input sample $\Sinplabeled$. Unfortunately, this leads to an error bound of $2\eta+O(\varepsilon)$. To see this, observe that the guarantee of the ERM is that the number of mistakes made on $\Sinplabeled$ is at most $\eta |\Sinplabeled|$. However, we have no guarantee about where these mistakes are made. In particular, it is possible that all the mistakes of the ERM are made on the set $\Sinplabeled\cap \Sclnlabeled$. On the other hand, we have no information about the adversarially removed clean examples (i.e., the points in $\Sclnlabeled\setminus\Sinplabeled$), and, therefore, cannot directly rule out the possibility that the ERM misclassifies all of those points. In total, the corresponding bound on the number of mistakes that the ERM makes on the set $\Sclnlabeled$ is $2\eta|\Sclnlabeled|$. Indeed, there are simple cases where any deterministic hypothesis (and, hence, ERM) is known to achieve error at least $2\eta$ \cite{BEK02}.

\subsubsection{The framework of Blanc for optimal error guarantees}
The recent work of Blanc \cite{Blanc26} proposed a different approach that achieves the optimal guarantee in a sample-efficient way, but with no computational guarantees. 

To explain the approach, observe that, since we do not know the true labeling function $c^\star$, it is natural to try to control $\err_D(h,c)$ simultaneously for every $c\in\C$. However, these objectives can be mutually incompatible: for example, if $\C$ contains both a concept $c$ and its negation $-c$, no randomized hypothesis can have error $o(1)$ with respect to both. 

The key observation of \cite{Blanc26} is that one should not require the same error guarantee for every candidate $c$. Rather, the error allowed with respect to $c$ should depend on how much corruption would have been necessary to produce the observed input sample had $c$ actually been the ground truth. This is captured by an appropriate correction term $\eta(c;\Sinplabeled,D)$ which serves as a \emph{corruption certificate}. Concretely, the approach is to minimize a loss function of the form
\begin{equation}
    L_{\Sinplabeled,D}(h)
    =
    \sup_{c\in\C}
    \Bigl(
        \err_D(h,c)
        -
        \eta(c;\Sinplabeled,D)
    \Bigr),
    \label{equation:blanc-objective}
\end{equation}
where $\eta(c;\Sinplabeled,D)$ is a distributional proxy for the fraction of points of the candidate $c$-labeled clean set
\[
    \Sclnlabeled^c
    =
    \{(x,c(x)):x\in\Sclnunlabeled\}
\]
that would need to be corrupted in order to obtain $\Sinplabeled$. Equivalently, it serves as a proxy for $d_{\mathsf{TV}}(\Sinplabeled,\Sclnlabeled^c)$ when the algorithm does not have access to $\Sclnunlabeled$ itself, but only to samples from $D$. Thus, minimizing \eqref{equation:blanc-objective} amounts to finding a hypothesis that, simultaneously for every candidate ground truth $c\in\C$, achieves error approximately no larger than the amount of corruption needed to make $c$ consistent with the observed data.

Hence, the question becomes whether there exists a sufficiently simple hypothesis $h$ that simultaneously achieves these candidate-dependent optimal error benchmarks. The bulk of the work of \cite{Blanc26} is devoted to showing that $\eta(c;\Sinplabeled,D)$ can be defined such that
\begin{equation}
    \eta(c;\Sinplabeled,D)
    \le
    d_{\mathsf{TV}}(\Sinplabeled,\Sclnlabeled^c)
    +
    O(\varepsilon)
    \qquad\text{for all }c\in\C,
    \label{equation:blanc-certificate}
\end{equation}
and there exists some randomized $h$ obtained by mixing hypotheses in $\operatorname{Maj}_k(\C)$, where $k=\mathrm{poly}(1/\varepsilon)$, such that
\begin{equation}
    L_{\Sinplabeled,D}(h)\le\varepsilon.
    \label{equation:blanc-small-loss}
\end{equation}
When both \eqref{equation:blanc-small-loss} and \eqref{equation:blanc-certificate} hold, applying them to the true concept $c^\star$ yields the desired error bound \eqref{equation:optimal-guarantee}.

\subsubsection{A connection to online optimization}
Although the framework of \cite{Blanc26} is sample-efficient for every class of bounded VC dimension, its direct implementation is highly inefficient. In particular, it searches over mixtures of hypotheses in $\operatorname{Maj}_k(\C)$ for an approximate minimizer of $L_{\Sinplabeled,D}$, and such a search is not known to be efficient even given access to an ERM oracle for $\C$. {This is due to the choice of the objective function $L_{\Sinplabeled,D}$.} Our main conceptual contribution is to replace this search with an online optimization procedure.

To motivate this connection, consider the minimax dual of \eqref{equation:blanc-objective}. After allowing randomized hypotheses, minimizing $L_{\Sinplabeled,D}$ is equivalent to maximizing, over distributions $\mu$ on $\C$, the objective
\[
    U_{\Sinplabeled,D}(\mu)
    =
    \inf_h
    \E_{c\sim\mu}
    \left[
        \err_D(h,c)-\eta(c;\Sinplabeled,D)
    \right].
\]
Moreover, the corruption certificate is itself defined through a supremum over a suitable class $\mathcal{M}\subseteq \{X\to[-1,1]\}$
of discrepancy witnesses:
\[
    \eta(c;\Sinplabeled,D)
    =
    \sup_{m\in\mathcal{M}}
    \eta(c,m;\Sinplabeled,D).
\]
Consequently, the dual objective can be written as
\begin{equation}
    U_{\Sinplabeled,D}(\mu)
    =
    \inf_h
    \E_{c\sim\mu}
    \left[
        \err_D(h,c)
        -
        \sup_{m\in\mathcal{M}}
        \eta(c,m;\Sinplabeled,D)
    \right].
    \label{equation:dual-objective}
\end{equation}

One first attempt would be to directly maximize \eqref{equation:dual-objective} by analyzing its first-order properties. The usual route to obtaining supergradients (e.g., via Danskin's theorem) would require exact solutions to the inner optimizations over $h$ and the discrepancy witnesses $m$. Instead, a structural result of \cite{Blanc26} gives, for any fixed $\mu$, a pair $(h,m)$ whose substituted objective value is at most $\varepsilon$:

\begin{lemma}[\cite{Blanc26}]
    \label{lemma:dual-small-loss}
    Let $\mu$ be any distribution over $\C$, and define
    \[
        h_\mu(x)
        =
        \sign\left(\E_{c\sim\mu}[c(x)]\right).
    \]
    Then, for $\bar\mu(x) = \E_{c\sim \mu}[c(x)]$, we have:
    \[
        \E_{c\sim\mu}
        \left[
            \err_D(h_\mu,c)
            -
            \eta(c,\bar\mu;\Sinplabeled,D)
        \right]
        \le \varepsilon.
    \]
\end{lemma}

The pair $(h_\mu,\bar\mu)$ need not even approximately solve the inner optimizations in Eq.~\eqref{equation:dual-objective}: the lemma only bounds its substituted objective value from above by $\varepsilon$, whereas the dual objective may be substantially smaller. Thus, the lemma does not directly yield a supergradient of $U_{\Sinplabeled,D}$.

Our key observation is that \Cref{lemma:dual-small-loss} nevertheless provides exactly the type of response needed for a low-regret online optimization argument to suffice for our purposes. At round $t$, the online algorithm chooses a distribution $\mu_t$ over $\C$. We then construct $h_t=h_{\mu_t}$ and a corresponding witness $m_t={\bar\mu_t}$, and define the reward function
\[
    U_t(\mu)
    =
    \E_{c\sim\mu}
    \left[
        \err_D(h_t,c)
        -
        \eta(c,\bar\mu_t;\Sinplabeled,D)
    \right].
\]
Although the reward function changes from round to round, \Cref{lemma:dual-small-loss} guarantees that the reward at the distribution chosen by the algorithm always satisfies
\begin{equation}
    U_t(\mu_t)\le\varepsilon.
    \label{equation:reward-at-each-step}
\end{equation}
Suppose that we can choose the distributions $\mu_1,\ldots,\mu_T$ using an online reward-maximization algorithm with sublinear regret. Then, for every distribution $\mu$ over $\C$,
\begin{equation}
    \sum_{t=1}^T U_t(\mu)
    \le
    \sum_{t=1}^T U_t(\mu_t)
    +
    \operatorname{Reg}_T,
    \qquad
    \operatorname{Reg}_T=o(T).
    \label{equation:sublinear-regret}
\end{equation}

The significance of this formulation is that the regret guarantee can be applied to the point mass on the unknown true concept $c^\star$, even though the algorithm never observes $c^\star$. Combining Eq.~\eqref{equation:reward-at-each-step} and Eq.~\eqref{equation:sublinear-regret} with $\mu=\delta_{c^\star}$ gives
\[
    \frac{1}{T}\sum_{t=1}^T
    \err_D(h_t,c^\star)
    \le
    \frac{1}{T}\sum_{t=1}^T
    \eta(c^\star,\bar\mu_t;\Sinplabeled,D)
    +
    \varepsilon
    +
    \frac{\operatorname{Reg}_T}{T}.
\]
By taking ${\cal M}$ to be the convex hull of $\C$, every $\bar\mu_t$ is a valid discrepancy witness, and hence, by Eq.~\eqref{equation:blanc-certificate},
\[
    \eta(c^\star,\bar\mu_t;\Sinplabeled,D)
    \le
    \eta(c^\star;\Sinplabeled,D)
    \le
    \eta+O(\varepsilon).
\]
Therefore, the randomized hypothesis $h$ that samples $\mathbf{t}\sim\operatorname{Unif}([T])$ and outputs $h_{\mathbf{t}}$ satisfies
\[
    \err_D(h,c^\star)
    \le
    \eta+O(\varepsilon)
    +
    \frac{\operatorname{Reg}_T}{T}.
\]
It remains to show that the online updates can be implemented efficiently.

\subsubsection{Efficient implementation with an ERM oracle}
The online framework discussed above requires an efficient reward maximization routine to obtain sublinear regret as in \eqref{equation:sublinear-regret}. We show that this can be achieved efficiently as long as we have access to an ERM oracle for the class $\C$.

In particular, by using appropriate uniform convergence tools, it suffices to focus on a finite-sample variant of the objective, i.e.:
\[
    \hat{U}_t(\mu)
    =
    \E_{c\sim\mu}
    \left[
        \widehat\err_{\Srefunlabeled}(h_t,c)
        -
        \eta(c,m_t;\Sinplabeled,\Srefunlabeled)
    \right],
\]
where $\Srefunlabeled$ is a set of independent reference samples drawn from $D$. The first crucial observation is that at each time step $t$, the objective function $\hat{U}_t(\mu)$ can be written in the following linearized form:
\[
    \hat{U}_t(\mu)
    =
    C_t+\langle g_t,\E_{c\sim\mu}[c]\rangle\,,
\]
for some scalar $C_t$ and vector $g_t\in\R^{|\Sinplabeled|+|\Srefunlabeled|}$, both independent of $\mu$, where $c$ is also viewed as a vector in $\{\pm1\}^{|\Sinplabeled|+|\Srefunlabeled|}$ taking the values $c(x)$ for $x\in\Sinpunlabeled$ and $x\in\Srefunlabeled$. A further simplification is to view the input $\mu$, which corresponds to some distribution over $\C$, as some vector in $[-1,1]^{|\Sinplabeled|+|\Srefunlabeled|}$ taking the values $\E_{c\sim\mu}[c(x)]$ for $x\in\Sinpunlabeled$ and $x\in\Srefunlabeled$. With this abuse of notation we have:
\[
    \hat{U}_t(\mu)=C_t+\langle g_t,\mu\rangle\,.
\]

Since these rewards are affine in $\mu$, we are then able to use an online Frank-Wolfe procedure to achieve regret $O(T^{3/4})$. Concretely, the online Frank-Wolfe procedure is a projection-free online optimization algorithm with sublinear regret, whose only algorithmic requirement is an oracle that returns an approximate maximizer $\nu_t$ of the following program:\footnote{Here $F_t$ is a regularized sum of the rewards from preceding rounds; we omit its precise form from this high-level overview.}
\begin{align*}
    \max_\nu\;&\langle\nabla F_t(\mu_t),\nu\rangle\\
    \text{s.t. }&\nu\text{ lies in the convex hull of }\C.
\end{align*}
The next iterate $\mu_{t+1}$ is then formed as a convex combination of $\nu_t$ and the current iterate $\mu_t$, and as such it remains in the convex hull of $\C$. We show that $\nu_t$ can be computed efficiently assuming access to an ERM oracle for $\C$, as it suffices to find a concept $c_t$ in $\C$ with near-maximal correlation with appropriately chosen weighted labels $(z_x)_x$ over the concatenated sample $\Sinplabeled\cup\Srefunlabeled$. {Due to the linearity of the objective, $\nu_t$ can always be chosen to be a single point-mass $\delta_{c_t}$ for some $c_t\in\C$.} In particular, we label each datapoint by the sign of the corresponding coordinate of $\nabla F_t(\mu_t)$ and assign it a weight equal to the absolute value of that coordinate. This weighted ERM problem can be reduced to ordinary ERM by sampling from the resulting weighted empirical distribution (see \Cref{fact:weighted-erm}).

We note that ERM oracles have been used to obtain Frank-Wolfe-type updates in prior work in other contexts \cite{GoelGollakotaKlivans20,gopalan_et_al:LIPIcs.ITCS.2023.60}, as well as in the work of \cite{Blanc26}, but here we use it in an online variant of Frank-Wolfe.

\subsubsection{Efficient implementation using approximating polynomials}\label{sec:poly-overview}
In order to obtain end-to-end efficient algorithms that do not require access to an ERM oracle, we combine several technical tools, including Zinkevich's online projected gradient ascent, approximating polynomials, and a new iterative filtering procedure.

\paragraph{A first attempt} Recent prior work obtained efficient algorithms for learning with bounded contamination up to error $2\eta+\varepsilon$ for classes that can be approximated by low-degree polynomials under the clean distribution \cite{klivans-learning-constant-depth-in-malicious-noise-models,KSTV25}. Roughly, the idea behind those results is to filter the input dataset $\Sinpunlabeled$ to obtain a subset $\Sfiltunlabeled$ that preserves the relevant low-degree properties of the clean distribution, without removing too many clean points from $\Sinpunlabeled$. After filtering, the ground truth is approximated by some low-degree polynomial under the uniform distribution over $\Sfiltunlabeled$, and therefore, the classical polynomial regression algorithm of \cite{kalai2008agnostically} yields a near-optimal classifier on $\Sfiltlabeled$, which can be shown to achieve error $2\eta+\varepsilon$ under the clean distribution.

Concretely, \cite{KSTV25} developed a filtering algorithm called Iterative Polynomial Filtering (IPF) which provides the following guarantee for all low-degree polynomials $p$:
\begin{equation}
        \frac{1}{|\Sinplabeled|} \sum_{x\in \Sfiltunlabeled}p(x) \lesssim \E_{x\sim D}[|p(x)|] \label{equation:ipf-guarantee}
\end{equation}
One first idea would be to apply the IPF procedure once in the beginning of our algorithm, and then substitute the objective $\hat{U}_t(\mu)$ with a surrogate $\hat{U}_t^{\mathrm{poly}}(q)$, over low-degree polynomials $q$ of the following form:
\begin{equation}
\label{eq:loss-def-intro}
\hat{U}_t^{\mathrm{poly}}(q)
:=
\widehat{\err}_{S_{\mathrm{ref}}}(h_t,q)
-
\eta^{\mathrm{poly}}
\left(
q,m_t;
\bar S_{\mathrm{filt}},S_{\mathrm{ref}}
\right),
\end{equation}
where the corruption certificate $\eta^{\mathrm{poly}}$ uses witnesses $\mathcal{M}^{\mathrm{poly}}$ that are themselves polynomials. There are three properties of the surrogate objective which, if satisfied, would be jointly sufficient for our purposes:
\begin{enumerate}
    \item \emph{(Linearization)} If the reward functions $\hat{U}_t^{\mathrm{poly}}(q)$ are affine in the coefficients of the polynomial $q$ (we refer to this as linearization), then, we may apply a black-box online convex optimization result by \cite{Zin03} to obtain regret $O(\sqrt{T})$.
    \item \emph{(Controllability)} A version of \Cref{lemma:dual-small-loss} holds for $\hat{U}_t^{\mathrm{poly}}(q)$, i.e., for each polynomial $q_t$ of the trajectory of the online learner, $h_t$ and $m_t\in{\mathcal{M}}^\mathrm{poly}$ can be chosen so that $\hat{U}_t^{\mathrm{poly}}(q_t) \le \varepsilon$.
    \item \emph{(Validity)} For the polynomial approximator $q^\star$ of $c^\star$ under $D$, we have: 
    \begin{equation}\eta^{\mathrm{poly}}(q^\star,m_t;\bar S_{\mathrm{filt}},S_{\mathrm{ref}}) \le \eta+O(\varepsilon), \text{ for all }t\in [T] \label{equation:valid-certificate-poly}\end{equation}
\end{enumerate}
The first two properties can be satisfied with appropriate choices of $\eta^{\mathrm{poly}}$ and $\mathcal{M}^{\mathrm{poly}}$. However, for the third property, it is important to ensure that $q^\star$ approximates $c^\star$ accurately under the uniform distribution over $\Sfiltunlabeled$. It turns out that if $q^\star$ is only guaranteed to approximate $c^\star$ in squared loss under $D$, then the guarantee of IPF \eqref{equation:ipf-guarantee} does not seem to suffice to obtain the sharp guarantee of \eqref{equation:valid-certificate-poly}. Assuming stronger notions of approximation, such as sandwiching polynomials \cite{Baz09} would indeed render the IPF guarantee sufficient. However, we only assume the existence of approximating polynomials.

\paragraph{Weighted Iterative Polynomial Filtering} In order to circumvent the obstruction described above, we follow a different approach, where we run a separate polynomial filtering procedure for each time step $t$ of the online learner. This allows us to specialize our filtering procedure to the relevant per-step quantities in \eqref{equation:valid-certificate-poly}, achieving the sharp bound for the corruption certificate. 

In particular, we introduce a new iterative filtering procedure called Weighted Iterative Polynomial Filtering (WIPF) (see \Cref{fig:weighted-ipf}), which, for given $v_t:X\times\{\pm 1\}\to [-1,1]$, achieves the following guarantee for any low-degree polynomial $p$ upon completion:
\begin{equation}
        \frac{1}{|\Sinplabeled|} \sum_{(x,y)\in \Sfiltlabeled} v_t(x,y)p(x) \lesssim \E_{x\sim D}[|{p(x)}|] \label{equation:weighted-ipf-guarantee}
\end{equation}
We choose $v_t$ depending on the distinguisher $m_t$ so as to control the quantities that appear in an appropriate decomposition of  $\eta^{\mathrm{poly}}(q^\star,m_t;\bar S_{\mathrm{filt}},S_{\mathrm{ref}})$, obtaining the sharp bound of $\eta+O(\varepsilon)$.

The reason we are able to use one filtering procedure at each time step $t$ is that iterative filtering satisfies a composable invariant: apart from a small initial loss of clean points, the number of clean points removed is at most the number of adversarial points removed.

\paragraph{Beyond ${\cal L}_2$ approximation} 
Our main result on end-to-end efficient algorithms works under ${\cal L}_2$ polynomial approximation. To obtain a result that works under ${\cal L}_1$ approximation, in \Cref{sec:l1-approximation}, we observe that if $c^\star$ admits a low-degree ${\cal L}_1$ approximator $q_1^\star$, then we have:
\[
\E[(\sign(q_1^\star(x)) - c^\star(x))^2] = 4 \Pr[\sign(q_1^\star(x)) \neq c^\star(x)] \le 4\E[|q_1^\star(x) - c^\star(x)|] \le 4\varepsilon
\]
Therefore, it suffices to prove the existence of a ${\cal L}_2$ approximator for the function $\sign(q_1^\star(x))$, which is always a polynomial threshold function (PTF), as this would imply the existence of ${\cal L}_2$ polynomial approximators for $c^\star$ as well. {In fact, we observe that $\sign(q_1^\star(x))$ is an $\varepsilon$-approximator for $c^\star$ even with respect to the zero-one loss, leading to improved rates.}

The ${\cal L}_2$ $\varepsilon$-approximation degree of a degree-$\ell_1$ PTF is known to be $\mathrm{poly}(\ell_1/\varepsilon)$ under the Gaussian distribution \cite{kane2011gaussian}, as well as under the uniform over the Boolean hypercube \cite{OpenAIPTF26}.\footnote{The result on the Boolean hypercube appears in OpenAI's 2026 bulk release of mathematical papers in the context of resolving the Gotsman-Linial conjecture. The authors have not verified the correctness of OpenAI's claimed result.}

Moreover, in \Cref{sec:noise-sensitivity}, we observe that under noise sensitivity assumptions one can show the existence of ${\cal L}_1$ approximators with an additional structural property that suffices for the purposes of our algorithm, yielding bounds that nearly match those achievable in agnostic learning for several classes under Gaussian or uniform marginals.

\subsection{Related Work}

\paragraph{Learning with Contamination} The problem of supervised learning with contamination dates back to the early days of learning theory \cite{Valiant85,KearnsL93}. The precise model of bounded contamination we study in this work was initially defined by \cite{BEK02}. There has since been a line of research on the computational complexity of learning in this model \cite{klivans2009learning,awasthi2017power,diakonikolas-learning-geometric-concepts-with-nasty-noise,goel2024tolerant, klivans-learning-constant-depth-in-malicious-noise-models}, but before our work, there was no known efficient learning algorithm with error $\eta+\varepsilon$ even for the basic class of halfspaces with respect to the Gaussian distribution. Instead, all the previous efficient algorithms achieved an error guarantee of at least $2\eta$ \cite{KSTV25} (or worse).

From an information-theoretic perspective, the work of \cite{BEK02} provided a lower bound demonstrating that there are cases where any learner that outputs a deterministic hypothesis must incur error $2\eta$. The work of \cite{10.1145/324133.324221}, which worked under a slightly relaxed noise model called \emph{malicious noise}, taken together with a result from \cite{blanc2026nasty}, which relates the malicious noise model to bounded contamination, provided an improved error bound of $3\eta/2+\varepsilon$ for VC classes, by allowing randomized hypotheses. The recent work of \cite{Blanc26} showed that the $3\eta/2+\varepsilon$ error bound can be achieved with sample complexity linear in the VC dimension of the learned class, and gave an efficient reduction to agnostic ERM. Moreover, \cite{Blanc26} achieved error $\eta+\varepsilon$ with a similar sample complexity bound, but with an algorithm that is not efficient, even provided access to an ERM oracle. Our work is the first to establish an efficient reduction to ERM that achieves error $\eta+\varepsilon$.

\paragraph{Online Convex Optimization} Our main conceptual contribution is a new connection between the (offline) problem of learning with contamination, and online convex optimization. For our ERM reduction, we use an online version of Frank-Wolfe algorithm whose classical analysis can be found in \cite{HazanOCO}. Note that several previous works have used the ERM oracle to simulate the gradient updates of the offline Frank-Wolfe algorithm or similar procedures in the context of statistical query lower bounds \cite{GoelGollakotaKlivans20}, indistinguishability and omniprediction \cite{gopalan_et_al:LIPIcs.ITCS.2023.60}, as well as robust learning \cite{Blanc26}. In particular, \cite{Blanc26} uses an offline Frank-Wolfe approach to obtain the suboptimal error bound of $3\eta/2+\varepsilon$. For our end-to-end efficient algorithms, we use a slightly different online convex optimization algorithm by \cite{Zin03} with better regret guarantees.

\paragraph{Iterative Filtering} To achieve end-to-end efficient algorithms, we use iterative filtering within the online optimization loop. Such filtering algorithms have been thoroughly explored in the literature of robust learning and learning with distribution shift \cite{klivans2009learning,diakonikolas-learning-geometric-concepts-with-nasty-noise,diakonikolas2019robust,diakonikolas2019sever,goel2024tolerant,klivans-learning-constant-depth-in-malicious-noise-models,KSTV25,chandrasekaran2026iterative}. We use and extend the filtering guarantees of \cite{KSTV25} to control polynomial averages. Apart from a small initial loss of clean points, filtering removes at least as many corrupted points as clean points.

\paragraph{Polynomial Approximation} We use low-degree ${\cal L}_2$ polynomial approximators as an analytic tool to obtain our runtime upper bounds. Polynomial approximation has a long history in learning theory, where it is closely connected to low-degree Fourier and Hermite concentration \cite{linial1993constant,bshouty1996fourier,klivans2008learning}. This connection has led to approximation results for a broad range of function classes, including $\mathsf{AC}^0$ circuits, intersections of halfspaces, monotone functions, convex sets, and polynomial threshold functions under distributions such as the standard Gaussian or the uniform distribution on the hypercube \cite{linial1993constant,bshouty1996fourier,klivans2008learning,kane2011gaussian,kane2013correct,kane2014average,tal2017tight}.

\paragraph{Agnostic Learning} Another related notion of robust learning is agnostic learning \cite{Haussler:92,kearns1992toward}. In this model, the noise can only affect the labels, and the error is measured on the distribution of the input samples rather than the underlying noiseless distribution. The work of \cite{kalai2008agnostically} showed that ${\cal L}_1$ polynomial approximators suffice for efficient learnability with near-optimal error in this model. In fact, it is known that efficient learnability with near-optimal error is characterized by the degree of ${\cal L}_1$ polynomial approximation, through matching statistical query lower bounds under the uniform distribution over the Boolean hypercube \cite{dachman2014approximate}, as well as the standard Gaussian \cite{diakonikolas2021optimality}.

\section*{AI Disclosure}

The main ideas underlying this work were human-generated. These include the use of online convex optimization, as well as the connection to iterative polynomial filtering. An earlier version of the paper, which included the ERM reduction and an efficient end-to-end algorithm under a stronger assumption called sandwiching approximation, was prepared entirely without AI assistance.

For the current version, we observed that a reweighted version of iterative filtering can be applied within the online optimization loop leading to potential improvements and, subsequently, queried GPT-6 Astra (OpenAI) to complete the proof. This allowed us to relax the sandwiching requirement to the notion of ${\cal L}_2$ polynomial approximation. We have fully verified the correctness of the proof and incorporated it into the current version. 

Finally, GPT-6 Astra (OpenAI) proposed a simple approach to obtain ${\cal L}_2$ approximators from ${\cal L}_1$ approximators, under any distribution that admits low-degree ${\cal L}_2$ approximators for polynomial threshold functions.

\section{Preliminaries}

\paragraph{Notation.} 

For an integer $n\ge1$, write $[n]:=\{1,\ldots,n\}$. For a class of functions $\mathcal F$, $\operatorname{conv}(\mathcal F)$ denotes the set of all finite convex combinations of functions in $\mathcal F$. For $t\in[-1,1]$, $\operatorname{Rad}(t)$ denotes the distribution on $\{\pm1\}$ assigning probabilities $(1+t)/2$ and $(1-t)/2$ to $1$ and $-1$, respectively. We use $\widetilde O(\cdot)$ to suppress polylogarithmic factors.

Let $x^{\otimes\ell}\in\R^N$ denote the vector\footnote{The notation $x^{\otimes\ell}$ often denotes the homogeneous tensor power of degree exactly $\ell$; throughout this paper, it denotes the vector of all monomials of degree at most $\ell$.} of all monomials in $x$ of degree at most $\ell$, where $N=\binom{d+\ell}{\ell}$. Every degree-$\ell$ polynomial can consequently be written as $q_w(x)=\langle w,x^{\otimes\ell}\rangle$ with the coefficient vector $w\in\R^N$. Define the degree-$\ell$ polynomial-threshold class
\begin{equation}\label{eq:ptf-class}
\mathcal H_0 := \left\{x\mapsto\sign\left(\langle w,x^{\otimes\ell}\rangle\right): w\in\R^N \right\}.
\end{equation}
For $t\in[-1,1]$ and an integer $k\ge1$, the expected majority vote of $k$ independent draws from $\operatorname{Rad}(t)$ is
\begin{equation}\label{eq:maj-map}
\operatorname{maj}_k(t)
:=
\E_{z_1,\ldots,z_k\sim\operatorname{Rad}(t)}
\left[
\sign\left(\frac1k\sum_{i=1}^k z_i\right)
\right].
\end{equation}
Define the class of $k$-wise majorities of concepts in $\mathcal C$ by
\begin{equation}\label{eq:maj-class}
\operatorname{Maj}_k(\mathcal C)
:=
\left\{
x\mapsto
\sign\left(
\frac1k\sum_{i=1}^k c_i(x)
\right)
:
c_1,\ldots,c_k\in\mathcal C
\right\}.
\end{equation}

\paragraph{Sample conventions.} Each sample entry is treated as a distinct instance, even when two entries represent the same feature vector or labeled example. Set operations on samples respect these identities. If $S$ denotes the unlabeled version of a labeled sample $\bar S$, removal of the labels does not change the instance identities. An adversarial substitution removes one clean instance and inserts a new one.

\begin{definition}[Randomized $0$-$1$ error]
\label{def:linear-error}
Let $\bar h:X\to[-1,1]$ and $q:X\to\R$. For a distribution $\mathcal{D}$ over $X$ and a reference sample
$S_{\mathrm{ref}}$, define
$$
\err_{\mathcal{D}}(\bar h,q)
:=
\E_{x \sim \mathcal{D}}
\left[
\frac{1-\bar h(x)q(x)}{2}
\right]\qquad\text{and}
\qquad
\widehat{\err}_{S_{\mathrm{ref}}}(\bar h,q)
:=
\frac{1}{|S_{\mathrm{ref}}|}
\sum_{x\in S_{\mathrm{ref}}}
\frac{1-\bar h(x)q(x)}{2}.
$$
\end{definition}
When $q$ is not $[-1,1]$-valued, these quantities denote the linear extension of randomized $0$-$1$ error.

\section{Robust learning with access to an ERM oracle}\label{sec:learning-with-erm}

In this section, we present an algorithm that produces a hypothesis that achieves error at most $\eta+\varepsilon$ using a polynomial number of calls to an ERM oracle for $\mathcal C$. This resolves an open problem stated in \cite{Blanc26}.  For a labeled dataset, the oracle returns a concept in $\mathcal C$ that maximizes empirical agreement with its labels.

\begin{definition}[ERM oracle]
An ERM oracle for a concept class $\mathcal C$ takes as input a dataset $\bar S=\{(x_1,y_1),\ldots,(x_n,y_n)\}$ and returns any concept $c\in\mathcal C$ that maximizes
$
\sum_{(x,y) \in \bar S} c(x)y.
$
\end{definition}

We use the following reduction, which implements a weighted ERM oracle with one call to an ordinary ERM oracle.

\begin{fact}[Weighted ERM]
\label{fact:weighted-erm}
For any $\varepsilon,\delta>0$, there is an algorithm that takes as input a dataset $\bar S=\{(x_1,y_1),\ldots,(x_n,y_n)\}$ and a weight $w(x,y)\in[-b,b]$ for each $(x,y)\in\bar S$ such that it runs in time $\operatorname{poly}\left(|\bar S|, b,1/\varepsilon,\log(1/\delta)\right)$,
uses one call to an ERM oracle for $\mathcal C$, and, with probability at least $1-\delta$, outputs a concept $c\in\mathcal C$ such that
$$
\frac{1}{|\bar S|}\sum_{(x,y) \in \bar S} w(x,y) c(x)y
\ge
\max_{c'\in\mathcal C}
\left(
\frac{1}{|\bar S|}\sum_{(x,y) \in \bar S} w(x,y) c'(x)y
\right)
-
\varepsilon.
$$

The algorithm generates a dataset consisting of
$O\left((b^2/\varepsilon^2)(|\bar S|+\log(1/\delta))\right)$ independent samples from $\bar S$, where $(x,\sign(w(x,y))y)$ is sampled with probability $\frac{|w(x,y)|}{\sum_{(x',y') \in \bar S}|w(x',y')|}$. It then runs unweighted ERM on this synthetic dataset.
\end{fact}

Later in this section, we show that the Frank--Wolfe update corresponds to one call to a weighted ERM oracle. We note that the Frank--Wolfe update serves as a substitute for the projected gradient update in the polynomial framework of Section~\ref{sec:learning-with-polynomials}.

Throughout this section, for any function $f:X\to\R$, we use $f[S]$ as shorthand for the vector of its values on $\bar S_{\mathrm{inp}}$ and $S_{\mathrm{ref}}$:
$$
f[S]
:=
\left(
\bigl(f(x)\bigr)_{(x,y)\in\bar S_{\mathrm{inp}}},
\bigl(f(x)\bigr)_{x\in S_{\mathrm{ref}}}
\right)
\in
\R^{|\bar S_{\mathrm{inp}}|+|S_{\mathrm{ref}}|}.
$$
For any vectors $f[S],g[S]\in\R^{|\bar S_{\mathrm{inp}}|+|S_{\mathrm{ref}}|}$, we equip this with the Euclidean inner product and norm
$$
\begin{aligned}
&\langle f[S],g[S]\rangle
:=
\sum_{(x,y)\in\bar S_{\mathrm{inp}}}
f(x)g(x)
+
\sum_{x\in S_{\mathrm{ref}}}
f(x)g(x),
\\
&\|f[S]\|_2
:=
\sqrt{\langle f[S],f[S]\rangle}.
\end{aligned}
$$

\subsection{Robust learning algorithm via online Frank-Wolfe}\label{sec:ofw-algorithm}

We first define a corruption certificate $\eta(c,\bar\mu;\bar S_{\mathrm{inp}},S_{\mathrm{ref}})$ that represents a discrepancy between the labeled samples $(x,y)\sim\bar S_{\mathrm{inp}}$ from the corrupted distribution and the unlabeled samples $x\sim S_{\mathrm{ref}}$ from the clean distribution $\mathcal D$.

\begin{definition}[Corruption certificate]
\label{def:erm-empirical-certificate}
For a concept $c:X\to\{\pm1\}$ belonging to $\mathcal C$
and a function $\bar\mu:X\to[-1,1]$ belonging to
$\operatorname{conv}(\mathcal C)$, define
$$
\eta(c,\bar\mu;\bar S_{\mathrm{inp}},S_{\mathrm{ref}})
:=
\Pr_{(x,y)\sim\bar S_{\mathrm{inp}}}
[c(x)\neq y\ \text{or}\ \operatorname{Rad}(f_{c,\bar\mu}(x))=1]
-
\Pr_{x\sim S_{\mathrm{ref}}}
[\operatorname{Rad}(f_{c,\bar\mu}(x))=1]
$$
where the detection function is $f_{c,\bar\mu}(x):=\frac{c(x)\bar\mu(x)-1}{1+|\bar\mu(x)|}$.
\end{definition}

Intuitively, it estimates how many corruptions the adversary would need to turn a clean sample $\bar S_{\mathrm{cln}}^c = \{(x,c(x)):x\in\Sclnunlabeled\}$ labeled by some concept $c \in \mathcal{C}$ into the observed corrupted sample $\bar S_{\mathrm{inp}}$:
$$
\eta(c,\bar\mu;\bar S_{\mathrm{inp}},S_{\mathrm{ref}})
=
\Pr_{(x,y)\sim\bar S_{\mathrm{inp}}}
[c(x)\neq y\ \text{or}\ \operatorname{Rad}(f_{c,\bar\mu}(x))=1]
-
\Pr_{x\sim S_{\mathrm{ref}}}
[c(x)\neq c(x)\ \text{or}\ \operatorname{Rad}(f_{c,\bar\mu}(x))=1],
$$
where $c(x)\neq y$ accounts for label disagreement and detector function $f_{c,\bar\mu}$ captures covariate shift.\footnote{An equivalent form, convenient for the proofs, is given in Appendix~\ref{app:proofs-of-erm}.}

An important part of \cite{Blanc26} is devoted to showing that these choices of $\eta(c,\bar\mu;\bar S_{\mathrm{inp}},S_{\mathrm{ref}})$ and $f_{c,\bar\mu}$ establish the existence of a randomized hypothesis $h$ supported on $\operatorname{Maj}_k(\mathcal{C})$ with error at most $\eta+\varepsilon$. In particular, it satisfies two important properties:

\begin{itemize}
    \item \emph{Controllability},
    \cite[Lemma~8.2]{Blanc26}.
    For $f_{c,\bar\mu}(x)=\frac{c(x)\bar\mu(x)-1}{1+|\bar\mu(x)|}$ and $h_\mu(x):=\sign(\bar\mu(x))$,
    $$
    \E_{c\sim\mu}\bigl[\err_{\mathcal D}(h_\mu,c)-\eta(c, \bar\mu;\bar S_{\mathrm{inp}},\mathcal{D})\bigr]\le\varepsilon.
    $$

    \item \emph{Validity},
    \cite[Claim~8.1]{Blanc26}.
    For sufficiently large samples, the true concept $c^\star$ satisfies, with high probability,
    $$
    \sup_{\bar\mu\in \operatorname{conv}(\mathcal C)}\eta(c^\star,\bar\mu;\bar S_{\mathrm{inp}},\mathcal D)\le\eta+\varepsilon.
    $$
\end{itemize}

These properties motivate an objective that measures how much the error of a hypothesis can exceed the corruption level certified for a candidate concept. For $\bar h:X\to[-1,1]$,
we define
\begin{equation}\label{eq:optimization-objective}
L(\bar h)
:=
\sup_{c\in\mathcal C}
\left(
\widehat{\err}_{S_{\mathrm{ref}}}(\bar h,c)
-
\sup_{\bar\mu\in\operatorname{conv}(\mathcal C)}
\eta
\left(
c,\bar\mu;
\bar S_{\mathrm{inp}},S_{\mathrm{ref}}
\right)
\right).
\end{equation}
To control this objective, the algorithm constructs a sequence $(\bar\mu_t)_{t=1}^T$ in $\operatorname{conv}(\mathcal C)$. As in the information-theoretic construction of \cite{Blanc26}, $\bar\mu_t$ determines the hypothesis $h_t$ through the $k$-wise majority vote $h_t(x):=\operatorname{maj}_k\bigl(\bar\mu_t(x)\bigr)$, where $\operatorname{maj}_k(\cdot)$ is defined in~Eq.~\eqref{eq:maj-map}. With $h_t$ and $\bar\mu_t$ fixed, the iteration objective for each
concept $c\in\mathcal C$ is
\begin{equation}
\label{eq:erm-loss-def}
U_t(c)
:=
\widehat{\err}_{S_{\mathrm{ref}}}(h_t,c)
-
\eta
\left(
c,\bar\mu_t;
\bar S_{\mathrm{inp}},S_{\mathrm{ref}}
\right).
\end{equation}
Let $\bar h_T:=T^{-1}\sum_{t=1}^T h_t$. Since each $\bar\mu_t$ is a feasible choice in the supremum that defines $L(\bar h_T)$,
\begin{align*}
L(\bar h_T)
&=
\sup_{c\in\mathcal C}
\left(
\frac1T\sum_{t=1}^T
\widehat{\err}_{S_{\mathrm{ref}}}(h_t,c)
-
\sup_{\bar\mu\in\operatorname{conv}(\mathcal C)}
\eta
\left(
c,\bar\mu;
\bar S_{\mathrm{inp}},S_{\mathrm{ref}}
\right)
\right)
\nonumber\\
&\le
\sup_{c\in\mathcal C}
\frac1T\sum_{t=1}^T
\left(
\widehat{\err}_{S_{\mathrm{ref}}}(h_t,c)
-
\eta
\left(
c,\bar\mu_t;
\bar S_{\mathrm{inp}},S_{\mathrm{ref}}
\right)
\right)
\nonumber\\
&=
\sup_{c\in\mathcal C}
\frac1T\sum_{t=1}^T U_t(c).
\end{align*}
Thus, it suffices to bound $T^{-1}\sum_{t=1}^T U_t(c)$ simultaneously for every $c\in\mathcal C$. Our main ingredient is online optimization \cite{HazanOCO}.
The iteration objective $U_t(c)$ changes with $t$ through the current iterate $\bar\mu_t$ and the corresponding hypothesis $h_t$. The algorithm applies online Frank--Wolfe \cite{HazanKale12} to the sequence $(U_t)_{t=1}^T$ and outputs $\bar h_T$. We bound the average objective through its no-regret guarantee.

\begin{algorithm}[H]
\caption{Algorithm for $\eta+\varepsilon$ learning with an ERM oracle.}
\label{fig:erm-oracle-algorithm}
\begin{algorithmic}[1]
\Procedure{Online-Frank-Wolfe}
{$\bar S_{\mathrm{inp}},S_{\mathrm{ref}},\varepsilon,\delta,\mathrm{ERM\ oracle}$}

\State \textbf{Input:} A sample $\bar S_{\mathrm{inp}}$, a sample $S_{\mathrm{ref}}$, accuracy $\varepsilon$, failure probability $\delta$, ERM oracle for $\mathcal C$.

\State \label{line:erm-initialization}Initialize $\bar\mu_1=c$ for an arbitrary $c\in\mathcal C$ and $F_1(\bar\mu[S]):=
-\frac{\|\bar\mu[\bar S_{\mathrm{inp}}]-\bar\mu_1[\bar S_{\mathrm{inp}}]\|_2^2}{2|\bar S_{\mathrm{inp}}|}
-\frac{\|\bar\mu[S_{\mathrm{ref}}]-\bar\mu_1[S_{\mathrm{ref}}]\|_2^2}{2|S_{\mathrm{ref}}|}$.

\State Set
$T:=\left\lceil(300/\varepsilon)^4\right\rceil$ and
$\gamma_t:=\min(1,2/\sqrt t)$.

\For{$t=1,\ldots,T$}
    \State Set $h_t(x) = \operatorname{maj}_{\left\lceil 64/\varepsilon^2\right\rceil}\bigl(\bar\mu_t(x)\bigr)$ from Eq.~\eqref{eq:maj-map}.

    \State \parbox[t]{\dimexpr\linewidth-\algorithmicindent-\algorithmicindent\relax}{%
    Use one weighted ERM call (\Cref{fact:weighted-erm}) on
    $\bar S_{\mathrm{inp}}\cup\{(x,1):x\in S_{\mathrm{ref}}\}$,
    with the weights $w_t$ from Eq.~\eqref{eq:erm-weights},
    accuracy $\gamma_t/25$, and failure probability $\delta/T$,
    to find $c_t\in\mathcal C$ such that
    }
    $$
    \Bigl\langle
    \nabla F_t(\bar\mu_t[S]), c_t[S]
    \Bigr\rangle
    \ge
    \max_{v \in K}
    \Bigl\langle
    \nabla F_t(\bar\mu_t[S]),v
    \Bigr\rangle
    -
    \frac{\gamma_t}{25}.
    $$

    \State Form $g_t\in \R^{|\Sinplabeled|+|\Srefunlabeled|}$ as in Eq.~\eqref{eq:erm-grad} and update, for every $\bar\mu\in\operatorname{conv}(\mathcal C)$,
    $$
    F_{t+1}(\bar\mu[S]):=F_{t}(\bar\mu[S])+\alpha\langle g_{t},\bar\mu[S]\rangle
    \quad\quad\text{where}\quad
    \alpha:=\frac1{2T^{3/4}}.
    $$

    \State Update
    $$
    \bar\mu_{t+1}
    :=
    (1-\gamma_t)\bar\mu_t+\gamma_t c_t.
    $$
\EndFor

\State \Return $h(x)\sim\operatorname{Rad}(\bar h_T(x))$ with
$\bar h_T(x):=T^{-1}\sum_{t=1}^T h_t(x)$.

\EndProcedure
\end{algorithmic}
\end{algorithm}

Every quantity in Algorithm~\ref{fig:erm-oracle-algorithm}, in particular the iterates $\bar\mu_t$ and $h_t$, is computed from the corrupted sample $\bar S_{\mathrm{inp}}$ and the reference sample $S_{\mathrm{ref}}$, using the ERM oracle. In contrast, the clean sample $\bar S_{\mathrm{cln}}$, the target concept $c^\star$, and the corruption rate $\eta$ are used only in the analysis and are unavailable to the learner.

The initialization of $F_1(\bar\mu[S])$ in line~\ref{line:erm-initialization} of Algorithm~\ref{fig:erm-oracle-algorithm} introduces a quadratic regularizer. Although each iteration objective $U_t$ is affine, this regularizer makes the objectives $F_t$ strongly concave, as required by the proof of \Cref{lemma:erm-online-frank-wolfe}.

\subsection{No-regret bound}\label{sec:erm-no-regret}

In this subsection, we verify the assumptions of \Cref{lemma:erm-online-frank-wolfe}, adapted from \cite{HazanOCO} (see Appendix~\ref{app:D} for a self-contained proof of this theorem). In addition, we prove in \Cref{lemma:erm-linear-oracle} that each Frank--Wolfe update can be implemented with one call to a weighted ERM oracle.

\begin{theorem}[Online Frank--Wolfe, adapted from {\cite[Theorem~7.3]{HazanOCO}}]
\label{lemma:erm-online-frank-wolfe}
Let $K$ be a compact convex set with diameter at most $D_2$, where $D_2>0$. Let $T\in\mathbb N$, and let $(g_t)_{t=1}^T$ be any sequence satisfying $\|g_t\|_2\le G_2$ for every $t\in[T]$ and some $G_2 > 0$. Fix $u_1 \in K$, set $\alpha:=D_2/(2G_2T^{3/4})$ and $\gamma_t:=\min\left(1, 2/\sqrt t\right)$. For every $t \in [T]$, define $F_t(u):= \alpha \sum_{r=1}^{t-1}\langle g_r,u\rangle - \|u-u_1\|_2^2$, choose $v_t\in K$ such that
\begin{equation}\label{eq:approx-max}
\left\langle
\nabla F_t(u_t),v_t
\right\rangle
\ge
\max_{v\in K}
\left\langle
\nabla F_t(u_t),v
\right\rangle
-
\frac{D_2^2\gamma_t}{100},
\end{equation}
and update
$$
u_{t+1}
:=
(1-\gamma_t)u_t+\gamma_t v_t.
$$
Then, for every $u\in K$, the following no-regret bound holds:
\begin{equation}
\label{eq:erm-frank-wolfe-regret}
\sum_{t=1}^T
\langle u,g_t\rangle
\le
\sum_{t=1}^T
\langle u_t,g_t\rangle
+
9D_2G_2T^{3/4}.
\end{equation}
\end{theorem}
We define $K$ to be the convex hull of the vectors $c[S]$:
\begin{equation}\label{eq:erm-convex-set}
K
:=
\operatorname{conv}\{c[S]:c\in\mathcal C\}
\subseteq
[-1,1]^{|\bar S_{\mathrm{inp}}|+|S_{\mathrm{ref}}|}.
\end{equation}
\begin{lemma}[Feasible-set geometry]
\label{lemma:erm-feasible-set-geometry}
$K$ is convex, compact, and has diameter
$$
D_2 \le 2\sqrt{|\bar S_{\mathrm{inp}}|+|S_{\mathrm{ref}}|}.
$$
\end{lemma}

\begin{proof}
Convexity follows immediately from the definition of $K$ as a convex hull. Moreover, $K\subseteq[-1,1]^{|\bar S_{\mathrm{inp}}|+|S_{\mathrm{ref}}|}$, which gives the diameter bound. Finally, $K$ is the convex hull of \textit{a finite set}, as $\{c[S]:c\in\mathcal C\}\subseteq\{\pm1\}^{|\bar S_{\mathrm{inp}}|+|S_{\mathrm{ref}}|}$, and is therefore compact.
\end{proof}

Next, we show that our objective is linear in the concept $c \in \mathcal C$.

\begin{lemma}[Linearization]
\label{lemma:erm-linearization}

For each $t\in[T]$, $\bar\mu_t\in\operatorname{conv}(\mathcal C)$ and
$h_t=\operatorname{maj}_k\circ\bar\mu_t$ are fixed. There exist a constant $C_t\in\R$ and
a vector
$
g_t\in
\R^{|\bar S_{\mathrm{inp}}|+|S_{\mathrm{ref}}|}
$
such that, for every $c\in\mathcal C$,
$$
U_t(c)
=
\widehat{\err}_{S_{\mathrm{ref}}}(h_t,c)
-
\eta
\left(
c,\bar\mu_t;
\bar S_{\mathrm{inp}},S_{\mathrm{ref}}
\right)
=
C_t+\left\langle g_t,c[S]\right\rangle,
$$
where
\begin{equation}
\label{eq:erm-grad}
\left\langle g_t,c[S]\right\rangle
=
\sum_{(x,y)\in\bar S_{\mathrm{inp}}}
\frac{
y\bigl(2+|\bar\mu_t(x)|\bigr)-\bar\mu_t(x)
}{
4|\bar S_{\mathrm{inp}}|(1+|\bar\mu_t(x)|)
}
c(x)
+
\sum_{x\in S_{\mathrm{ref}}}
\frac1{2|S_{\mathrm{ref}}|}
\left(
\frac{\bar\mu_t(x)}{1+|\bar\mu_t(x)|}
-
h_t(x)
\right)c(x).
\end{equation}

\end{lemma}

Under the setup of \Cref{lemma:erm-linearization}, we use the same notation $U_t$ for its affine extension to $K$ from Eq.~\eqref{eq:erm-convex-set}, defined by
\begin{equation}\label{eq:erm-affine-loss}
U_t(\bar\mu[S]) := C_t+\left\langle g_t,\bar\mu[S]\right\rangle.
\end{equation}
In particular, $U_t(c[S])=U_t(c)$ for every $c\in\mathcal C$. We next bound the norm of $g_t$.

\begin{lemma}[Gradient bound]
\label{lemma:erm-gradient-bound}
For every $t\in[T]$, the vector $g_t$ from
\Cref{lemma:erm-linearization} satisfies
$$
\|g_t \|_2
\le
\sqrt{
\frac{1}{4|\bar S_{\mathrm{inp}}|}
+
\frac{9}{16|S_{\mathrm{ref}}|}
}.
$$
\end{lemma}

We prove \Cref{lemma:erm-linearization} in Appendix~\ref{sec:erm-linearization} and \Cref{lemma:erm-gradient-bound} in Appendix~\ref{sec:erm-gradient-bound}. It remains to show that the Frank-Wolfe update can be implemented via weighted ERM.

\begin{lemma}[ERM reduction]
\label{lemma:erm-linear-oracle}
At iteration $t \in [T]$, the Frank--Wolfe optimization step is
a weighted ERM problem on the disjoint union
$\bar S_{\mathrm{inp}}\cup\{(x,1):x\in S_{\mathrm{ref}}\}$
with weights
\begin{equation}
\label{eq:erm-weights}
\begin{aligned}
w_t(x,y)
&:=
\frac{(|\bar S_{\mathrm{inp}}|+|S_{\mathrm{ref}}|)y}
{|\bar S_{\mathrm{inp}}|}
\left(
\frac{\alpha}{4}\sum_{r=1}^{t-1}
\frac{y(2+|\bar\mu_r(x)|)-\bar\mu_r(x)}{1+|\bar\mu_r(x)|}
-\bigl(\bar\mu_t(x)-\bar\mu_1(x)\bigr)
\right),
\\[6pt]
w_t(x,1)
&:=
\frac{|\bar S_{\mathrm{inp}}|+|S_{\mathrm{ref}}|}
{|S_{\mathrm{ref}}|}
\left(
\frac{\alpha}{2}\sum_{r=1}^{t-1}
\left(
\frac{\bar\mu_r(x)}{1+|\bar\mu_r(x)|}-h_r(x)
\right)
-\bigl(\bar\mu_t(x)-\bar\mu_1(x)\bigr)
\right),
\end{aligned}
\end{equation}
where the first formula applies to $(x,y) \in \bar S_{\mathrm{inp}}$
and the second to $x \in S_{\mathrm{ref}}$. With probability at least $1-\delta$, all $T$ steps achieve additive accuracy $\gamma_t/25$ using at most $T=O(\varepsilon^{-4})$ ERM calls.
\end{lemma}

\begin{proof}

By the definition of $F_t$ in Algorithm~\ref{fig:erm-oracle-algorithm},
$$
\nabla F_t(\bar\mu_t[S])=\alpha\sum_{r=1}^{t-1}g_r-\left(\frac{\bar\mu_t|_{\bar S_{\mathrm{inp}}}-\bar\mu_1|_{\bar S_{\mathrm{inp}}}}{|\bar S_{\mathrm{inp}}|},\frac{\bar\mu_t|_{S_{\mathrm{ref}}}-\bar\mu_1|_{S_{\mathrm{ref}}}}{|S_{\mathrm{ref}}|}\right).
$$
Therefore, for every $c\in\mathcal C$,
\begin{align*}
\left\langle\nabla F_t(\bar\mu_t[S]),c[S]\right\rangle
&=
\sum_{(x,y)\in\bar S_{\mathrm{inp}}}
\left(
\frac{\alpha}{4|\bar S_{\mathrm{inp}}|}
\sum_{r=1}^{t-1}
\frac{
y\bigl(2+|\bar\mu_r(x)|\bigr)-\bar\mu_r(x)
}{
1+|\bar\mu_r(x)|
}
-
\left(\frac{\bar\mu_t(x)-\bar\mu_1(x)}{|\bar S_{\mathrm{inp}}|}\right)
\right)c(x)
\\
&+
\sum_{x\in S_{\mathrm{ref}}}
\left(
\frac{\alpha}{2|S_{\mathrm{ref}}|}
\sum_{r=1}^{t-1}
\left(
\frac{\bar\mu_r(x)}{1+|\bar\mu_r(x)|}
-
h_r(x)
\right)
-
\left(\frac{\bar\mu_t(x)-\bar\mu_1(x)}{|S_{\mathrm{ref}}|}\right)
\right)c(x).
\end{align*}
We may assume that $\nabla F_t(\bar\mu_t[S])\neq 0$. If $\nabla F_t(\bar\mu_t[S])=0$, then every $c\in\mathcal C$ maximizes the linear objective, since $\left\langle \nabla F_t(\bar\mu_t[S]),c[S] \right\rangle=0$. Thus, the algorithm may choose any $c_t\in\mathcal C$; the required optimization guarantee holds with zero error, and no ERM call is needed.

Under this assumption, at least one of the weights defined below in this proof is nonzero. Apply \Cref{fact:weighted-erm} to the disjoint union $\bar S_{\mathrm{inp}}\cup
\{(x,1):x\in S_{\mathrm{ref}}\}$,
with weight bound
$\frac{3(|\bar S_{\mathrm{inp}}|+|S_{\mathrm{ref}}|)}
{\min\{|\bar S_{\mathrm{inp}}|,|S_{\mathrm{ref}}|\}}T^{1/4}$,
accuracy $\gamma_t/25$,
and failure probability $\delta/T$:
$$
\begin{aligned}
&(x,y)\in\bar S_{\mathrm{inp}}:\;
w_t(x,y)
:=
\frac{(|\bar S_{\mathrm{inp}}|+|S_{\mathrm{ref}}|)y}
{|\bar S_{\mathrm{inp}}|}
\left(
\frac{\alpha}{4}\sum_{r=1}^{t-1}
\frac{y(2+|\bar\mu_r(x)|)-\bar\mu_r(x)}{1+|\bar\mu_r(x)|}
-\bigl(\bar\mu_t(x)-\bar\mu_1(x)\bigr)
\right),
\\[6pt]
&x\in S_{\mathrm{ref}}:\;
w_t(x,1)
:=
\frac{|\bar S_{\mathrm{inp}}|+|S_{\mathrm{ref}}|}
{|S_{\mathrm{ref}}|}
\left(
\frac{\alpha}{2}\sum_{r=1}^{t-1}
\left(\frac{\bar\mu_r(x)}{1+|\bar\mu_r(x)|}-h_r(x)\right)
-\bigl(\bar\mu_t(x)-\bar\mu_1(x)\bigr)
\right).
\end{aligned}
$$
Since $y^2=1$,
\begin{align*}
&\frac{1}{|\bar S_{\mathrm{inp}}|+|S_{\mathrm{ref}}|}
\left(\sum_{(x,y)\in\bar S_{\mathrm{inp}}}w_t(x,y)c(x)y +
\sum_{x\in S_{\mathrm{ref}}}w_t(x,1)c(x)\right)
=
\left\langle
\nabla F_t(\bar\mu_t[S]),c[S]
\right\rangle.
\end{align*}
Thus, maximizing the linear objective is exactly a weighted ERM problem.  Since $K=\operatorname{conv}\{c[S]:c\in\mathcal C\}$, $v = \sum_j \lambda_j c_j[S]$ for $\sum_j \lambda_j = 1$  and $\lambda_j \ge 0$. Then linearity gives
$$
\max_{v\in K} \left\langle \nabla F_t(\bar\mu_t[S]),v \right\rangle = \max_{\substack{c_j\in\mathcal C,\ \lambda_j\ge0\\
\sum_j\lambda_j=1}}
\sum_j\lambda_j
\left\langle
\nabla F_t(\bar\mu_t[S]),c_j[S]
\right\rangle = \max_{c\in\mathcal C} \left\langle\nabla F_t(\bar\mu_t[S]),c[S]\right\rangle.
$$
Each weight is $|\bar S_{\mathrm{inp}}|+|S_{\mathrm{ref}}|$ times the corresponding coordinate of $\nabla F_t(\bar\mu_t[S])$, up to multiplication by
the corresponding label in $\{\pm1\}$. Therefore,
\begin{align*}
\max\left(
\max_{(x,y)\in\bar S_{\mathrm{inp}}}|w_t(x,y)|,
\max_{x\in S_{\mathrm{ref}}}|w_t(x,1)|
\right)
&\le
\frac{|\bar S_{\mathrm{inp}}|+|S_{\mathrm{ref}}|}
{\min(|\bar S_{\mathrm{inp}}|,|S_{\mathrm{ref}}|)}
(\alpha T+2)
\tag{coordinate bounds}
\\
&=
\frac{|\bar S_{\mathrm{inp}}|+|S_{\mathrm{ref}}|}
{\min(|\bar S_{\mathrm{inp}}|,|S_{\mathrm{ref}}|)}
\left(\frac{1}{2}T^{1/4}+2\right)
\tag{$\alpha=1/(2T^{3/4})$}
\\
&\le
\frac{3(|\bar S_{\mathrm{inp}}|+|S_{\mathrm{ref}}|)}
{\min(|\bar S_{\mathrm{inp}}|,|S_{\mathrm{ref}}|)}
T^{1/4}.
\tag{$T\ge1$}
\end{align*}
By \Cref{fact:weighted-erm}, with probability at least
$1-\delta/T$, the ERM call returns a concept
$c_t\in\mathcal C$ satisfying
$$
\left\langle\nabla F_t(\bar\mu_t[S]),c_t[S]\right\rangle
\ge
\max_{v\in K}
\left\langle\nabla F_t(\bar\mu_t[S]),v\right\rangle
-\frac{\gamma_t}{25}.
$$
Moreover, $\frac1{\gamma_t}\le\sqrt T,$ so the synthetic sample used by \Cref{fact:weighted-erm} has size
$
\operatorname{poly}
\left(
|\bar S_{\mathrm{inp}}|,|S_{\mathrm{ref}}|,
T,
\log\frac{T}{\delta}
\right).
$
A union bound shows that all $T$ optimization guarantees hold simultaneously with probability at least $1-\delta$.
\end{proof}

\subsection{Learning guarantee}\label{sec:ofw-guarantee}

In this section, we prove the following theorem.

\begin{theorem}[Robust learning with access to an ERM oracle]
\label{thm:erm-main-guarantee}
Let \(\mathcal{D}\) be a distribution over \(X\), and let
\(\mathcal C\subseteq\{\pm1\}^X\) be a concept class of finite VC dimension. For \(\eta\in[0,1)\) and
\(\varepsilon,\delta\in(0,1/2)\), suppose
$
|\bar S_{\mathrm{inp}}|
\ge
C
\frac{
\operatorname{VCdim}(\mathcal C)+\log(1/\delta)
}{
\varepsilon^2
}
$
and
$
|S_{\mathrm{ref}}|
\ge
C
\frac{
\operatorname{VCdim}(\mathcal C)\log(2/\varepsilon)
+\log(1/\delta)
}{
\varepsilon^4
}
$
for a sufficiently large universal constant $C$.
Under $\eta$-bounded contamination, Algorithm~\ref{fig:erm-oracle-algorithm}, run with accuracy $\varepsilon$ and failure probability $\delta/3$, outputs a randomized hypothesis $h$ such that, with probability at least $1-\delta$,
$$
\err_{\mathcal D}(h,c^\star)\le\eta+\varepsilon.
$$
The algorithm uses
$
T
=
O(\varepsilon^{-4})
$
ERM calls and runs in time
$\operatorname{poly}(|\bar S_{\mathrm{inp}}|,|S_{\mathrm{ref}}|,1/\varepsilon,
\log(1/\delta))$.
\end{theorem}

To prove this theorem, we use two additional lemmas.
\Cref{lemma:erm-error-guarantee} shows that the hypothesis $\bar h_T$
produced by our algorithm has small loss $L(\bar h_T)$.

\begin{lemma}[Error guarantee]
\label{lemma:erm-error-guarantee}

Let $\bar h_T$ be the hypothesis produced by Algorithm \ref{fig:erm-oracle-algorithm}. Suppose that every Frank--Wolfe optimization step satisfies Eq.~\eqref{eq:approx-max}, then
\begin{equation}
\label{eq:erm-all-concepts}
L(\bar h_T) = \sup_{c \in \mathcal C}(\widehat{\err}_{S_{\mathrm{ref}}}(\bar h_T,c)
-
\sup_{\bar\mu\in\operatorname{conv}(\mathcal C)}
\eta
(
c,\bar\mu;
\bar S_{\mathrm{inp}},S_{\mathrm{ref}}
))
\le
\frac{\varepsilon}{4}.
\end{equation}

\end{lemma}

\Cref{lemma:erm-certificate-validity} indicates that the corruption
certificate for the target concept $c^\star$ is uniformly bounded over
$\bar\mu\in\operatorname{conv}(\mathcal C)$ by
$\eta+\varepsilon/2$.

\begin{lemma}[Certificate guarantee]
\label{lemma:erm-certificate-validity}

Let $\mathcal C\subseteq\{\pm1\}^X$ have finite VC dimension and let $\mathcal{D}$
be a distribution over $X$. There is a universal constant $C>0$ such that, if
$
|\bar S_{\mathrm{inp}}|,|S_{\mathrm{ref}}|
\ge
C
\frac{
\operatorname{VCdim}(\mathcal C)+\log(4/\delta)
}{
\varepsilon^2
},
$
then, with probability at least $1-\delta$,
\begin{equation}
\label{eq:erm-certificate-validity}
\sup_{\bar\mu\in\operatorname{conv}(\mathcal C)}
\eta(c^\star,\bar\mu;\bar S_{\mathrm{inp}},S_{\mathrm{ref}})
\le
\eta+\frac{\varepsilon}{2}.
\end{equation}

\end{lemma}

We prove \Cref{lemma:erm-error-guarantee} in Appendix~\ref{sec:erm-error-guarantee} and \Cref{lemma:erm-certificate-validity} in Appendix~\ref{sec:erm-certificate-guarantee}. We also use standard generalization bounds in \Cref{lemma:erm-generalization} from Appendix~\ref{app:generalization}. Before these proofs, we show how these two lemmas, along with an ERM reduction in \Cref{lemma:erm-linear-oracle}, imply the main guarantee for robust learning with access to ERM oracle.

\begin{proof}[Proof of \Cref{thm:erm-main-guarantee} assuming
\Cref{lemma:erm-error-guarantee}, \Cref{lemma:erm-certificate-validity}, and \Cref{lemma:erm-generalization}
]
The assumed sample sizes allow us to apply \Cref{lemma:erm-certificate-validity} and \Cref{lemma:erm-generalization} with failure probability $\delta/3$ each. We also allocate failure probability $\delta/3$ to all weighted ERM calls and apply \Cref{lemma:erm-linear-oracle} with this parameter. A union bound shows that all three guarantees hold simultaneously with probability at least $1-\delta$. Throughout the remainder of the proof, we condition on this event.

Our algorithm produces a hypothesis $\bar h_T := T^{-1}\sum_{t=1}^T h_t$. For every $t\in[T]$, $h_t= \operatorname{maj}_k\circ\bar\mu_t \in \operatorname{conv}(\operatorname{Maj}_k(\mathcal C))$ and hence $\bar h_T = T^{-1}\sum_{t=1}^T h_t \in \operatorname{conv}(\operatorname{Maj}_k(\mathcal C))$. Therefore, Eq.~\eqref{eq:erm-error-generalization} in \Cref{lemma:erm-generalization} applies to $\bar h_T$:
\begin{equation}\label{eq:erm-step-1}
\err_{\mathcal{D}}(\bar h_T,c^\star)
\le
\widehat{\err}_{S_{\mathrm{ref}}}
(\bar h_T,c^\star)
+
\frac{\varepsilon}{4}.
\end{equation}
Therefore,
\begin{align*}
\err_{\mathcal{D}}(\bar h_T,c^\star)
&\le
\widehat{\err}_{S_{\mathrm{ref}}}
(\bar h_T,c^\star)
+
\frac{\varepsilon}{4}
\tag{Eq.~\eqref{eq:erm-step-1}}
\\
&\le
\sup_{\bar\mu \in \operatorname{conv}(\mathcal C)}
\eta
\left(
c^\star,\bar\mu;
\bar S_{\mathrm{inp}},S_{\mathrm{ref}}
\right)
+
\frac{\varepsilon}{2}
\tag{\Cref{lemma:erm-error-guarantee}}
\\
&\le
\eta+\varepsilon.
\tag{\Cref{lemma:erm-certificate-validity}}
\end{align*}
For the output hypothesis
$h(x)\sim\operatorname{Rad}(\bar h_T(x))$,
$$
\err_{\mathcal{D}}(h,c^\star)
=
\E_{x\sim \mathcal{D}}
\left[
\frac{1-c^\star(x)\bar h_T(x)}{2}
\right]
=
\err_{\mathcal{D}}(\bar h_T,c^\star) \le \eta + \varepsilon.
$$
The runtime follows from \Cref{lemma:erm-linear-oracle}.

\end{proof}

\section{Robust learning with low-degree polynomials}\label{sec:learning-with-polynomials}

The goal of this section is to give an algorithm that outputs a hypothesis with error at most $\eta+\varepsilon$ when the concept class $\mathcal C$ admits low-degree ${\cal L}_2$ approximators. In Subsection~\ref{sec:algorithm}, we present the algorithm. In Subsection~\ref{sec:no-regret}, we show that its update is an instance of a well-known online optimization procedure. In Subsection~\ref{sec:guarantee}, we prove that the resulting hypothesis achieves error at most $\eta+\varepsilon$.

Throughout this section, we assume that $\mathcal D$ is $A$-hypercontractive in the following sense.

\begin{definition}[Hypercontractivity]\label{def:hypercontractivity}
A distribution $\mathcal{D}$ over $\R^d$ is hypercontractive with respect to polynomials if there is a constant $A\ge1$ such that for any polynomial $p$ over $\R^d$ and any $r\ge2$ we have
\begin{enumerate}
    \item $\E_{ x\sim \mathcal{D}}[|p( x)|^r]\le (Ar)^{\ell r}\left(\E_{ x\sim \mathcal{D}}[|p( x)|]\right)^r$, where $\ell$ is the degree of $p$.
    \item The absolute expectations of degree-$1$ monomials are finite under $\mathcal{D}$.
\end{enumerate}
\end{definition}

We use the standard notion of ${\cal L}_2$ polynomial approximation
\cite{kalai2008agnostically}.

\begin{definition}[Polynomial approximators]
\label{def:polynomial-approximators}
For $\varepsilon\in(0,1)$, we say that a class $\mathcal C\subseteq\{\pm1\}^X$ has $\varepsilon$-approximate degree $\ell$ with respect to a distribution $\mathcal{D}$ over $X$ if, for every $c\in\mathcal C$, there is a polynomial $q$ of degree $\ell$ such that
$$
\E_{x\sim \mathcal{D}}\bigl[(c(x)-q(x))^2\bigr]\le\varepsilon.
$$
\end{definition}

This allows us to associate every concept $c\in\mathcal C$ with a low-degree polynomial. For optimization, we use polynomials of the form  
\begin{equation}\label{eq:poly-in-different-coordinates}
q_w(x):=\langle w,\psi_\ell(x)\rangle,
\end{equation}
where $\psi_\ell$ is a normalized feature map. Instead of working directly in the standard monomial basis $x^{\otimes\ell}$, we define $\psi_\ell(x):=M_\ell^{-1/2}x^{\otimes\ell}$ with $M_\ell:=\frac1{|S_{\mathrm{ref}}|}\sum_{x\in S_{\mathrm{ref}}}x^{\otimes\ell}(x^{\otimes\ell})^\top$.\footnote{If $M_\ell$ is singular, $M_\ell^{-1/2}$ denotes the Moore--Penrose pseudoinverse of $M_\ell^{1/2}$.} We project the feature vectors $x^{\otimes\ell}$ onto the range of $M_\ell$ and work in orthonormal coordinates on this range. This choice makes $\|w\|_2^2$ equal to the average of $q_w(x)^2$ over $S_{\mathrm{ref}}$ and allows us to bound the gradients without assuming bounded monomial coefficients. We therefore define the feasible set of polynomials
\begin{equation}\label{eq:feasible-set}
\mathcal Q=\mathcal Q(\ell,S_{\mathrm{ref}})
:=
\left\{
q_w:\|w\|_2\le2
\right\}.
\end{equation}
We show below that $\mathcal Q$ contains a polynomial $q^\star$
that approximates the target concept $c^\star$. A general $q\in\mathcal Q$ need not approximate any concept in $\mathcal C$.

\subsection{Robust learning algorithm via online gradient ascent}\label{sec:algorithm}

We extend the corruption certificate (\Cref{def:erm-empirical-certificate}) to polynomials using its explicit form in Eq.~\eqref{eq:erm-certificate-again}.

\begin{definition}[Polynomial corruption certificate]\label{def:polynomial-certificate}
For polynomials $q,m$ of degree $\ell$ and a filtered sample $\bar S_{\mathrm{filt}}\subseteq\bar S_{\mathrm{inp}}$,
define
$$
\eta^{\mathrm{poly}}
(q,m;\bar S_{\mathrm{filt}},S_{\mathrm{ref}})
:=
\frac1{2|\bar S_{\mathrm{inp}}|}
\sum_{(x,y)\in\bar S_{\mathrm{filt}}}
\frac{|m(x)|+ym(x)-2yq(x)}
{2+|m(x)|+ym(x)}
-
\frac1{2|S_{\mathrm{ref}}|}
\sum_{x\in S_{\mathrm{ref}}}
\frac{q(x)m(x)-1}{1+|m(x)|}.
$$
\end{definition}
For $q=c$ and $m=\bar\mu$, $\eta^{\mathrm{poly}}$ recovers
the corruption certificate from Section~\ref{sec:learning-with-erm} when evaluated on the full input sample. For a general real-valued polynomial $q$, it defines an affine extension of this discrepancy measure, but no longer admits the same probabilistic interpretation. As outlined in the technical overview~\ref{sec:poly-overview}, small approximation error under $\mathcal D$ does not control the polynomial at corrupted points. We therefore apply weighted iterative polynomial filtering and sum only over the remaining sample $\bar S_{\mathrm{filt}}$.

As in Section~\ref{sec:learning-with-erm}, we use an online optimization procedure. At iteration $t\in[T]$, the current polynomial $q_t$ is the analogue of the mixture $\bar\mu_t$ in the ERM setting: we set $m_t=q_t$ and $h_t=\sign(q_t)$. For every $q\in\mathcal Q$, define the iteration objective
\begin{equation}
\label{eq:loss-def}
U_t^{\mathrm{poly}}(q)
:=
\widehat{\err}_{S_{\mathrm{ref}}}(h_t,q)
-
\eta^{\mathrm{poly}}
\left(
q,m_t;
\bar S_{\mathrm{filt},t},S_{\mathrm{ref}}
\right).
\end{equation}
Let $\bar h_T:=T^{-1}\sum_{t=1}^T h_t$,
\begin{equation}\label{eq:loss-bound}
\sup_{q\in\mathcal Q}
\frac1T\sum_{t=1}^T
\left(
\widehat{\err}_{S_{\mathrm{ref}}}(h_t,q)
-
\eta^{\mathrm{poly}}
\left(
q,m_t;
\bar S_{\mathrm{filt},t},S_{\mathrm{ref}}
\right)
\right)
=
\sup_{q\in\mathcal Q}
\frac1T\sum_{t=1}^T U_t^{\mathrm{poly}}(q).
\end{equation}
Thus, it suffices to bound the average iteration objective simultaneously for every $q\in\mathcal Q$. The algorithm applies projected gradient ascent to
$(U_t^{\mathrm{poly}})_{t=1}^T$ over the polynomial coefficients and outputs $\bar h_T$. We bound the right-hand side of Eq.~\eqref{eq:loss-bound} using a no-regret guarantee.

Unlike the ERM setting, the iteration objective changes not only through $m_t$ and $h_t$, but also through the filtered sample $\bar S_{\mathrm{filt},t}$. Before forming each objective, the algorithm repeats weighted iterative polynomial filtering on the previously filtered sample, with weights determined by the current $m_t$. We describe the weighted filtering procedure in Algorithm~\ref{fig:weighted-ipf}; its guarantees are established in \Cref{thm:weighted-ipf}.

\begin{algorithm}[H]
\caption{Weighted iterative polynomial filtering.}
\label{fig:weighted-ipf}
\small
\begin{algorithmic}[1]
\Procedure{WIPF}
{$\bar S_{\mathrm{inp},t},S_{\mathrm{ref}},m_t,\varepsilon,
|\bar S_{\mathrm{inp}}|$}

\vspace{0.5em}

\State \textbf{Input:} A sample $\bar S_{\mathrm{inp},t}$,
a sample $S_{\mathrm{ref}}$, a function $m_t$,
accuracy $\varepsilon$, degree $\ell$, sample size $|\bar S_{\mathrm{inp}}|$.

\State \textbf{Output:} A filtered sample
$\bar S_{\mathrm{filt}}$.

\State\label{line:weighted-ipf-parameters}Set $R:=2$, $\beta:=2(2A)^{4\ell}$, $\gamma:=\frac{\varepsilon}{64}$, $B:=16(d+1)^{\ell}\sqrt{\frac{\beta}{\varepsilon}}$, $\Delta:=\frac{\varepsilon}{32B}$.

\State\label{line:poly-family}Define
$
\mathcal P
:=
\{
p:
\deg(p)\le2\ell,\;
\frac1{|S_{\mathrm{ref}}|}
\sum_{x\in S_{\mathrm{ref}}}|p(x)|\le\gamma,\;
\frac1{|S_{\mathrm{ref}}|}
\sum_{x\in S_{\mathrm{ref}}}p(x)^2\le\beta
\}.
$

\State\label{line:weighted-ipf-initial}Initialize $\bar S_{\mathrm{filt}} := \{(x,y)\in\bar S_{\mathrm{inp},t}: |p(x)|\le B\text{ for every }p\in\mathcal P\}.$

\For{$i=0,\ldots,|\bar S_{\mathrm{inp},t}|$}
    \State\label{line:maximization-step} Set $v_t(x,y) = \frac{2+|m_t(x)|-ym_t(x)}
    {8(1+|m_t(x)|)}$ and compute
    \[
    \begin{aligned}
    p^\star
    \in
    \operatorname*{argmax}_{p\in\mathcal P}
    \sum_{(x,y)\in\bar S_{\mathrm{filt}}}
    v_t(x,y) p(x) \qquad
    \lambda^\star
    :=
    \frac1{|\bar S_{\mathrm{inp}}|}
    \sum_{(x,y)\in\bar S_{\mathrm{filt}}}
    v_t(x,y)p^\star(x).
    \end{aligned}
    \]

    \State If $\lambda^\star\le\varepsilon/64$,
    \Return $\bar S_{\mathrm{filt}}$.

    \State Otherwise, let $\tau^\star\ge0$ be the
    smallest value such that
    \[
    \frac{
    |\{x\in S_{\mathrm{filt}}:
    |p^\star(x)|>\tau^\star\}|
    }{|\bar S_{\mathrm{inp}}|}
    \ge
    R\frac{
    |\{x\in S_{\mathrm{ref}}:
    |p^\star(x)|>\tau^\star\}|
    }{|S_{\mathrm{ref}}|}
    +\Delta.
    \]

    \State Set $\bar S_{\mathrm{filt}} := \bar S_{\mathrm{filt}}\setminus \{(x,y)\in\bar S_{\mathrm{filt}}: |p^\star(x)|>\tau^\star\}.
    $
\EndFor
\EndProcedure
\end{algorithmic}
\end{algorithm}

Algorithm \ref{fig:polynomial-algo} combines (weighted) iterative polynomial filtering with an online convex optimization procedure over low-degree polynomials.

\begin{algorithm}[H]
\caption{Algorithm for $\eta+\varepsilon$ learning with polynomials.}
\label{fig:polynomial-algo}
\begin{algorithmic}[1]
\Procedure{Online-Gradient-Ascent}
{$\bar S_{\mathrm{inp}},S_{\mathrm{ref}},\varepsilon,\delta$}

\State \textbf{Input:} A corrupted sample $\bar S_{\mathrm{inp}}$, a sample $S_{\mathrm{ref}}$, accuracy $\varepsilon$, failure probability $\delta$, degree $\ell$.

\State Set $T:=\left\lceil (64C A^2(d+1))^{2\ell}/\varepsilon^5 \right\rceil$ for a sufficiently large universal constant $C$.

\State Initialize $\bar S_{\mathrm{filt},0}:=\bar S_{\mathrm{inp}}$ and $q_{w_1} = 0$.

\For{$t=1,\ldots,T$}
    \State Set $m_t(x)=q_{w_t}(x)$ and $h_t(x)=\sign(q_{w_t}(x))$ for $q_{w_t}$ from Eq.~\eqref{eq:poly-in-different-coordinates}.
    
    \State\label{line:polynomial-filtering}$\bar S_{\mathrm{filt},t} := \hyperref[fig:weighted-ipf]{\textsc{WIPF}} \left(
    \bar S_{\mathrm{filt},t-1},
    S_{\mathrm{ref}},
    m_t,
    \varepsilon,
    |\bar S_{\mathrm{inp}}|
    \right)
    $.

    \State Form $C_t$ and $g_t$ from Eq.~\eqref{def:g} so that
    $$
    U_t^{\mathrm{poly}}(q_w)
    :=
    \widehat{\err}_{S_{\mathrm{ref}}}(h_t,q_w)
    -
    \eta^{\mathrm{poly}}
    \left(
    q_w,m_t;
    \bar S_{\mathrm{filt},t},S_{\mathrm{ref}}
    \right)
    =
    C_t+\langle w,g_t\rangle.
    $$
    \State Update over $K$ from Eq.~\eqref{eq:feasible-set-of-coeff}
    $$
    w_{t+1} = \operatorname*{argmin}_{w\in K}\left\|w-(w_t+\alpha g_t)\right\|_2^2\quad\quad\text{where}\quad \alpha := \frac{4\varepsilon^{3/2}}{(CA^2(d+1))^{\ell}\sqrt T}.
    $$
\EndFor

\State \Return $h(x)\sim\operatorname{Rad}(\bar h_T(x))$ with $\bar h_T(x)=T^{-1}\sum_{t=1}^T h_t(x)$.

\EndProcedure
\end{algorithmic}
\end{algorithm}
Every quantity in Algorithm~\ref{fig:polynomial-algo}, in particular the iterates $q_t,m_t,h_t$, is computed from the corrupted sample $\bar S_{\mathrm{inp}}$ and the reference sample $S_{\mathrm{ref}}$. In contrast, the clean sample $\bar S_{\mathrm{cln}}$, the target concept $c^\star$, and its ${\cal L}_2$ approximator $q^\star$ are used only in the analysis and are unavailable to the learner.

\subsection{No-regret bound}\label{sec:no-regret}

The online optimization procedure in Algorithm \ref{fig:polynomial-algo} is an instance of projected gradient ascent. We recall its no-regret guarantee~\eqref{eq:zink-no-regret} from \cite{Zin03}. This guarantee allows us to bound the right-hand side of Eq.~\eqref{eq:loss-bound}. In this subsection, we verify the conditions required for this guarantee.

\begin{theorem}[Projected gradient ascent, {\cite[Theorem~1]{Zin03}}]
\label{thm:zinkevich-projected-gradient}
Let $K\subseteq\R^N$ be a compact convex set with diameter $D_2=\sup_{u,w\in K}\|u-w\|_2$. Let $T \in \mathbb N$, and let $(g_t)_{t=1}^T$ be vectors satisfying $\|g_t\|_2\le G_2$ for every $t\in[T]$ and some $G_2>0$. Run projected gradient ascent update with learning rate $\alpha:=D_2/(G_2\sqrt T)$ and some $w_1 \in K$
$$
w_{t+1}=\operatorname*{argmin}_{w\in K}\left\|w-(w_t+\alpha g_t)\right\|_2^2.
$$
Then, for every $w\in K$, the following no-regret bound holds:
\begin{equation}\label{eq:zink-no-regret}
\sum_{t=1}^T\langle w,g_t\rangle
\le
\sum_{t=1}^T\langle w_t,g_t\rangle
+
D_2G_2\sqrt T.
\end{equation}
\end{theorem}
For the feasible set of polynomials defined in Eq.~\eqref{eq:feasible-set}, the corresponding feasible set of coefficients is
\begin{equation}\label{eq:feasible-set-of-coeff}
K
:=
\left\{
w:q_w\in\mathcal Q
\right\}
=
\left\{
w:\|w\|_2\le 2
\right\}.
\end{equation}
\noindent The set $K$ is a nonempty, compact, and convex set. Consequently, the Euclidean projection
$
\operatorname*{argmin}_{w\in K}
\|w-z\|_2^2
$
exists and is unique. Moreover, projection onto the $\ell_2$-ball has a closed-form
expression and can be computed in linear time \cite[Section~6.5.1]{ParikhBoyd2014}.

It remains to show that the objective~\eqref{eq:loss-def} is affine in the polynomial coefficients $w$ and that its gradient has bounded norm.

\begin{lemma}[Linearization]\label{lemma:linearization}

For each $t \in [T]$, $m_t=q_t$ and $h_t=\sign(q_t)$ are fixed. There exist a constant $C_t\in\R$ and a vector $g_t$ such that, for every polynomial $q(x)=\langle w, \psi_\ell(x)\rangle$, we have
$$
U_t^{\mathrm{poly}}(q) =\widehat{\err}_{S_{\mathrm{ref}}}(h_t,q) - \eta^{\mathrm{poly}}\left(q,m_t;\bar S_{\mathrm{filt},t},S_{\mathrm{ref}}\right) =C_t+\langle w,g_t\rangle,
$$
where
\begin{equation}\label{def:g}
\begin{aligned}
g_t
=
\sum_{(x,y)\in\bar S_{\mathrm{filt},t}}
\frac{
y\bigl(2+|m_t(x)|\bigr)-m_t(x)
}{4|\bar S_{\mathrm{inp}}|
(1+|m_t(x)|)
}
\psi_\ell(x)
+
\sum_{x\in S_{\mathrm{ref}}}
\frac{1}{2|S_{\mathrm{ref}}|}
\left(
\frac{m_t(x)}{1+|m_t(x)|}-h_t(x)
\right)
\psi_\ell(x).
\end{aligned}
\end{equation}

\end{lemma}

\begin{lemma}[Gradient bound]\label{lemma:gradient-bound}
For every iteration $t\in[T]$ of the algorithm in
\Cref{fig:polynomial-algo}, the gradient defined in
Eq.~\eqref{def:g} satisfies, for a sufficiently large universal constant $C>0$, 
$$
\|g_t\|_2
\le
\frac{(CA^2 (d+1))^{\ell}}{\varepsilon^{3/2}}.
$$
\end{lemma}

We prove \Cref{lemma:linearization} in Appendix~\ref{sec:linearization} and \Cref{lemma:gradient-bound} in Appendix~\ref{sec:gradient-bound}.

\subsection{Learning guarantee}\label{sec:guarantee}

In this subsection, we show that a hypothesis $h$ produced by Algorithm~\ref{fig:polynomial-algo} achieves the desired error bound. 

\begin{theorem}[Robust learning with low-degree polynomials, formal version of \Cref{thm:main-guarantee}]\label{thm:main-guarantee-restated}
Let $\mathcal{D}$ be an $A$-hypercontractive distribution and suppose $\mathcal C$ admits ${\cal L}_2(\mathcal{D})$ polynomial approximators of degree $\ell \ge 1$ with squared error $\varepsilon^2/1536^2$. For $\varepsilon,\delta\in(0,1)$,
suppose $|\bar S_{\mathrm{inp}}| \ge
\frac{(C_\star A(d+1))^{4\ell}}{\varepsilon^3}
\log\frac{16}{\delta}$ and $|S_{\mathrm{ref}}| \ge
\frac{(C_\star A(d+1))^{4\ell}}{\varepsilon^3}
(\log\frac{16}{\delta})^{8\ell+1}$
for a sufficiently large universal constant $C_\star$.
Under $\eta$-bounded contamination, Algorithm~\ref{fig:polynomial-algo} runs in time
$\operatorname{poly}(|\bar S_{\mathrm{inp}}|,|S_{\mathrm{ref}}|,(d+1)^{\ell})$
and outputs a randomized hypothesis $h$ such that,
with probability at least $1-\delta$,
$$
\err_{\mathcal{D}}(h,c^\star)\le\eta+\varepsilon.
$$
\end{theorem}

We use two lemmas to prove the main guarantee. \Cref{lemma:error-guarantee} bounds the loss of the output hypothesis.

\begin{lemma}[Error guarantee]
\label{lemma:error-guarantee}

Let $\bar h_T$ be the hypothesis produced by Algorithm \ref{fig:polynomial-algo}. Then
$$
\sup_{q \in \mathcal Q} \left( \widehat{\err}_{S_{\mathrm{ref}}}(\bar h_T,q)
-
\frac1T\sum_{t=1}^T
\eta^{\mathrm{poly}}
(q,m_t;\bar S_{\mathrm{filt},t},S_{\mathrm{ref}})\right)
\le
\frac{\varepsilon}{16}.
$$

\end{lemma}

\Cref{lemma:certificate-guarantee} bounds the corruption certificate for the ${\cal L}_2$ approximator $q^\star$ of the target concept.

\begin{lemma}[Certificate guarantee]
\label{lemma:certificate-guarantee}

Let $\mathcal{D}$ be an $A$-hypercontractive distribution, and let $q^\star$
be a degree-$\ell$ polynomial with ${\cal L}_2$ approximation error $\varepsilon/1536$ for $c^\star$.
There is a sufficiently large universal constant $C_\star$ such that,
if
$
|\bar S_{\mathrm{inp}}| \ge \frac{(C_\star A(d+1))^{4\ell}}{\varepsilon^3}\log\frac{16}{\delta}
$
and
$
|S_{\mathrm{ref}}| \ge \frac{(C_\star A(d+1))^{4\ell}}{\varepsilon^3}
\left(\log\frac{16}{\delta}\right)^{8\ell+1},
$
then, under $\eta$-bounded contamination, with probability at least
$1-\delta$,
$$
\frac{1}{T}\sum_{t=1}^T\eta^{\mathrm{poly}} \left(q^\star,m_t;\bar S_{\mathrm{filt},t},S_{\mathrm{ref}}\right) \le \eta+\frac{13\varepsilon}{64}.
$$

\end{lemma}

We prove \Cref{lemma:error-guarantee} in Appendix~\ref{sec:proof-of-error-guarantee} and \Cref{lemma:certificate-guarantee} in Appendix~\ref{sec:proof-of-certificate-guarantee}. We also use standard generalization bounds in \Cref{lemma:generalization} from Appendix~\ref{app:generalization}. For now, we show how these results imply our main guarantee for learning with bounded contamination using low-degree polynomials. 

\begin{proof}[Proof of \Cref{thm:main-guarantee-restated} assuming
\Cref{lemma:error-guarantee} and \Cref{lemma:certificate-guarantee}]

The chosen sample sizes allow us to apply
\Cref{lemma:certificate-guarantee}, \Cref{lemma:generalization}, and
\Cref{lemma:empirical-approximation} for a sufficiently large universal constant $C_\star$. We apply each result with failure probability $\delta/3$ and take a union bound, so their conclusions hold simultaneously with probability at least $1-\delta$. Throughout the remainder of the proof, we work on this event.

Our algorithm produces a hypothesis $\bar h_T := T^{-1}\sum_{t=1}^T h_t$. For every $t\in[T]$, $h_t=\sign(q_t)\in\mathcal H_0$, and hence $\bar h_T = T^{-1}\sum_{t=1}^T h_t \in \operatorname{conv}(\mathcal H_0)$. Therefore, Eq.~\eqref{eq:population-transfer-error} in \Cref{lemma:generalization} applies to $\bar h_T$:
\begin{equation}\label{eq:step-1}
\err_{\mathcal{D}}(\bar h_T,c^\star)
\le
\widehat{\err}_{S_{\mathrm{ref}}}
(\bar h_T,c^\star)
+
\frac{\varepsilon}{16}.
\end{equation}
By the ${\cal L}_2$ approximation assumption, let $q^\star$ be a degree $\ell$ polynomial satisfying 
\begin{equation}\label{eq:l2-approx}
\E_{x\sim \mathcal{D}}\bigl[(q^\star(x)-c^\star(x))^2\bigr] \le \frac{\varepsilon^2}{1536^2}.
\end{equation}
First, we show that replacing $c^\star$ by its ${\cal L}_2$ approximator $q^\star$ in the empirical error costs at most $\varepsilon/32$:
\begin{align*}
\big|
\widehat{\err}_{S_{\mathrm{ref}}}(\bar h_T,c^\star)
-
\widehat{\err}_{S_{\mathrm{ref}}}(\bar h_T,q^\star)
\big|
&=
\left|
\frac1{2|S_{\mathrm{ref}}|}
\sum_{x\in S_{\mathrm{ref}}}
\bar h_T(x)\bigl(q^\star(x)-c^\star(x)\bigr)
\right|
\tag{\Cref{def:linear-error}}
\\
&\le
\frac1{2|S_{\mathrm{ref}}|}
\sum_{x\in S_{\mathrm{ref}}}
|q^\star(x)-c^\star(x)|
\tag{$|\bar h_T(x)|\le1$}
\\
&\le
\frac{\varepsilon}{32}.
\tag{\Cref{lemma:empirical-approximation}}
\end{align*}
Therefore,
\begin{equation}\label{eq:step-2}
\widehat{\err}_{S_{\mathrm{ref}}}(\bar h_T,c^\star)
\le
\widehat{\err}_{S_{\mathrm{ref}}}(\bar h_T,q^\star) + \frac{\varepsilon}{32}.
\end{equation}
For $M_\ell$ nonsingular\footnote{Otherwise, we use the normalized representation of $q^\star$, which agrees with it on the reference and filtered samples.}, write $q^\star(x)=\langle w^\star,\psi_\ell(x)\rangle$. By \Cref{lemma:empirical-approximation}, $(q^\star)^2-1\in\mathcal P$, so
$$
\|w^\star\|_2^2
=
\frac1{|S_{\mathrm{ref}}|}
\sum_{x\in S_{\mathrm{ref}}}q^\star(x)^2
\le
1+\frac1{|S_{\mathrm{ref}}|}
\sum_{x\in S_{\mathrm{ref}}}|q^\star(x)^2-1|
\le
1+\frac{\varepsilon}{64}
<4.
$$
Thus, $q^\star\in\mathcal Q$. We apply \Cref{lemma:error-guarantee} and \Cref{lemma:certificate-guarantee} to $q^\star$ to obtain
\begin{align*}
\err_{\mathcal{D}}(\bar h_T,c^\star)
&\le
\widehat{\err}_{S_{\mathrm{ref}}}
(\bar h_T,c^\star)
+
\frac{\varepsilon}{16}
\tag{Eq.~\eqref{eq:step-1}}
\\
&\le
\widehat{\err}_{S_{\mathrm{ref}}}
(\bar h_T,q^\star)
+
\frac{\varepsilon}{32}
+
\frac{\varepsilon}{16}
\tag{Eq.~\eqref{eq:step-2}}
\\
&\le
\frac{1}{T}\sum_{t=1}^T
\eta^{\mathrm{poly}}
\left(
q^\star,m_t;
\bar S_{\mathrm{filt},t},S_{\mathrm{ref}}
\right)
+
\frac{5\varepsilon}{32}
\tag{\Cref{lemma:error-guarantee}}
\\
&\le
\eta
+
\frac{13\varepsilon}{64}
+
\frac{5\varepsilon}{32}
\tag{\Cref{lemma:certificate-guarantee}}
\\
&\le
\eta+\varepsilon.
\end{align*}
For the output hypothesis $h(x)\sim\operatorname{Rad}(\bar h_T(x))$,
$$
\err_{\mathcal{D}}(h,c^\star)
=
\E_{x\sim \mathcal{D}}
\left[
\frac{1-c^\star(x)\bar h_T(x)}{2}
\right]
=
\err_{\mathcal{D}}(\bar h_T,c^\star) \le \eta + \varepsilon.
$$
By \Cref{thm:weighted-ipf}, each weighted filtering call runs in time $\operatorname{poly}(|\bar S_{\mathrm{inp}}|,
|S_{\mathrm{ref}}|,(d+1)^{2\ell})$. The algorithm makes $T=O((C_\star A^2(d+1))^{2\ell}/\varepsilon^5)$ such calls, and each gradient and projection
onto the $\ell_2$-ball can be computed in polynomial time. The lower bound on $|\bar S_{\mathrm{inp}}|$ therefore implies
$T=O(|\bar S_{\mathrm{inp}}|^2)$. Hence, the complete algorithm runs in time $\operatorname{poly}(|\bar S_{\mathrm{inp}}|, |S_{\mathrm{ref}}|,(d+1)^{\ell})$.

\end{proof}

\subsection{Consequences of \texorpdfstring{${\cal L}_1$}{L1} approximation}
\label{sec:l1-approximation}

For the Gaussian distribution and the uniform distribution on the Boolean hypercube, our result also yields guarantees for concept classes admitting low-degree ${\cal L}_1$ polynomial approximators. In this section, our main goal is to prove the following theorem.

\begin{theorem}[Robust learning via ${\cal L}_1$ approximation]
\label{cor:l1-robust-learning}
Let $\mathcal D$ be the standard Gaussian distribution over $\R^d$ or the uniform distribution on $\{\pm1\}^d$, and suppose $\mathcal C$ admits ${\cal L}_1(\mathcal D)$ polynomial approximators of degree $\ell\ge1$ with error $\varepsilon/8$. For $\varepsilon,\delta\in(0,1)$, set $\ell':= C_0\ell^2/\varepsilon^4$ and suppose $|\bar S_{\mathrm{inp}}| \ge
\frac{(C_\star(d+1))^{4\ell'}}{\varepsilon^3}
\log\frac{16}{\delta}$ and
$|S_{\mathrm{ref}}| \ge
\frac{(C_\star(d+1))^{4\ell'}}{\varepsilon^3}
(\log\frac{16}{\delta})^{8\ell'+1}$
for sufficiently large universal constants $C_0$ and $C_\star$. Under $\eta$-bounded contamination, Algorithm~\ref{fig:polynomial-algo}, with degree $\ell'$, accuracy $\varepsilon/2$, and failure parameter $\delta/2$, runs in time $\operatorname{poly}(|\bar S_{\mathrm{inp}}|,|S_{\mathrm{ref}}|,(d+1)^{\ell'})$ and outputs a randomized hypothesis $h$ such that, with probability at least $1-\delta$,
$$
\err_{\mathcal D}(h,c^\star)\le\eta+\varepsilon.
$$
\end{theorem}

We use the following approximation bound for polynomial threshold functions. For the uniform distribution on the Boolean hypercube, the bound relies on recent advances by OpenAI.

\begin{fact}[${\cal L}_2$ approximation of polynomial threshold functions, \cite{kane2011gaussian, OpenAIPTF26}]
\label{fact:ptf-l2-approximation}
Let $\mathcal D$ be the standard Gaussian distribution over $\R^d$  or the uniform distribution on $\{\pm1\}^d$. For every degree-$\ell$ PTF $f$ and every $\beta\in(0,1)$, there is
a degree $O(\ell^2/\beta^2)$ polynomial $p$ satisfying
\begin{equation}\label{eq:ptf-l2-approximation}
\E_{x\sim\mathcal D}\bigl[(p(x)-f(x))^2\bigr]\le\beta.
\end{equation}

\end{fact}

To prove \Cref{cor:l1-robust-learning}, we apply our main theorem (\Cref{thm:main-guarantee-restated}) with target $\sign(q)$, where $q$ approximates $c^\star$ in ${\cal L}_1$. We account for the disagreement between $\sign(q)$ and $c^\star$ as additional contamination and then bound the error against the original target $c^\star$.

\begin{proof}[Proof of \Cref{cor:l1-robust-learning}]

Fix a degree-$\ell$ approximator $q$ for $c^\star$. Then 
\begin{equation}\label{eq:l1-bound}
\Pr_{x\sim\mathcal D}[\sign(q(x))\neq c^\star(x)]
\le
\E_{x\sim\mathcal D}|q(x)-c^\star(x)|
\le\frac{\varepsilon}{8}.
\end{equation}
By Eq.~\eqref{eq:l1-bound},
Hoeffding's inequality gives
$$
\Pr\left[
\frac{1}{|\bar S_{\mathrm{cln}}|}
\sum_{(x,y)\in\bar S_{\mathrm{cln}}}
\mathbf 1\{\sign(q(x))\neq c^\star(x)\}
>
\frac{\varepsilon}{4}
\right]
\le
\exp\left(-\frac{|\bar S_{\mathrm{cln}}|\varepsilon^2}{32}\right)
\le
\frac{\delta}{2}.
$$
Thus, with probability at least $1-\delta/2$, the observed sample has contamination rate at most $\eta+\varepsilon/4$ relative to $\sign(q)$. By \Cref{fact:ptf-l2-approximation}, the degree $\ell$ PTF $\sign(q)$ admits a polynomial approximator of degree at most $\ell' = O(\ell^2 / \varepsilon^4)$ with squared ${\cal L}_2(\mathcal D)$ error at most $(\varepsilon/2)^2/1536^2$. The concentration and approximation bounds used in the proof of \Cref{thm:main-guarantee-restated} also hold with probability at least $1-\delta/2$, for target $\sign(q)$ and accuracy $\varepsilon/2$. A union bound therefore gives, with probability at least $1-\delta$,
$$
\err_{\mathcal D}(h,c^\star)
\le
\err_{\mathcal D}(h,\sign(q))
+
\Pr_{x\sim\mathcal D}[\sign(q(x))\neq c^\star(x)]
\le
\eta+\frac{\varepsilon}{4}
+\frac{\varepsilon}{2}
+\frac{\varepsilon}{8}
\le\eta+\varepsilon.
$$
The sample and runtime bounds follow from
\Cref{thm:main-guarantee-restated}. 
\end{proof}

As a separate consequence of \Cref{fact:ptf-l2-approximation}, we obtain a conversion from ${\cal L}_1$ to ${\cal L}_2$ polynomial approximation, which may be of independent interest.

\begin{lemma}[${\cal L}_1$ to ${\cal L}_2$ approximation]
\label{lemma:l1-to-l2}
Let $\mathcal D$ be the Gaussian distribution over $\R^d$ or the uniform distribution on $\{\pm1\}^d$. Let $c:X\to\{\pm1\}$ and suppose that, for some $\alpha\in(0,1)$, there is a polynomial $q$ of degree $\ell$ such that
$$
\E_{x\sim\mathcal D}|q(x)-c(x)|\le\alpha.
$$
Then there is a polynomial $p$ of degree $O(\ell^2/\alpha^2)$
such that
$$
\E_{x\sim\mathcal D}\bigl[(p(x)-c(x))^2\bigr]\le9\alpha.
$$
\end{lemma}

\begin{proof}
Since $\sign(q)$ and $c$ take values in $\{\pm1\}$, 
\begin{equation}\label{eq:ptf-upper-bound}
\E_{x\sim\mathcal D}\bigl[(\sign(q(x))-c(x))^2\bigr]
=
4\Pr_{x\sim\mathcal D}[\sign(q(x))\neq c(x)]
\le
4\E_{x\sim\mathcal D}|q(x)-c(x)|
\le 
4\alpha.
\end{equation}
By \Cref{fact:ptf-l2-approximation}, applied to $\sign(q)$ with $\beta=\alpha$, there is a polynomial $p$ of degree $O(\ell^2/\alpha^2)$ satisfying $\E_{x\sim\mathcal D}[(p(x)-\sign(q(x)))^2]\le\alpha$. By the triangle inequality,
\begin{align*}
\E_{x\sim\mathcal D}[(p(x)-c(x))^2]
&\le
\left(
\sqrt{\E_{x\sim\mathcal D}[(p(x)-\sign(q(x)))^2]}
+
\sqrt{\E_{x\sim\mathcal D}[(\sign(q(x))-c(x))^2]}
\right)^2
\\
&\le
\left(
\sqrt{\alpha}
+
2\sqrt{\Pr_{x\sim\mathcal D}[\sign(q(x))\neq c(x)]}
\right)^2
\tag{\Cref{fact:ptf-l2-approximation}}
\\
&\le
9\alpha.
\tag{Eq.~\eqref{eq:ptf-upper-bound}}
\end{align*}
\end{proof}

\subsection{Sharper bounds from noise sensitivity}
\label{sec:noise-sensitivity}

\Cref{cor:l1-robust-learning} shows that ${\cal L}_1$ polynomial approximation suffices for learning with bounded contamination under the Gaussian distribution or the uniform distribution on the Boolean hypercube. However, its runtime bounds do not match those for agnostic learning. We show that noise sensitivity bounds yield sharper runtimes that nearly match agnostic SQ lower bounds for several concept classes. Table~\ref{tab:runtime-applications} summarizes the resulting bounds. 

The key observation is that the proof of \Cref{thm:main-guarantee-restated} only requires two consequences of the ${\cal L}_2$ approximation assumption, which noise sensitivity bounds imply at a smaller polynomial degree.

\begin{definition}[Structured $\mathcal{L}_1$ approximation]
\label{def:joint-polynomial-approximation}
For $\alpha\in(0,1)$, a concept class $\mathcal C$ admits structured $\cal L_1$ polynomial approximators of degree $\ell$ and error $\alpha$ with respect to $\mathcal D$ if, for every $c\in\mathcal C$, there is a polynomial $q$ of degree at most $\ell$ such that
$$
\E_{x\sim\mathcal D}|q(x)-c(x)|\le\alpha
\qquad\text{and}\qquad
\E_{x\sim\mathcal D}|q(x)^2-1|\le\alpha.
$$
\end{definition}

The first condition controls approximation of the target concept. The second allows weighted iterative polynomial filtering to control the polynomial $q^2-1$. Both conditions must hold for the same polynomial $q$. The original $\mathcal L_2$ approximation assumption in Eq.~\eqref{eq:l2-approx} implies these conditions. If $\E_{x\sim\mathcal D}[(q(x)-c(x))^2]\le\alpha^2/9$, then Cauchy--Schwarz inequality gives
$$
\E_{x\sim\mathcal D}|q(x)-c(x)|
\le\frac{\alpha}{3}\le\alpha
\qquad\text{and}\qquad
\E_{x\sim\mathcal D}|q(x)^2-1|
\le
\frac{\alpha}{3}\left(2+\frac{\alpha}{3}\right)
\le\alpha.
$$

\begin{theorem}[Robust learning via structured $\mathcal{L}_1$ approximation]
\label{thm:joint-robust-learning}
Let $\mathcal D$ be an $A$-hypercontractive distribution, and let $\varepsilon,\delta\in(0,1)$. Suppose $\mathcal C$ admits structured $\mathcal{L}_1$ approximation of degree $\ell\ge1$ and error $\varepsilon/512$. Under the sample size bounds of \Cref{thm:main-guarantee-restated}, for a sufficiently large universal constant $C_\star$, Algorithm~\ref{fig:polynomial-algo} runs in time $\operatorname{poly}(|\bar S_{\mathrm{inp}}|,|S_{\mathrm{ref}}|,(d+1)^\ell)$ and outputs a randomized hypothesis $h$ such that, under $\eta$-bounded contamination, with probability at least $1-\delta$,
$$
\err_{\mathcal D}(h,c^\star)\le\eta+\varepsilon.
$$
\end{theorem}

The proof of \Cref{thm:joint-robust-learning} is in Appendix~\ref{sec:proof-of-joint-robust-learning}. We next show that noise sensitivity bounds imply structured $\mathcal L_1$ approximation at low degree. Noise sensitivity measures how often a Boolean function changes its value under a random perturbation of its input. We follow the notation of \cite{o2014analysis}.

\begin{definition}[Noise sensitivity, {\cite[Definition~2.43]{o2014analysis}}]
\label{def:noise-sensitivity}
For $\rho\in(0,1)$, a $\rho$-correlated pair $(X,Y)$ is defined as follows:
\begin{itemize}
    \item \emph{Gaussian distribution.} Let $X,Z$ be independent standard Gaussian vectors and set $Y:=\rho X+\sqrt{1-\rho^2}Z$.
    \item \emph{Boolean hypercube.} Let $X$ be uniform on $\{\pm1\}^d$. Independently for each coordinate, set $Y_i=-X_i$ with probability $(1-\rho)/2$ and $Y_i=X_i$ otherwise.
\end{itemize}
For a Boolean function $c:X\to\{\pm1\}$, its noise sensitivity at noise rate $\delta:=(1-\rho)/2$ is
$$
\operatorname{NS}_{\delta}[c]
:=
\Pr[c(X)\neq c(Y)].
$$
\end{definition}

The construction first smooths $c$ to obtain a bounded function close to $c$, then truncates its Hermite or Fourier expansion to obtain a polynomial that satisfies both approximation bounds.

\begin{lemma}[Existence of structured $\mathcal{L}_1$ approximators]
\label{lemma:joint-smoothing}
Let $\mathcal D$ be the standard Gaussian distribution over $\R^d$ or the uniform distribution on $\{\pm1\}^d$, and let $c:X\to\{\pm1\}$. For $\rho,\alpha\in(0,1)$, set $\delta:=(1-\rho)/2$ and suppose that $\operatorname{NS}_{\delta}[c]\le\frac{\alpha}{8}$. Then there is a polynomial $q$ of degree  $\ell=\left\lceil\frac1{1-\rho}\log\frac4\alpha\right\rceil$ such that
\begin{equation}\label{eq:goals}
\E_{x\sim\mathcal D}|q(x)-c(x)|\le\frac{\alpha}{2},
\qquad\text{and}\qquad
\E_{x\sim\mathcal D}|q(x)^2-1|\le\alpha.
\end{equation}
\end{lemma}

The proof of \Cref{lemma:joint-smoothing} is in Appendix~\ref{sec:proof-of-joint-smoothing}. For the uniform distribution on the Boolean hypercube, we use known noise sensitivity bounds for halfspaces and Boolean functions of $k$ halfspaces \cite{KOSOS04}, intersections of $k$ halfspaces \cite{kane2014average}, and monotone functions \cite{bshouty1996fourier,o2014analysis}. Together with \Cref{lemma:joint-smoothing,thm:joint-robust-learning}, these bounds give the corresponding runtimes in Table~\ref{tab:runtime-applications}.  For the Gaussian distribution, we use Gaussian surface area to obtain the required noise sensitivity bounds.

\begin{definition}[Gaussian surface area, {\cite[Definition~1]{klivans2008learning}}]
\label{def:gaussian-surface-area}
For a Borel set $K\subseteq\R^d$, let $K_r:=\{x\in\R^d:\operatorname{dist}(x,K)\le r\}$. The Gaussian surface area of $K$ is
$$
\operatorname{GSA}(K)
:=
\liminf_{r\downarrow0}
\frac{\Pr_{x\sim\mathcal N(0,I_d)}[x\in K_r\setminus K]}{r}.
$$
For a Boolean function $c:\R^d\to\{\pm1\}$, define $\operatorname{GSA}(c):=\operatorname{GSA}(\{x:c(x)=1\})$.
\end{definition}

A standard bound relating Gaussian surface area to noise sensitivity gives the following corollary.

\begin{corollary}[Classes with bounded Gaussian surface area]
\label{cor:gaussian-surface-area-learning}
Let $\mathcal D$ be the standard Gaussian distribution over $\R^d$, and suppose $\operatorname{GSA}(c)\le\Gamma$ for every $c\in\mathcal C$, where $\Gamma\ge1$. For every $\alpha\in(0,1)$, $\mathcal C$ admits structured $\mathcal L_1$ approximation with error $\alpha$ and degree $\ell=O\left(\frac{\Gamma^2}{\alpha^2}\log\frac4\alpha\right)$. For $\alpha=\varepsilon/512$, Algorithm~\ref{fig:polynomial-algo}, with degree $\ell$ and the sample sizes from \Cref{thm:joint-robust-learning}, achieves error $\eta+\varepsilon$ under $\eta$-bounded contamination with probability at least $1-\delta$.
\end{corollary}

The proof of \Cref{cor:gaussian-surface-area-learning} can be found in Appendix~\ref{sec:proof-of-gsa-learning}.

\newpage
\bibliographystyle{alpha}
\bibliography{references}

\appendix 

\section{
\texorpdfstring
{Robust learning with access to an ERM oracle: proofs in Section~\ref{sec:learning-with-erm}}
{Fixed-distribution learning with ERM oracle}
}\label{app:proofs-of-erm}

For fixed $(x,y)$ and $f_{c,\bar\mu}(x)=\frac{c(x)\bar\mu(x)-1}{1+|\bar\mu(x)|}$, the probability in the first term of the corruption certificate from \Cref{def:erm-empirical-certificate} can be written as follows:
\begin{align}
\Pr[c(x)\neq y\ \text{or}\ \operatorname{Rad}(f_{c,\bar\mu}(x))=1]
&=
1-\Pr[c(x)=y\ \text{and}\ \operatorname{Rad}(f_{c,\bar\mu}(x))=-1]
\notag\\
&=
\begin{cases}
1, & y\neq c(x)\\[6pt]
\displaystyle
\frac{|\bar\mu(x)|+y\bar\mu(x)}{2(1+|\bar\mu(x)|)},
& y=c(x)
\end{cases}
\notag\\[6pt]
&=
1-\frac{1+yc(x)}{2+|\bar\mu(x)|+y\bar\mu(x)}
\notag\\
&=
\frac{1}{2}+
\frac{|\bar\mu(x)|+y\bar\mu(x)-2yc(x)}
{2(2+|\bar\mu(x)|+y\bar\mu(x))}.
\label{eq:erm-first-term}
\end{align}
Similarly, the probability in the second term can be written as
\begin{equation}\label{eq:erm-second-term}
\Pr[\operatorname{Rad}(f_{c,\bar\mu}(x))=1]
=
\frac{1+f_{c,\bar\mu}(x)}2
=
\frac{1}{2}+\frac{c(x)\bar\mu(x)-1}{2(1+|\bar\mu(x)|)}.
\end{equation}
Averaging Eq.~\eqref{eq:erm-first-term} and Eq.~\eqref{eq:erm-second-term} over the input and reference samples, respectively, gives:
\begin{equation}\label{eq:erm-certificate-again}
\begin{aligned}
\eta(c,\bar\mu;\bar S_{\mathrm{inp}},S_{\mathrm{ref}})
&=
\frac1{2|\bar S_{\mathrm{inp}}|}
\sum_{(x,y)\in\bar S_{\mathrm{inp}}}
\frac{|\bar\mu(x)|+y\bar\mu(x)-2yc(x)}
{2+|\bar\mu(x)|+y\bar\mu(x)}
-
\frac1{2|S_{\mathrm{ref}}|}
\sum_{x\in S_{\mathrm{ref}}}
\frac{c(x)\bar\mu(x) - 1}{1+|\bar\mu(x)|}.
\end{aligned}
\end{equation}

\subsection{
\texorpdfstring
{Proof of \Cref{lemma:erm-linearization}}
{Proof of the linearization}
}
\label{sec:erm-linearization}

\begin{proof}
First, by definition of the randomized $0$-$1$ error
(\Cref{def:linear-error}), the $c$-dependent part of
$\widehat{\err}_{S_{\mathrm{ref}}}(h_t,c)$ is
\begin{equation}\label{eq:erm-part-1}
-\frac1{2|S_{\mathrm{ref}}|}
\sum_{x\in S_{\mathrm{ref}}}h_t(x)c(x).
\end{equation}
The $c$-dependent part of the input-sample term in Eq.~\eqref{eq:erm-certificate-again}, after we multiply the numerator and denominator by $2+|\bar\mu_t(x)|-y\bar\mu_t(x)$, is
\begin{equation}\label{eq:erm-part-2}
-\frac1{2|\bar S_{\mathrm{inp}}|}
\sum_{(x,y)\in\bar S_{\mathrm{inp}}}
\frac{2yc(x)}{2+|\bar\mu_t(x)|+y\bar\mu_t(x)}
=
-\frac1{4|\bar S_{\mathrm{inp}}|}
\sum_{(x,y)\in\bar S_{\mathrm{inp}}}
\frac{y(2+|\bar\mu_t(x)|)-\bar\mu_t(x)}
{1+|\bar\mu_t(x)|}c(x).
\end{equation}
The $c$-dependent part of the reference-sample term
in Eq.~\eqref{eq:erm-certificate-again} is
\begin{equation}\label{eq:erm-part-3}
-\frac1{2|S_{\mathrm{ref}}|}
\sum_{x\in S_{\mathrm{ref}}}
\frac{\bar\mu_t(x)}{1+|\bar\mu_t(x)|}c(x).
\end{equation}
Since
$U_t(c)=
\widehat{\err}_{S_{\mathrm{ref}}}(h_t,c)
-\eta(c,\bar\mu_t;\bar S_{\mathrm{inp}},S_{\mathrm{ref}})$,
the $c$-dependent part of $U_t(c)$ is obtained by subtracting
Eq.~\eqref{eq:erm-part-2} and~\eqref{eq:erm-part-3}
from Eq.~\eqref{eq:erm-part-1}:
$$
\sum_{(x,y)\in\bar S_{\mathrm{inp}}}
\frac{1}{4|\bar S_{\mathrm{inp}}|}
\left(\frac{
y\bigl(2+|\bar\mu_t(x)|\bigr)-\bar\mu_t(x)
}{
1+|\bar\mu_t(x)|
}\right) c(x)
+
\sum_{x\in S_{\mathrm{ref}}}
\frac1{2|S_{\mathrm{ref}}|}
\left(
\frac{\bar\mu_t(x)}{1+|\bar\mu_t(x)|}-h_t(x)
\right)c(x).
$$
This equals $\langle g_t,c[S]\rangle$ for the vector $g_t$
specified in the lemma.
All remaining terms are independent of $c$ and are absorbed
into $C_t$. Hence
$$
U_t(c)=C_t+\langle g_t,c[S]\rangle.
$$
\end{proof}

\subsection{
\texorpdfstring
{Proof of \Cref{lemma:erm-gradient-bound}}
{Proof of the gradient bound}
}
\label{sec:erm-gradient-bound}

\begin{proof}

By the triangle inequality, the coefficients corresponding to the input-sample and reference-sample satisfy, respectively,
\begin{align}
&\left|
\frac{
y\bigl(2+|\bar\mu_t(x)|\bigr)-\bar\mu_t(x)
}{
1+|\bar\mu_t(x)|
}
\right|
\le
\frac{
2+2|\bar\mu_t(x)|
}{
1+|\bar\mu_t(x)|
}
=
2,
\label{eq:erm-input-coordinate-bound}
\\
&\left|
\frac{\bar\mu_t(x)}{1+|\bar\mu_t(x)|}-h_t(x)
\right|
\le
\frac{|\bar\mu_t(x)|}{1+|\bar\mu_t(x)|}+1
\le
\frac32.
\label{eq:erm-reference-coordinate-bound}
\end{align}
Moreover, the sum of the squared coordinates of $g_t$ gives its squared Euclidean norm:
$$
\|g_t\|_2^2
=
\frac{1}{16|\bar S_{\mathrm{inp}}|^2}
\sum_{(x,y)\in\bar S_{\mathrm{inp}}}
\left(
\frac{
y\bigl(2+|\bar\mu_t(x)|\bigr)-\bar\mu_t(x)
}{
1+|\bar\mu_t(x)|
}
\right)^2
+
\frac1{4|S_{\mathrm{ref}}|^2}
\sum_{x\in S_{\mathrm{ref}}}
\left(
\frac{\bar\mu_t(x)}{1+|\bar\mu_t(x)|}-h_t(x)
\right)^2.
$$
Therefore, by the pointwise bounds in Eq.~\eqref{eq:erm-input-coordinate-bound} and \eqref{eq:erm-reference-coordinate-bound},
$$
\|g_t\|_2^2 \le \frac{1}{4|\bar S_{\mathrm{inp}}|} + \frac{9}{16|S_{\mathrm{ref}}|}.
$$
\end{proof}

\subsection{
\texorpdfstring
{Proof of \Cref{lemma:erm-error-guarantee}}
{Proof of error guarantee}
}\label{sec:erm-error-guarantee}

\begin{proof}

We prove the claim in two parts.

\noindent\textbf{Part 1: Small loss at the current iterate.} We first show that, for every $t\in[T]$,
\begin{equation}
\label{eq:erm-self-certification}
U_t(\bar\mu_t[S]) 
\le
\frac{\varepsilon}{8}.
\end{equation}
We extend $\eta$ affinely in its first argument through Eq.~\eqref{eq:erm-certificate-again}. Because $U_t$ is affine on $\operatorname{conv}(\mathcal C)$ and agrees on $\mathcal C$ with the original loss,
$$
U_t(\bar\mu_t[S])
=
\widehat{\operatorname{err}}_{S_{\mathrm{ref}}}
(h_t,\bar\mu_t)
-
\eta
\left(
\bar\mu_t,\bar\mu_t;
\bar S_{\mathrm{inp}},S_{\mathrm{ref}}
\right).
$$
The explicit formula for the certificate from Eq.~\eqref{eq:erm-certificate-again} gives
\begin{align*}
&
\eta
\left(
\bar\mu_t,\bar\mu_t;
\bar S_{\mathrm{inp}},S_{\mathrm{ref}}
\right)
-
\widehat{\operatorname{err}}_{S_{\mathrm{ref}}}
\left(
\operatorname{sign}(\bar\mu_t),\bar\mu_t
\right)
\\
&=
\frac{1}{4|\bar S_{\mathrm{inp}}|}
\sum_{(x,y)\in\bar S_{\mathrm{inp}}}
\left(
|\bar\mu_t(x)|-y\bar\mu_t(x)
\right)
+
\frac1{2|S_{\mathrm{ref}}|}
\sum_{x\in S_{\mathrm{ref}}}
\left(
\operatorname{sign}(\bar\mu_t(x))\bar\mu_t(x)
-
|\bar\mu_t(x)|
\right)
\\
&=
\frac{1}{4|\bar S_{\mathrm{inp}}|}
\sum_{(x,y)\in\bar S_{\mathrm{inp}}}
\left(
|\bar\mu_t(x)|-y\bar\mu_t(x)
\right)
\tag{$\operatorname{sign}(\bar\mu_t)\bar\mu_t=|\bar\mu_t|$}
\\
&\ge 0.
\tag{$y\bar\mu_t\le |\bar\mu_t|$ for $y\in\{\pm1\}$}
\end{align*}
The following pointwise approximation was proved in
\cite[Claim~8.3]{Blanc26}: for every $\mu\in[-1,1]$,
\begin{equation}\label{eq:maj-approx}
0
\le
\bigl(\operatorname{sign}(\mu)-\operatorname{maj}_k(\mu)\bigr)\mu
\le
\frac{2}{\sqrt{k}}.
\end{equation}
Since $h_t=\operatorname{maj}_k(\bar\mu_t)$, this estimate gives
\begin{align*}
&
\widehat{\operatorname{err}}_{S_{\mathrm{ref}}}(h_t,\bar\mu_t)
-
\widehat{\operatorname{err}}_{S_{\mathrm{ref}}}
(\operatorname{sign}(\bar\mu_t),\bar\mu_t)
\\
&=
\frac1{2|S_{\mathrm{ref}}|}
\sum_{x\in S_{\mathrm{ref}}}
\left(
\operatorname{sign}(\bar\mu_t(x))-\operatorname{maj}_k(\bar\mu_t(x))
\right)\bar\mu_t(x)
\\
&\le
\frac1{\sqrt{k}}
\tag{Eq.~\eqref{eq:maj-approx}}
\\
&\le
\frac{\varepsilon}{8}.
\tag{$k=\lceil64/\varepsilon^2\rceil$}
\end{align*}
The two bounds imply Eq.~\eqref{eq:erm-self-certification}:
$$
\begin{aligned}
U_t(\bar\mu_t[S])
&=
\widehat{\operatorname{err}}_{S_{\mathrm{ref}}}
(h_t,\bar\mu_t)
-
\bigl(\eta
\left(
\bar\mu_t,\bar\mu_t;
\bar S_{\mathrm{inp}},S_{\mathrm{ref}}
\right) 
-
\widehat{\operatorname{err}}_{S_{\mathrm{ref}}}
(\sign(\bar\mu_t),\bar\mu_t)\bigr)
-
\widehat{\operatorname{err}}_{S_{\mathrm{ref}}}
(\sign(\bar\mu_t),\bar\mu_t) 
\\
&\le 
\widehat{\operatorname{err}}_{S_{\mathrm{ref}}}(h_t,\bar\mu_t) 
-
\widehat{\operatorname{err}}_{S_{\mathrm{ref}}}(\sign(\bar\mu_t),\bar\mu_t) 
\le \frac{\varepsilon}{8}.
\end{aligned}
$$

\noindent\textbf{Part 2: No-regret.}
For every $c\in\mathcal C$, the affine extension agrees with
the original loss on $\mathcal C$, so
$$
U_t(c)
=
\widehat{\operatorname{err}}_{S_{\mathrm{ref}}}(h_t,c)
-
\eta
\left(
c,\bar\mu_t;
\bar S_{\mathrm{inp}},S_{\mathrm{ref}}
\right).
$$
By the definition of the affine extension in Eq.~\eqref{eq:erm-affine-loss} and \Cref{lemma:erm-linearization}, for every
$\bar\mu\in\operatorname{conv}(\mathcal C)$,
\begin{equation}\label{eq:affine-extension-application}
U_t(\bar\mu[S])-U_t(\bar\mu_t[S])
=
\left\langle
g_t,\bar\mu[S]-\bar\mu_t[S]
\right\rangle.
\end{equation}
Since $c\in\mathcal C\subseteq\operatorname{conv}(\mathcal C)$, the
choice $\bar\mu:=c$ is admissible and $\bar\mu[S]=c[S]\in K$. To apply \Cref{lemma:erm-online-frank-wolfe}, we need one additional technical step: we introduce the normalized coordinates
\[
\widetilde u_t[S]:=
\left(
\frac{\bar\mu_t[\bar S_{\mathrm{inp}}]}{\sqrt{2|\bar S_{\mathrm{inp}}|}},
\frac{\bar\mu_t[S_{\mathrm{ref}}]}{\sqrt{2|S_{\mathrm{ref}}|}}
\right) \qquad\text{and}\qquad
\widetilde u[S]:=
\left(
\frac{\bar\mu[\bar S_{\mathrm{inp}}]}{\sqrt{2|\bar S_{\mathrm{inp}}|}},
\frac{\bar\mu[S_{\mathrm{ref}}]}{\sqrt{2|S_{\mathrm{ref}}|}}
\right).
\]
Let $\widetilde K$ be the set corresponding to $K$ from Eq.~\eqref{eq:erm-convex-set} in normalized coordinates:
\[
\widetilde K:=
\operatorname{conv}\left\{
\left(
\frac{c[\bar S_{\mathrm{inp}}]}{\sqrt{2|\bar S_{\mathrm{inp}}|}},
\frac{c[S_{\mathrm{ref}}]}{\sqrt{2|S_{\mathrm{ref}}|}}
\right)
:c\in\mathcal C
\right\}.
\]
The set $\widetilde K$ is a linear image of $K$, so it is compact and convex by \Cref{lemma:erm-feasible-set-geometry}. Its diameter satisfies
\[
\operatorname{diam}(\widetilde K)^2
\le
\frac{4|\bar S_{\mathrm{inp}}|}{2|\bar S_{\mathrm{inp}}|}
+
\frac{4|S_{\mathrm{ref}}|}{2|S_{\mathrm{ref}}|}
=4.
\]
We also rescale the gradient coordinates so that the inner products remain unchanged under normalization. Write $g_t=(g_{t,\mathrm{inp}},g_{t,\mathrm{ref}})$ and define
\[
\widetilde g_t:=
\left(
\sqrt{2|\bar S_{\mathrm{inp}}|}\,g_{t,\mathrm{inp}},
\sqrt{2|S_{\mathrm{ref}}|}\,g_{t,\mathrm{ref}}
\right).
\]
The coordinate bounds in Eq.~\eqref{eq:erm-input-coordinate-bound} and Eq.~\eqref{eq:erm-reference-coordinate-bound} give
\[
\|\widetilde g_t\|_2^2
\le
2|\bar S_{\mathrm{inp}}|\cdot\frac1{4|\bar S_{\mathrm{inp}}|}
+
2|S_{\mathrm{ref}}|\cdot\frac9{16|S_{\mathrm{ref}}|}
=
\frac{13}{8}
<4.
\]
In these coordinates, the objectives in Algorithm~\ref{fig:erm-oracle-algorithm} take the form, for any $\widetilde v[S]\in\widetilde K$,
\[
\widetilde F_t(\widetilde v[S])
:=
\alpha\sum_{r=1}^{t-1}\langle\widetilde g_r,\widetilde v[S]\rangle
-
\|\widetilde v[S]-\widetilde u_1[S]\|_2^2.
\]
For every $z\in\operatorname{conv}(\mathcal C)$, the chain rule gives
\[
\left\langle
\nabla\widetilde F_t(\widetilde u_t[S]),
\left(
\frac{z[\bar S_{\mathrm{inp}}]}{\sqrt{2|\bar S_{\mathrm{inp}}|}},
\frac{z[S_{\mathrm{ref}}]}{\sqrt{2|S_{\mathrm{ref}}|}}
\right)
\right\rangle
=
\left\langle\nabla F_t(\bar\mu_t[S]),z[S]\right\rangle.
\]
Thus, the same ERM call satisfies the linear optimization guarantee on $\widetilde K$ with additive accuracy $\gamma_t/25$. We therefore apply \Cref{lemma:erm-online-frank-wolfe} with $D_2=G_2=2$. These choices give the algorithm's step size $\alpha=1/(2T^{3/4})$ and accuracy $D_2^2\gamma_t/100=\gamma_t/25$. Therefore,
\begin{align*}
\frac1T\sum_{t=1}^T U_t(c)
&=
\frac1T\sum_{t=1}^T U_t(c[S])
\tag{Eq.~\eqref{eq:erm-affine-loss}}
\\
&=
\frac1T\sum_{t=1}^T U_t(\bar\mu[S])
\tag{$\bar\mu=c$}
\\
&=
\frac1T\sum_{t=1}^T U_t(\bar\mu_t[S])
+
\frac1T\sum_{t=1}^T
\left\langle g_t,\bar\mu[S]-\bar\mu_t[S]\right\rangle
\tag{Eq.~\eqref{eq:affine-extension-application}}
\\
&=
\frac1T\sum_{t=1}^T U_t(\bar\mu_t[S])
+
\frac1T\sum_{t=1}^T
\left\langle\widetilde g_t,\widetilde u[S]-\widetilde u_t[S]\right\rangle
\tag{normalization}
\\
&\le
\frac1T\sum_{t=1}^T U_t(\bar\mu_t[S])
+
36T^{-1/4}
\tag{\Cref{lemma:erm-online-frank-wolfe}}
\\
&\le
\frac{\varepsilon}{8}
+
36T^{-1/4}
\tag{Eq.~\eqref{eq:erm-self-certification}}
\\
&\le
\frac{\varepsilon}{4}.
\tag{$T\ge(300/\varepsilon)^4$}
\end{align*}
Moreover, since $\bar h_T=T^{-1}\sum_{t=1}^T h_t$ and the empirical error is affine in its first argument,
\begin{equation}
\label{eq:erm-average-loss}
\begin{aligned}
\widehat{\operatorname{err}}_{S_{\mathrm{ref}}}(\bar h_T,c)
=
\frac1T\sum_{t=1}^T
\widehat{\operatorname{err}}_{S_{\mathrm{ref}}}(h_t,c)
&=
\frac1T\sum_{t=1}^T
\left[
U_t(c)
+
\eta\left(c,\bar\mu_t;\bar S_{\mathrm{inp}},S_{\mathrm{ref}}\right)
\right]
\\
&\le
\frac1T\sum_{t=1}^T
\eta\left(c,\bar\mu_t;\bar S_{\mathrm{inp}},S_{\mathrm{ref}}\right)
+
\frac{\varepsilon}{4}.
\end{aligned}
\end{equation}
Since $\bar\mu_t\in\operatorname{conv}(\mathcal C)$ for every $t\in[T]$,
\begin{equation}
\label{eq:erm-certificate-upper-bound}
\frac1T\sum_{t=1}^T
\eta\left(c,\bar\mu_t;\bar S_{\mathrm{inp}},S_{\mathrm{ref}}\right)
\le
\sup_{\bar\mu\in\operatorname{conv}(\mathcal C)}
\eta\left(c,\bar\mu;\bar S_{\mathrm{inp}},S_{\mathrm{ref}}\right).
\end{equation}
Eq.~\eqref{eq:erm-average-loss} and Eq.~\eqref{eq:erm-certificate-upper-bound} together complete the proof:
\[
\widehat{\operatorname{err}}_{S_{\mathrm{ref}}}(\bar h_T,c)
\le
\sup_{\bar\mu\in\operatorname{conv}(\mathcal C)}
\eta\left(c,\bar\mu;\bar S_{\mathrm{inp}},S_{\mathrm{ref}}\right)
+
\frac{\varepsilon}{4}.
\]
\end{proof}

\subsection{
\texorpdfstring
{Proof of \Cref{lemma:erm-certificate-validity}}
{Proof of ERM certificate validity}
}\label{sec:erm-certificate-guarantee}

\begin{proof}

The proof consists of two parts.

\noindent\textbf{Part 1: Population certificate.}
The uniform convergence analysis of
\cite[Facts~4.7--4.9 and the proof of Theorem~11]{Blanc26},
applied to 
$
x
\mapsto
\frac{c^\star(x)\bar\mu(x)-1}{1+|\bar\mu(x)|}
$
gives, with probability at least $1-\delta/2$,
$$
\sup_{\bar\mu\in\operatorname{conv}(\mathcal C)}
\left|
\frac{1}{|S_{\mathrm{cln}}|}
\sum_{x\in S_{\mathrm{cln}}}
\frac{c^\star(x)\bar\mu(x)-1}{1+|\bar\mu(x)|}
-
\E_{x \sim \mathcal{D}}
\left[
\frac{c^\star(x)\bar\mu(x)-1}{1+|\bar\mu(x)|}
\right]
\right|
\le
\frac{\varepsilon}{2}.
$$
Fix $\bar\mu\in\operatorname{conv}(\mathcal C)$. For $c=c^\star$, the
contribution of $(x,y)$ to the first average in the certificate is
$$
\begin{aligned}
&\frac14
\left(
1
+
\frac{c^\star(x)\bar\mu(x)-1}{1+|\bar\mu(x)|}
+
c^\star(x)y
\left(
\frac{c^\star(x)\bar\mu(x)-1}{1+|\bar\mu(x)|}
-
1
\right)
\right)
=
\begin{cases}
\displaystyle
\frac{1}{2}
\frac{c^\star(x)\bar\mu(x)-1}{1+|\bar\mu(x)|}
& y=c^\star(x)
\\[0.75em]
\displaystyle
\frac{1}{2}
& y\ne c^\star(x)
\end{cases}
\end{aligned}
$$
This contribution lies in $[-1/2,1/2]$, so replacing at most
$\eta n$ examples increases the first average by at most $\eta$.
Let $\eta(c,\bar\mu;\bar S_{\mathrm{inp}},\mathcal{D})$ denote the population version of the certificate, with the average over $S_{\mathrm{ref}}$ replaced by the corresponding expectation under $\mathcal{D}$. Consequently,
\begin{equation}\label{eq:certificate-population-bound}
\eta(c^\star,\bar\mu;\bar S_{\mathrm{inp}},\mathcal{D})
\le
\eta
+
\frac{1}{2}
\left(
\frac{1}{|S_{\mathrm{cln}}|}
\sum_{x\in S_{\mathrm{cln}}}
\frac{c^\star(x)\bar\mu(x)-1}{1+|\bar\mu(x)|}
-
\E_{x \sim \mathcal{D}}
\left[
\frac{c^\star(x)\bar\mu(x)-1}{1+|\bar\mu(x)|}
\right]
\right)
\le
\eta+\frac{\varepsilon}{4}.
\end{equation}
Since Eq.~\eqref{eq:certificate-population-bound} holds for every
$\bar\mu\in\operatorname{conv}(\mathcal C)$,
\begin{equation}
\label{eq:erm-population-certificate-validity}
\sup_{\bar\mu\in\operatorname{conv}(\mathcal C)}
\eta(c^\star,\bar\mu;\bar S_{\mathrm{inp}},\mathcal{D})
\le
\eta+\frac{\varepsilon}{4}.
\end{equation}

\noindent\textbf{Part 2: Reference-sample approximation.}
The same uniform convergence analysis gives, with probability at least
$1-\delta/2$,
$$
\sup_{\bar\mu\in\operatorname{conv}(\mathcal C)}
\left|
\frac1{|S_{\mathrm{ref}}|}
\sum_{x\in S_{\mathrm{ref}}}
\frac{c^\star(x)\bar\mu(x)-1}{1+|\bar\mu(x)|}
-
\E_{x \sim \mathcal{D}}
\left[
\frac{c^\star(x)\bar\mu(x)-1}{1+|\bar\mu(x)|}
\right]
\right|
\le
\frac{\varepsilon}{2}.
$$
Therefore,
\begin{equation}
\label{eq:erm-certificate-generalization}
\begin{aligned}
&\sup_{\bar\mu\in\operatorname{conv}(\mathcal C)}
\left|
\eta(c^\star,\bar\mu;\bar S_{\mathrm{inp}},S_{\mathrm{ref}})
-
\eta(c^\star,\bar\mu;\bar S_{\mathrm{inp}},\mathcal{D})
\right|
\\
&=
\frac{1}{2}
\sup_{\bar\mu\in\operatorname{conv}(\mathcal C)}
\left|
\frac1{|S_{\mathrm{ref}}|}
\sum_{x\in S_{\mathrm{ref}}}
\frac{c^\star(x)\bar\mu(x)-1}{1+|\bar\mu(x)|}
-
\E_{x \sim \mathcal{D}}
\left[
\frac{c^\star(x)\bar\mu(x)-1}{1+|\bar\mu(x)|}
\right]
\right|
\le
\frac{\varepsilon}{4}.
\end{aligned}
\end{equation}
On the intersection of these two events,
\begin{align*}
\sup_{\bar\mu\in\operatorname{conv}(\mathcal C)}
\eta(c^\star,\bar\mu;\bar S_{\mathrm{inp}},S_{\mathrm{ref}})
&\le
\sup_{\bar\mu\in\operatorname{conv}(\mathcal C)}
\eta(c^\star,\bar\mu;\bar S_{\mathrm{inp}},\mathcal{D})
+
\frac{\varepsilon}{4}
\le
\eta+\frac{\varepsilon}{4}+\frac{\varepsilon}{4}
=
\eta+\frac{\varepsilon}{2}.
\end{align*}
A union bound completes the proof.

\end{proof}

\section{
\texorpdfstring
{Robust learning with low-degree polynomials: proofs in Section~\ref{sec:learning-with-polynomials}}
{Fixed-distribution learning with polynomials}
}

\subsection{
\texorpdfstring
{Proof of \Cref{lemma:linearization}}
{Proof of the linearization lemma}
}
\label{sec:linearization}

\begin{proof}
The same calculation as in the proof of \Cref{lemma:erm-linearization} applies with $c=q$ and $\bar\mu_t=m_t$. The input-sample sum is now over $\bar S_{\mathrm{filt},t}$, but its normalization remains $|\bar S_{\mathrm{inp}}|$. For $q(x)=\langle w,\psi_\ell(x)\rangle$, the $q$-dependent terms equal $\langle w,g_t\rangle$, where $g_t$ is defined in Eq.~\eqref{def:g}. The remaining terms form the constant $C_t$.
\end{proof}

\subsection{
\texorpdfstring
{Proof of \Cref{lemma:gradient-bound}}
{Proof of gradient bound}
}
\label{sec:gradient-bound}

\begin{proof}
We project the feature vectors onto the range of $M_\ell$ and
work in orthonormal coordinates on this range, where $M_\ell$
is invertible. For every unit vector $u$ in these coordinates,
the definition of $\psi_\ell$ gives
\[
\frac1{|S_{\mathrm{ref}}|}
\sum_{x\in S_{\mathrm{ref}}}
\langle u,\psi_\ell(x)\rangle^2
=
u^\top M_\ell^{-1/2}M_\ell M_\ell^{-1/2}u
=
1.
\]
By Cauchy--Schwarz, the polynomial
$\gamma\langle u,\psi_\ell(x)\rangle$ has reference absolute
mean at most $\gamma$ and reference squared mean at most
$\gamma^2\le\beta$. Its degree is at most $\ell$, so it belongs
to $\mathcal P$. Consequently, line~\ref{line:weighted-ipf-initial} of \Cref{fig:weighted-ipf} gives, for every $t\in[T]$ and $x\in S_{\mathrm{filt},t}$,
\begin{equation}\label{eq:gradient-retained-features}
\gamma\|\psi_\ell(x)\|_2
=
\sup_{\|u\|_2=1}
\bigl|\gamma\langle u,\psi_\ell(x)\rangle\bigr|
\le B.
\end{equation}
All subsequent filtering steps only remove points, so this
bound holds on every $S_{\mathrm{filt},t}$. Moreover,
\[
\frac1{|S_{\mathrm{ref}}|}
\sum_{x\in S_{\mathrm{ref}}}\|\psi_\ell(x)\|_2^2
=
\operatorname{tr}
\left(M_\ell^{-1/2}M_\ell M_\ell^{-1/2}\right)
=
\operatorname{rank}(M_\ell)
\le N
\le 
(d+1)^{\ell}.
\]
Another application of Cauchy--Schwarz inequality therefore gives
\begin{equation}\label{eq:gradient-reference-features}
\frac1{|S_{\mathrm{ref}}|}
\sum_{x\in S_{\mathrm{ref}}}\|\psi_\ell(x)\|_2
\le
\sqrt{N}
\le
(d+1)^{\ell/2}.
\end{equation}
Since $m_t=q_t$ and $h_t=\sign(q_t)$, the coefficients
in Eq.~\eqref{def:g} satisfy
\begin{equation}\label{eq:gradient-coefficients}
\begin{aligned}
&\left|
\frac{y(2+|m_t(x)|)-m_t(x)}
{1+|m_t(x)|}
\right|
\le2,\\
&\left|
\frac{m_t(x)}{1+|m_t(x)|}-h_t(x)
\right|
\le1.
\end{aligned}
\end{equation}
Therefore,
\begin{align*}
\|g_t\|_2
&\le
\frac1{2|\bar S_{\mathrm{inp}}|}
\sum_{(x,y)\in\bar S_{\mathrm{filt},t}}
\|\psi_\ell(x)\|_2
+
\frac1{2|S_{\mathrm{ref}}|}
\sum_{x\in S_{\mathrm{ref}}}\|\psi_\ell(x)\|_2
\tag{\text{triangle ineq. and Eq.~\eqref{eq:gradient-coefficients}}}
\\
&\le
\frac{|\bar S_{\mathrm{filt},t}|}
{2|\bar S_{\mathrm{inp}}|}
\frac{B}{\gamma}
+
\frac{\sqrt N}{2}
\tag{\text{Eq.~\eqref{eq:gradient-retained-features} and~\eqref{eq:gradient-reference-features}}}
\\
&\le
\frac{1}{2}\left(\frac{B}{\gamma}+\sqrt N\right)
\tag{$|\bar S_{\mathrm{filt},t}|\le|\bar S_{\mathrm{inp}}|$}
\\
&\le
\frac{(CA^2(d+1))^{\ell}}{\varepsilon^{3/2}}.
\tag{Algorithm \ref{fig:weighted-ipf}; $N\le(d+1)^{\ell}$}
\end{align*}
\end{proof}

\subsection{
\texorpdfstring
{Proof of \Cref{lemma:error-guarantee}}
{Proof of the error-guarantee}
}
\label{sec:proof-of-error-guarantee}

\begin{proof}

The proof consists of two essential steps.

\noindent\textbf{Part 1: Small loss at the current iterate.} We first show that, for every $t\in[T]$,
\begin{equation}\label{eq:self-cert}
U_t^{\mathrm{poly}}(q_t)
=
\widehat{\err}_{S_{\mathrm{ref}}}(h_t,q_t)
-
\eta^{\mathrm{poly}}
\left(q_t,m_t;\bar S_{\mathrm{filt},t},S_{\mathrm{ref}}\right)
\le0.
\end{equation}
Recall that $m_t=q_t$ and $h_t=\sign(q_t)$. We use $(q_t(x)^2-1)/(1+|q_t(x)|)=|q_t(x)|-1$ in \Cref{def:polynomial-certificate} to obtain
\begin{align*}
&
\eta^{\mathrm{poly}}
\left(q_t,q_t;\bar S_{\mathrm{filt},t},S_{\mathrm{ref}}\right)
-
\widehat{\err}_{S_{\mathrm{ref}}}(h_t,q_t)
\\
&=
\frac1{2|\bar S_{\mathrm{inp}}|}
\sum_{(x,y)\in\bar S_{\mathrm{filt},t}}
\frac{|q_t(x)|-yq_t(x)}
{2+|q_t(x)|+yq_t(x)}
+
\frac1{2|S_{\mathrm{ref}}|}
\sum_{x\in S_{\mathrm{ref}}}
\left(h_t(x)q_t(x)-|q_t(x)|\right)
\\
&=
\frac1{4|\bar S_{\mathrm{inp}}|}
\sum_{(x,y)\in\bar S_{\mathrm{filt},t}}
\left(|q_t(x)|-yq_t(x)\right)
+
\frac1{2|S_{\mathrm{ref}}|}
\sum_{x\in S_{\mathrm{ref}}}
\left(h_t(x)q_t(x)-|q_t(x)|\right)
\tag{$yq_t(x)\in\{\pm|q_t(x)|\}$}
\\
&=
\frac1{4|\bar S_{\mathrm{inp}}|}
\sum_{(x,y)\in\bar S_{\mathrm{filt},t}}
\left(|q_t(x)|-yq_t(x)\right)
\tag{$h_t(x)q_t(x)=|q_t(x)|$}
\\
&\ge0.
\tag{$yq_t(x)\le|q_t(x)|$}
\end{align*}

\noindent\textbf{Part 2: No-regret.} Since $K=\{w:\|w\|_2\le2\}$, its Euclidean diameter is $4$.
We apply \Cref{thm:zinkevich-projected-gradient} with the gradient bound $\|g_t\|_2\le\frac{(CA^2(d+1))^{\ell}}{\varepsilon^{3/2}}$ from \Cref{lemma:gradient-bound} and learning rate $\alpha:=\frac{4\varepsilon^{3/2}}{(CA^2(d+1))^{\ell}\sqrt T}$. By \Cref{lemma:linearization}, for $T\ge (C_\star A^2(d+1))^{2\ell}/\varepsilon^5$ and a sufficiently large universal constant $C_\star$,
$$
\begin{aligned}
\frac1T\sum_{t=1}^T U_t^\mathrm{poly}(q)
&\le
\frac1T\sum_{t=1}^T U_t^\mathrm{poly}(q_t)
+
\frac{4(CA^2(d+1))^{\ell}}{\varepsilon^{3/2}\sqrt T}
\le
\frac1T\sum_{t=1}^T U_t^\mathrm{poly}(q_t)
+
\frac{\varepsilon}{16}.
\end{aligned}
$$
By Eq.~\eqref{eq:self-cert},
$U_t^\mathrm{poly}(q_t)\le 0$ for every $t\in[T]$. Then, expanding $U_t^\mathrm{poly}(q)$ according to Eq.~\eqref{eq:loss-def}, we get
$$
\frac{1}{T}\sum_{t=1}^T
\widehat{\err}_{S_{\mathrm{ref}}}(h_t,q) - \frac{1}{T}\sum_{t=1}^T \eta^{\mathrm{poly}}(q,m_t; \bar S_{\mathrm{filt},t}, S_{\mathrm{ref}}) \le \frac{\varepsilon}{16}.
$$
Since the empirical error is linear in its first argument, $\frac{1}{T}\sum_{t=1}^T \widehat{\err}_{S_{\mathrm{ref}}}(h_t,q) = \widehat{\err}_{S_{\mathrm{ref}}}(\bar h_T,q)$. Thus
$$
\widehat{\err}_{S_{\mathrm{ref}}}(\bar h_T,q)
\le
\frac1T\sum_{t=1}^T
\eta^{\mathrm{poly}}
(q,m_t;\bar S_{\mathrm{filt},t},S_{\mathrm{ref}})
+\frac{\varepsilon}{16}.
$$

\end{proof}

\subsection{
\texorpdfstring
{Proof of \Cref{lemma:certificate-guarantee}}
{Proof of certificate-guarantee}
}
\label{sec:proof-of-certificate-guarantee}

\begin{proof}

The assumed sample sizes allow us to apply \Cref{lemma:generalization}, \Cref{lemma:empirical-approximation}, and the clean-point preservation guarantee of \Cref{thm:weighted-ipf} with failure probability $\delta/3$ each. Their conclusions hold simultaneously with probability at least $1-\delta$. Assume this throughout the proof and fix $t\in[T]$.

As noted above, the polynomial certificate $\eta^\mathrm{poly}$ from \Cref{def:polynomial-certificate} is
$$
\eta^{\mathrm{poly}}
(q,m;\bar S_{\mathrm{filt}},S_{\mathrm{ref}})
:=
\frac1{2|\bar S_{\mathrm{inp}}|}
\sum_{(x,y)\in\bar S_{\mathrm{filt}}}
\frac{|m(x)|+ym(x)-2yq(x)}
{2+|m(x)|+ym(x)}
-
\frac1{2|S_{\mathrm{ref}}|}
\sum_{x\in S_{\mathrm{ref}}}
\frac{q(x)m(x)-1}{1+|m(x)|}.
$$
In Part 1, we bound the reference sample sum, which involves only clean unlabeled instances. In Part 2, we bound the more challenging filtered sample sum, which involves both clean and corrupted labeled instances. In Part 3, we combine these bounds to obtain the final bound on $\eta^{\mathrm{poly}}$.

\noindent\textbf{Part 1: Reference sample term.}
We first replace $q^\star$ by $c^\star$ in the reference term using Eq.~\eqref{eq:reference-approximation} and then replace the reference average by the clean sample average using Eq.~\eqref{eq:population-transfer-detector}. Since $|m_t(x)|/(1+|m_t(x)|)\le1$, 
\[
\begin{aligned}
\sum_{x\in S_{\mathrm{ref}}}
\frac{1-q^\star(x)m_t(x)}{2|S_{\mathrm{ref}}|(1+|m_t(x)|)}
\le
\sum_{x\in S_{\mathrm{ref}}}
\frac{1-c^\star(x)m_t(x)}{2|S_{\mathrm{ref}}|(1+|m_t(x)|)}
+\frac{\varepsilon}{32}
\le
\sum_{x\in S_{\mathrm{cln}}}
\frac{1-c^\star(x)m_t(x)}{2|\bar S_{\mathrm{inp}}|(1+|m_t(x)|)}
+\frac{\varepsilon}{16}.
\end{aligned}
\]
The input and reference sample terms now have the same normalization, so we can compare their contributions.

\noindent\textbf{Part 2: Filtered sample term.} There are three possibilities: the adversary did not change an example and the filter did not delete it. An original clean instance was deleted by the adversary or the filter. An example inserted by the adversary was not deleted by the filter. We define the weights $v_t(x,y)\in[0,1/4]$ by
\begin{equation}\label{eq:weights}
v_t(x,y)
=
\frac{2+|m_t(x)|-ym_t(x)}{8(1+|m_t(x)|)}
=
\frac1{2(2+|m_t(x)|+ym_t(x))}.
\end{equation}
For every $(x,y)$, the corrupted sample term in the corruption certificate $\eta^\mathrm{poly}$ satisfies
\begin{align*}
\frac{|m_t(x)|+ym_t(x)-2yq^\star(x)}
{2(2+|m_t(x)|+ym_t(x))}
&=
\frac12-
\frac{1+yq^\star(x)}
{2+|m_t(x)|+ym_t(x)}
\tag{rearrangement}
\\
&=
\frac12-2v_t(x,y)\bigl(1+yq^\star(x)\bigr)
\tag{Eq.~\eqref{eq:weights}}
\\
&=
\frac12 -v_t(x,y)\bigl(q^\star(x)+y\bigr)^2 + v_t(x,y)\bigl(q^\star(x)^2-1\bigr).
\tag{$y^2=1$}
\end{align*}
By \Cref{lemma:empirical-approximation}, $(q^\star)^2-1\in\mathcal P$. Therefore, the guarantees of (weighted) iterative polynomial filtering (\Cref{thm:weighted-ipf}) give, for any $t \in [T]$,
\begin{equation}\label{eq:weighted-ipf-guarantee}
\left|
\frac1{|\bar S_{\mathrm{inp}}|}
\sum_{(x,y)\in\bar S_{\mathrm{filt},t}}
v_t(x,y)\bigl(q^\star(x)^2-1\bigr)
\right|
\le\frac{\varepsilon}{64}.
\end{equation}
It remains to bound $\frac12-v_t(x,y)\bigl(q^\star(x)+y\bigr)^2$. For a clean example in the filtered sample, $y=c^\star(x)$, so
\[
\frac12-v_t(x,y)\bigl(q^\star(x)+c^\star(x)\bigr)^2
-\frac{c^\star(x)m_t(x)-1}{2(1+|m_t(x)|)}
=
v_t(x,y)\left(4-\bigl(q^\star(x)+c^\star(x)\bigr)^2\right).
\]
Since $(c^\star(x))^2=1$, we have
\begin{align*}
v_t(x,y)\left(4-\bigl(q^\star(x)+c^\star(x)\bigr)^2\right)
&=
-4v_t(x,y)c^\star(x)\bigl(q^\star(x)-c^\star(x)\bigr)
-v_t(x,y)\bigl(q^\star(x)-c^\star(x)\bigr)^2
\notag
\\
&\le
4v_t(x,y)|q^\star(x)-c^\star(x)|
\tag{\text{$2$nd term is nonpositive}}
\\
&\le
|q^\star(x)-c^\star(x)|.
\tag{$v_t(x,y)\le1/4$}
\end{align*}
By Eq.~\eqref{eq:clean-retained-approximation}, the total contribution of these clean examples, after division by $|\bar S_{\mathrm{inp}}|$, is at most $\varepsilon/16$. For every corrupted example in the filtered sample, the input contribution is at most $1/2$:
\begin{equation}\label{eq:corrupted-cont}
\frac12-v_t(x,y)\bigl(q^\star(x)+y\bigr)^2
\le\frac12.
\end{equation}
For every clean example removed by the adversary or the filter, only the clean sample term contributes, also by at most $1/2$:
\begin{equation}\label{eq:clean-deleted-cont}
-\frac{c^\star(x)m_t(x)-1}{2(1+|m_t(x)|)}
\le\frac12.
\end{equation}
The bounds in Eq.~\eqref{eq:corrupted-cont}, Eq.~\eqref{eq:clean-deleted-cont} and Eq.~\eqref{eq:weighted-ipf-guarantee}, together with approximation on the clean examples in the filtered sample and the result of Part~1, give
\begin{equation}\label{eq:semi-conclusion}
\begin{aligned}
&\eta^{\mathrm{poly}}
\left(q^\star,m_t;\bar S_{\mathrm{filt},t},S_{\mathrm{ref}}\right)
\le
\frac{
|S_{\mathrm{cln}}\setminus S_{\mathrm{filt},t}|}{2|\bar S_{\mathrm{inp}}|}
+
\frac{|S_{\mathrm{filt},t}\setminus S_{\mathrm{cln}}|}
{2|\bar S_{\mathrm{inp}}|}
+\frac{\varepsilon}{16}
+\frac{\varepsilon}{64}
+\frac{\varepsilon}{16}.
\end{aligned}
\end{equation}
The three error terms come from the reference replacements, weighted polynomial preservation, and approximation on the remaining clean examples, respectively.

\noindent\textbf{Part 3: Contribution of corruptions and filtering.} Since $|S_{\mathrm{cln}}|=|S_{\mathrm{inp}}|$ and $S_{\mathrm{filt},t}\subseteq S_{\mathrm{inp}}$, $\eta$-bounded contamination (\Cref{def:bounded-contamination}) and the clean-point preservation guarantee (\Cref{thm:weighted-ipf}) give
\[
\begin{aligned}
\frac{|S_{\mathrm{filt},t}\setminus S_{\mathrm{cln}}|}{2|\bar S_{\mathrm{inp}}|}
+
\frac{|S_{\mathrm{cln}}\setminus S_{\mathrm{filt},t}|}{2|\bar S_{\mathrm{inp}}|}
&=
\frac{|S_{\mathrm{inp}}\setminus S_{\mathrm{cln}}|}
{|\bar S_{\mathrm{inp}}|}
+
\frac{
|(S_{\mathrm{inp}} \cap S_{\mathrm{cln}})\setminus S_{\mathrm{filt},t}|}{2|\bar S_{\mathrm{inp}}|}
-
\frac{|(S_{\mathrm{inp}}\setminus S_{\mathrm{cln}})\setminus S_{\mathrm{filt},t}|
}{2|\bar S_{\mathrm{inp}}|}
\le
\eta+\frac{\varepsilon}{64}.
\end{aligned}
\]
Together with Eq.~\eqref{eq:semi-conclusion}, this results in
\[
\begin{aligned}
\eta^{\mathrm{poly}}
\left(q^\star,m_t;\bar S_{\mathrm{filt},t},S_{\mathrm{ref}}\right)
&\le
\eta+\frac{\varepsilon}{64}
+\frac{\varepsilon}{16}
+\frac{\varepsilon}{64}
+\frac{\varepsilon}{16}
=
\eta+\frac{5\varepsilon}{32}
\le
\eta+\frac{13\varepsilon}{64}.
\end{aligned}
\]
This bound holds simultaneously for every $t\in[T]$, so the average over $t$ satisfies the same bound.
\end{proof}

\subsection{
\texorpdfstring
{Proof of \Cref{thm:joint-robust-learning}}
{Proof of joint robust learning}
}
\label{sec:proof-of-joint-robust-learning}

\begin{proof}

We verify that the results used in the proof of \Cref{thm:main-guarantee-restated} remain valid under the new approximation assumption.

\begin{table}[H]
\centering
\small
\renewcommand{\arraystretch}{1.15}
\begin{tabular}{>{\centering\arraybackslash}p{0.20\linewidth}>{\centering\arraybackslash}p{0.16\linewidth}>{\centering\arraybackslash}p{0.18\linewidth}>{\centering\arraybackslash}p{0.32\linewidth}}
\toprule
\textbf{Lemma} & \textbf{Reference} & \textbf{Change status} & \textbf{Comment} \\
\midrule
Linearization & [\Cref{lemma:linearization}] & Unchanged & Algebraic identity \\
Gradient bound & [\Cref{lemma:gradient-bound}] & Unchanged & Algebraic bound \\
Error guarantee & [\Cref{lemma:error-guarantee}] & Unchanged & Uniform over $\mathcal Q$ \\
Certificate guarantee & [\Cref{lemma:certificate-guarantee}] & Same proof & Conditional on \Cref{lemma:empirical-approximation} \\
Generalization & [\Cref{lemma:generalization}] & Unchanged & Independent of approximation \\
Empirical approx. & [\Cref{lemma:empirical-approximation}] & Changed & The main goal of this proof \\
Weighted IPF & [\Cref{thm:weighted-ipf}] & Unchanged & Approximation bound holds \\
Feasibility of $q^\star$ & [\Cref{thm:main-guarantee-restated}] & Same proof & Conditional on \Cref{lemma:empirical-approximation} \\
\bottomrule
\end{tabular}
\caption{Changes under structured $\mathcal L_1$ approximation (\Cref{def:joint-polynomial-approximation}).}
\label{table:changes}
\end{table}

Fix an approximator $q^\star$ for $c^\star$. Based on Table~\ref{table:changes}, it remains to verify the conclusions of \Cref{lemma:empirical-approximation} under the new approximation assumptions in \Cref{def:joint-polynomial-approximation}. These assumptions give
$$
\E_{x\sim\mathcal D}|q^\star(x)^2-1|\le\frac{\varepsilon}{512}
\qquad\text{and}\qquad
\E_{x\sim\mathcal D}|q^\star(x)-c^\star(x)|
\le\frac{\varepsilon}{512}.
$$
The proof of \Cref{lemma:empirical-approximation} applies with two changes. In Part 1, the bound $\E_{x\sim\mathcal D}|q^\star(x)^2-1|\le\varepsilon/512$ holds directly by assumption, so \Cref{thm:weighted-ipf} still gives $(q^\star)^2-1\in\mathcal P$. In Part 2, we instead use
$$
\Pr_{x\sim\mathcal D}[c^\star(x)q^\star(x)\le0]
\le
\E_{x\sim\mathcal D}|q^\star(x)-c^\star(x)|
\le\frac{\varepsilon}{512}.
$$
The remainder of Part 2 and all of Part 3 are unchanged. Thus, all conclusions of \Cref{lemma:empirical-approximation} remain valid. The remaining steps are the same as in the proof of \Cref{thm:main-guarantee-restated}.

\end{proof}

\subsection{
\texorpdfstring
{Proof of \Cref{lemma:joint-smoothing}}
{Proof of joint smooth}
}
\label{sec:proof-of-joint-smoothing}

\begin{proof}
Define the noise operator $T_\rho c(x):=\E[c(Y)\mid X=x]$, where $(X,Y)$ is the $\rho$-correlated pair from \Cref{def:noise-sensitivity}. Since $c(x) \in \{\pm 1\}$, we have $|T_\rho c(x)|\le1$ and
$$
\E_{x\sim\mathcal D}|T_\rho c(x)-c(x)|
=
\E[1-c(X)c(Y)]
=
2\Pr[c(X)\neq c(Y)]
=
2\operatorname{NS}_{\delta}[c]
\le\frac{\alpha}{4}.
$$
Write $c=\sum_{j\ge0}c^{=j}$ for its Hermite expansion in the Gaussian case or its Fourier expansion on the hypercube, where $c^{=j}$ is the degree-$j$ component. Then $T_\rho c=\sum_{j\ge0}\rho^j c^{=j}$. We truncate this expansion at degree $\ell$ to obtain the polynomial $q:=\sum_{j=0}^{\ell}\rho^j c^{=j}$. By orthogonality, $\E[c^2]=1$, and $\log \rho \le \rho - 1$,
$$
\begin{aligned}
\E_{x\sim\mathcal D}[(q(x)-T_\rho c(x))^2]
=
\sum_{j>\ell}\rho^{2j}\E_{x\sim\mathcal D}[(c^{=j}(x))^2]
\le
\rho^{2(\ell+1)}
\le
e^{2(\ell+1)\log(\rho)}
\le
e^{-2(1-\rho)(\ell+1)}
\le
\frac{\alpha^2}{16}.
\end{aligned}
$$
To prove the first part of Eq.~\eqref{eq:goals},
$$
\E_{x\sim\mathcal D}|q(x)-c(x)|
\le
\E_{x\sim\mathcal D}|q(x)-T_\rho c(x)|
+
\E_{x\sim\mathcal D}|T_\rho c(x)-c(x)|
\le
\sqrt{\E_{x\sim\mathcal D}[(q(x)-T_\rho c(x))^2]}
+\frac{\alpha}{4}
\le\frac{\alpha}{2}.
$$
By orthogonality of the degree components and $|T_\rho c(x)|\le1$,
$$
\E_{x\sim\mathcal D}[q(x)^2]
=
\sum_{j=0}^{\ell}\rho^{2j}
\E_{x\sim\mathcal D}[(c^{=j}(x))^2]
\le
\sum_{j\ge0}\rho^{2j}
\E_{x\sim\mathcal D}[(c^{=j}(x))^2]
=
\E_{x\sim\mathcal D}[(T_\rho c(x))^2]
\le1.
$$
Thus, Young's inequality gives
$$
\E_{x\sim\mathcal D}[(q(x)+T_\rho c(x))^2]
\le
2\E_{x\sim\mathcal D}[q(x)^2]
+
2\E_{x\sim\mathcal D}[(T_\rho c(x))^2]
\le4.
$$
To prove the second part of Eq.~\eqref{eq:goals}, we note that $|T_\rho c(x)|\le1$ and $c(x)^2=1$ imply $|(T_\rho c(x))^2-1|\le2|T_\rho c(x)-c(x)|$. The triangle inequality and Cauchy--Schwarz inequality therefore give
$$
\begin{aligned}
\E_{x\sim\mathcal D}|q(x)^2-1|
&\le
\E_{x\sim\mathcal D}|(T_\rho c(x))^2-1|
+
\E_{x\sim\mathcal D}|q(x)^2-(T_\rho c(x))^2|
\\
&\le
2\E_{x\sim\mathcal D}|T_\rho c(x)-c(x)|
+
\sqrt{\E_{x\sim\mathcal D}[(q(x)-T_\rho c(x))^2]}
\sqrt{\E_{x\sim\mathcal D}[(q(x)+T_\rho c(x))^2]}
\\
&\le
2 \cdot \frac{\alpha}{4}+\frac{\alpha}{4}\cdot2
=
\alpha.
\end{aligned}
$$
\end{proof}

\subsection{
\texorpdfstring
{Proof of \Cref{cor:gaussian-surface-area-learning}}
{Proof of GSA}
}
\label{sec:proof-of-gsa-learning}

\begin{proof}
By \cite[Corollary~14]{klivans2008learning}, for every $c\in\mathcal C$ and $\rho\in(0,1)$,
$$
\operatorname{NS}_{(1-\rho)/2}[c]
\le
O\left(\operatorname{GSA}(c)\sqrt{1-\rho}\right)
\le
O\left(\Gamma\sqrt{1-\rho}\right).
$$
Choose $1-\rho=\alpha^2/(C\Gamma^2)$ for a sufficiently large universal constant $C$. Then $\operatorname{NS}_{(1-\rho)/2}[c]\le\alpha/8$, and \Cref{lemma:joint-smoothing} gives the approximation bounds with degree $\ell = O(\frac{\Gamma^2}{\alpha^2}\log\frac4\alpha)$.
The learning guarantee follows from \Cref{thm:joint-robust-learning} with $\alpha=\varepsilon/512$.
\end{proof}

\section{Generalization bounds}\label{app:generalization}

In this section, we establish the concentration bounds that connect the empirical quantities used by the algorithm with their population counterparts under $\mathcal{D}$. We first state the definitions of VC dimension and Rademacher complexity and then collect the standard properties required for the analysis.

\begin{definition}[VC dimension]\label{def:gc-dimension}
For a class $\mathcal H\subseteq\{\pm1\}^X$,
$\operatorname{VCdim}(\mathcal H)$ is the cardinality of the largest
set $S\subseteq X$ on which $\mathcal H$ realizes all $2^{|S|}$
labelings.
\end{definition}

\begin{definition}[Rademacher complexity]
\label{def:rademacher-complexity}
Let $\mathcal F$ be a class of functions from $X$ to $[-1,1]$, let
$\mathcal{D}$ be a distribution over $X$, and let $|S|\ge1$. We define
$$
\mathfrak R_{|S|,\mathcal{D}}(\mathcal F)
:=
\mathbb E\left[
\sup_{f\in\mathcal F}
\frac{1}{|S|}\sum_{i=1}^{|S|} \sigma_i f(X_i)
\right],
$$
where $X_1,\ldots,X_n$ are drawn independently from $\mathcal{D}$ and
$\sigma_1,\ldots,\sigma_n$ are independent random variables uniformly
distributed on $\{\pm 1\}$.
\end{definition}

\begin{fact} We present six standard facts from empirical process theory.
\begin{factlist}
    \item\label{fact:vc-dimension} \cite[Proposition~4.20]{Wainwright2019} The class $\mathcal H_0$ is a linear-threshold class in the $N$-dimensional space, so 
    $$
    \operatorname{VCdim}(\mathcal H_0)\le N
    $$
    \item\label{fact:vc-rademacher} \cite[Example~5.24]{Wainwright2019} For every class $\mathcal H\subseteq\{\pm1\}^X$ and every sample size $|S|$, with $C_1>0$ a universal constant,
    $$
    \mathfrak R_{|S|,\mathcal{D}}(\mathcal H)
    \le
    C_1\sqrt{\frac{\operatorname{VCdim}(\mathcal H)}{|S|}}
    $$
    \item\label{fact:convexification}
    \cite[Lemma~7.4]{MohriRostamizadehTalwalkar2018}
    Rademacher complexity is invariant under convexification:
    $$
    \mathfrak R_{|S|,\mathcal{D}}
    \left(
    \operatorname{conv}(\mathcal H)
    \right)
    =
    \mathfrak R_{|S|,\mathcal{D}}(\mathcal H).
    $$
    \item\label{fact:contraction} \cite[Lemma~5.7]{MohriRostamizadehTalwalkar2018}
    For any class $\mathcal G$ and any family
    $(\psi_x)_{x\in X}$ of $K$-Lipschitz maps,
    $$
    \mathfrak R_{|S|,\mathcal{D}}(\psi_x\circ\mathcal G) \le K \cdot \mathfrak R_{|S|,\mathcal{D}}(\mathcal G),
    $$
    \item
    \label{fact:vc-majority}
    \cite[Theorem~8]{FreundSchapire97}
    For every concept class $\mathcal C$ and every $k\ge1$,
    $$
    \operatorname{VCdim}
    \left(
    \operatorname{Maj}_k(\mathcal C)
    \right)
    \le
    O\left(
    \operatorname{VCdim}(\mathcal C)\,k\log(2k)
    \right).
    $$
    \item\label{fact:rademacher-generalization} \cite[Theorem~3.3]{MohriRostamizadehTalwalkar2018}
    For any class $\mathcal F$ taking values on a unit interval, with probability at least $1-\delta$,
    $$
    \sup_{f\in\mathcal F}
    \left|
    \E_{x \sim \mathcal{D}}\left[f(x)\right]
    -
    \frac{1}{|S|}\sum_{i=1}^{|S|} f(x_i)
    \right| \le 2\cdot\mathfrak R_{|S|,\mathcal{D}}(\mathcal F) + \sqrt{\frac{\log(2/\delta)}{2|S|}}.
    $$
\end{factlist}
\end{fact}

\subsection{Generalization bounds for learning with low-degree polynomials}

For the same degree $\ell$, define the degree-$\ell$ polynomial class
$$
\mathcal M:=\{m:X\to\mathbb R:\deg(m)\le \ell\}.
$$

\begin{lemma}[Statistical guarantee]
\label{lemma:generalization}

Let $\mathcal{D}$ be $A$-hypercontractive, and let $S_{\mathrm{cln}}$ and $S_{\mathrm{ref}}$ be i.i.d. samples from $\mathcal{D}$. There is a universal constant $C_0>0$ such that, if $|\bar S_{\mathrm{inp}}|
\ge
\frac{C_0}{\varepsilon^2}
\left(N+\log\frac{8}{\delta}\right)$ and $|S_{\mathrm{ref}}|
\ge
\frac{C_0}{\varepsilon^2}
\left(N+\log\frac{8}{\delta}\right)$,
then, with probability at least $1-\delta$, the following inequalities
hold simultaneously:
\begin{align}
&\sup_{\bar h\in \operatorname{conv}(\mathcal H_0)}
\left|
\err_{\mathcal{D}}(\bar h, c^\star) - \widehat{\err}_{S_{\mathrm{ref}}}(\bar h, c^\star)
\right|
\le
\frac{\varepsilon}{16},
\label{eq:population-transfer-error}
\\[0.75em]
&\sup_{m\in\mathcal M}
\left| \frac{1}{|S_{\mathrm{cln}}|}\sum_{x \in S_{\mathrm{cln}}}
\frac{c^\star(x)m(x)-1}{1+|m(x)|} - \frac{1}{|S_{\mathrm{ref}}|}\sum_{x \in S_{\mathrm{ref}}} \frac{c^\star(x)m(x)-1}{1+|m(x)|} \right|
\le
\frac{\varepsilon}{16}.
\label{eq:population-transfer-detector}
\end{align}
\end{lemma}

\begin{proof}

The proof has two parts, corresponding respectively to Eq.~\eqref{eq:population-transfer-error} and~\eqref{eq:population-transfer-detector}.

\noindent \textbf{Part 1: Eq.~\eqref{eq:population-transfer-error}}. By
\Cref{fact:vc-dimension}, \Cref{fact:vc-rademacher}, and \Cref{fact:convexification}
$$
\mathfrak R_{|S_{\mathrm{ref}}|,\mathcal{D}}
\bigl(\operatorname{conv}(\mathcal H_0)\bigr) = \mathfrak R_{|S_{\mathrm{ref}}|,\mathcal{D}}
\bigl(\mathcal H_0\bigr)
\le
C_1\sqrt{\frac{N}{|S_{\mathrm{ref}}|}}.
$$
For every $x$, the map $z\mapsto \frac{1-c^\star(x)z}{2}$ is $1/2$-Lipschitz. Combining \Cref{fact:contraction} and \Cref{fact:rademacher-generalization} therefore gives, with probability $1 - \delta / 4$, Eq.~\eqref{eq:population-transfer-error}:
$$
\sup_{\bar h\in \operatorname{conv}(\mathcal H_0)}
\left|\err_{\mathcal{D}}(\bar h,c^\star)- \widehat{\err}_{S_{\mathrm{ref}}}(\bar h,c^\star)\right| \le
C_1\sqrt{\frac{N}{|S_{\mathrm{ref}}|}}
+
\sqrt{\frac{\log(8/\delta)}{2|S_{\mathrm{ref}}|}}
\le
\frac{\varepsilon}{16}.
$$
Thus, Part~1 holds provided $|S_{\mathrm{ref}}| \ge \frac{C_0}{\varepsilon^2}\left(N+\log\frac{8}{\delta}\right)$, for a sufficiently large universal constant $C_0>0$.

\noindent\textbf{Part 2: Eq.~\eqref{eq:population-transfer-detector}.}
Define
$$
\mathcal G
:=
\left\{
x\mapsto
\sign
\left(
c^\star(x)m_w(x)-z
\right)
:
m_w\in\mathcal M,\ z\ge0
\right\}.
$$
Since $c^\star(x)m_w(x)-z = \left\langle w,c^\star(x)x^{\otimes\ell}\right\rangle-z$,
the class $\mathcal G$ is a linear-threshold class with $N+1$
parameters. Therefore,
$$
\operatorname{VCdim}(\mathcal G)\le N+1.
$$
For every $m\in\mathcal M$, the reference sample term from \Cref{def:polynomial-certificate} admits the following integral representation:
\begin{align*}
\int_0^\infty
\sign(c^\star(x)m(x)-z)\frac{dz}{(1+z)^2}
&=
\begin{cases}
\displaystyle
\int_0^{c^\star(x)m(x)}\frac{dz}{(1+z)^2}
-
\int_{c^\star(x)m(x)}^\infty\frac{dz}{(1+z)^2},
& c^\star(x)m(x)\ge0,\\[0.75em]
\displaystyle
-\int_0^\infty\frac{dz}{(1+z)^2},
& c^\star(x)m(x)<0,
\end{cases}
\\
&=
\begin{cases}
1-\dfrac{2}{1+c^\star(x)m(x)},
& c^\star(x)m(x)\ge0,\\[0.5em]
-1, & c^\star(x)m(x)<0.
\end{cases}
\\
&=
\frac{c^\star(x)m(x)-1}{1+|m(x)|}.
\tag{$|c^\star(x)|=1$}
\end{align*}
Define $\mathcal F:=\{x\mapsto\frac{c^\star(x)m(x)-1}{1+|m(x)|}:m\in\mathcal M\}$. On every finite sample $S$, the integral representation expresses
the evaluation vector of each function in $\mathcal F$ as a convex
combination of evaluation vectors from $\mathcal G$: as $z$ varies,
the threshold vector takes only finitely many values, and the
weights integrate to one. Since Rademacher complexity depends only
on these evaluation vectors, \Cref{fact:vc-rademacher} and \Cref{fact:convexification} give
$$
\mathfrak R_{|S|,\mathcal{D}}(\mathcal F)
\le
\mathfrak R_{|S|,\mathcal{D}}(\mathcal G)
\le
C_1\sqrt{\frac{N+1}{|S|}}.
$$
Since every function in $\mathcal F$ takes values in $[-1,1]$, we apply \Cref{fact:rademacher-generalization} to $\{(1+f)/2:f\in\mathcal F\}$ and multiply the resulting bounds by $2$. This factor is absorbed into the universal constant $C_0$. Thus, the following bounds hold, each with probability at least $1-\delta/4$:
\begin{align}
&
\sup_{m\in\mathcal M}
\left|
\frac{1}{|S_{\mathrm{cln}}|}
\sum_{x\in S_{\mathrm{cln}}}
\frac{c^\star(x)m(x)-1}{1+|m(x)|}
-
\E_{x \sim \mathcal{D}}
\left[
\frac{c^\star(x)m(x)-1}{1+|m(x)|}
\right]
\right|
\le
\frac{\varepsilon}{32},
\label{eq:first-bound}
\\
&
\sup_{m\in\mathcal M}
\left|
\frac1{|S_{\mathrm{ref}}|}
\sum_{x\in S_{\mathrm{ref}}}
\frac{c^\star(x)m(x)-1}{1+|m(x)|}
-
\E_{x \sim \mathcal{D}}
\left[
\frac{c^\star(x)m(x)-1}{1+|m(x)|}
\right]
\right|
\le
\frac{\varepsilon}{32}.
\label{eq:second-bound}
\end{align}
By the triangle inequality, Eq.~\eqref{eq:first-bound} and~\eqref{eq:second-bound} imply, on an event of probability at least
$1-\delta/2$,
\begin{align*}
&
\sup_{m\in\mathcal M}
\left|
\frac{1}{|S_{\mathrm{cln}}|}
\sum_{x\in S_{\mathrm{cln}}}
\frac{c^\star(x)m(x)-1}{1+|m(x)|}
-
\frac1{|S_{\mathrm{ref}}|}
\sum_{x\in S_{\mathrm{ref}}}
\frac{c^\star(x)m(x)-1}{1+|m(x)|}
\right|
\\
&\le
\sup_{m\in\mathcal M}
\left|
\frac{1}{|S_{\mathrm{cln}}|}
\sum_{x\in S_{\mathrm{cln}}}
\frac{c^\star(x)m(x)-1}{1+|m(x)|}
-
\E_{x \sim \mathcal{D}}
\left[
\frac{c^\star(x)m(x)-1}{1+|m(x)|}
\right]
\right|
\\
&+
\sup_{m\in\mathcal M}
\left|
\E_{x \sim \mathcal{D}}
\left[
\frac{c^\star(x)m(x)-1}{1+|m(x)|}
\right]
-
\frac1{|S_{\mathrm{ref}}|}
\sum_{x\in S_{\mathrm{ref}}}
\frac{c^\star(x)m(x)-1}{1+|m(x)|}
\right|
\le
\frac{\varepsilon}{32}
+
\frac{\varepsilon}{32}
=
\frac{\varepsilon}{16}.
\end{align*}
Thus, Part~2 holds provided
$
|\bar S_{\mathrm{inp}}|
\ge
\frac{C_0}{\varepsilon^2}
\left(
N+\log\frac{8}{\delta}
\right)
$
and
$
|S_{\mathrm{ref}}|
\ge
\frac{C_0}{\varepsilon^2}
\left(
N+\log\frac{8}{\delta}
\right),
$
for a sufficiently large universal constant $C_0>0$. A union bound over the two parts completes the proof.
\end{proof}

\subsection{Generalization bounds for learning with access to an ERM oracle}

\begin{lemma}[Statistical guarantee]
\label{lemma:erm-generalization}

Let $\mathcal C\subseteq\{\pm1\}^X$ have finite VC dimension, let $\mathcal{D}$
be a distribution over $X$, let $c^\star\in\mathcal C$, and fix
$k\in\mathbb N$. Let $S_{\mathrm{ref}}$ be an independent sample from $\mathcal{D}$. There is a universal constant $C>0$ such that, if
$
|S_{\mathrm{ref}}|
\ge
C
\frac{
\operatorname{VCdim}(\mathcal C)k\log(2k)+\log(4/\delta)
}{
\varepsilon^2
}$,
then, with probability at least $1-\delta$, 
\begin{align}
&\sup_{\bar h\in
\operatorname{conv}(\operatorname{Maj}_k(\mathcal C))}
\left|
\widehat{\err}_{S_{\mathrm{ref}}}
(\bar h,c^\star)
-
\err_{\mathcal{D}}(\bar h,c^\star)
\right|
\le
\frac{\varepsilon}{4}.
\label{eq:erm-error-generalization}
\end{align}

\end{lemma}

\begin{proof}
\Cref{fact:vc-majority} gives the VC-dimension bound for
$\operatorname{Maj}_k(\mathcal C)$. For the fixed concept $c^\star$, the
maps $z \mapsto \frac{1-c^\star(x)z}{2}$ are $1/2$-Lipschitz. Therefore, combining \Cref{fact:vc-rademacher}, \Cref{fact:contraction}, and \Cref{fact:rademacher-generalization} gives,
with probability at least $1-\delta$,
$$
\sup_{h\in\operatorname{Maj}_k(\mathcal C)}
\left|
\widehat{\err}_{S_{\mathrm{ref}}}(h,c^\star)
-
\err_{\mathcal{D}}(h,c^\star)
\right|
\le
\frac{\varepsilon}{4}.
$$
For every
$\bar h\in\operatorname{conv}(\operatorname{Maj}_k(\mathcal C))$,
write $\bar h = \sum_{j=1}^r\lambda_j h_j$ with $h_j\in\operatorname{Maj}_k(\mathcal C)$ for $\lambda_j\ge0$ and $\sum_{j=1}^r\lambda_j=1$. By linearity of the error in its first argument,
$$
\begin{aligned}
\left|
\widehat{\err}_{S_{\mathrm{ref}}}
(\bar h,c^\star)
-
\err_{\mathcal{D}}(\bar h,c^\star)
\right|
&=
\left|
\sum_{j=1}^r\lambda_j
\left(
\widehat{\err}_{S_{\mathrm{ref}}}
(h_j,c^\star)
-
\err_{\mathcal{D}}(h_j,c^\star)
\right)
\right|
\\
&\le
\sum_{j=1}^r\lambda_j
\left|
\widehat{\err}_{S_{\mathrm{ref}}}
(h_j,c^\star)
-
\err_{\mathcal{D}}(h_j,c^\star)
\right|\le
\frac{\varepsilon}{4}.
\end{aligned}
$$
Taking the supremum over
$\bar h\in\operatorname{conv}(\operatorname{Maj}_k(\mathcal C))$
proves Eq.~\eqref{eq:erm-error-generalization}.

\end{proof}

\section{Weighted Iterative Polynomial Filtering}

In this section, we present a weighted variant of the iterative polynomial filtering algorithm of \cite{KSTV25}. The key modification is the introduction of bounded weights $v(x,y)$.

\begin{theorem}[Adaptive weighted iterative filtering]
\label{thm:weighted-ipf}
Let $\mathcal D$ be an $A$-hypercontractive distribution, and let
$\varepsilon,\delta\in(0,1)$ and $\ell\ge1$.
For a sufficiently large universal constant $C$, suppose
$|\bar S_{\mathrm{inp}}| \ge \frac{(CA(d+1))^{4\ell}}{\varepsilon^3}\log\frac C\delta$ and $|S_{\mathrm{ref}}| \ge \frac{(CA(d+1))^{4\ell}}{\varepsilon^3} \left(\log\frac C\delta\right)^{8\ell+1}$. The first call to Algorithm~\ref{fig:weighted-ipf} takes $\bar S_{\mathrm{inp}}$ and each later call takes the output of the previous call. Every call uses the same parameters, reference sample $S_{\mathrm{ref}}$, and normalization $|\bar S_{\mathrm{inp}}|$. The function $m$ may change between calls based on previous outputs but is fixed within each call. Write $v(x,y):=\frac{2+|m(x)|-ym(x)}{8(1+|m(x)|)}$ for its weights. The following guarantees hold:
\begin{enumerate}
    \item \textbf{Clean-point preservation.}
    With probability at least $1-\delta$, every filtered sample
    $\bar S_{\mathrm{filt}}$ satisfies
    \[
    \frac{
    |(S_{\mathrm{cln}}\cap S_{\mathrm{inp}})\setminus S_{\mathrm{filt}}|}{{|\bar S_{\mathrm{inp}}|}}
    -
    \frac{|(S_{\mathrm{inp}}\setminus S_{\mathrm{cln}})\setminus S_{\mathrm{filt}}|
    }{|\bar S_{\mathrm{inp}}|}
    \le\frac{\varepsilon}{32}.
    \]

    \item \textbf{Weighted polynomial preservation.}
    For any polynomial $p$ of degree at most $2\ell$ satisfying
    $
    \E_{x\sim\mathcal D}|p(x)|
    \le\frac{\varepsilon}{128},
    $
    with probability at least $1-\delta$ over
    $S_{\mathrm{ref}}$, we have $p\in\mathcal P$, and every call
    with function $m$ returns a sample $\bar S_{\mathrm{filt}}$
    satisfying
    \[
    \left|
    \frac1{|\bar S_{\mathrm{inp}}|}
    \sum_{(x,y)\in\bar S_{\mathrm{filt}}}
    v(x,y)p(x)
    \right|
    \le\frac{\varepsilon}{64}.
    \]
\end{enumerate}
Each call can be implemented in time
$\operatorname{poly}(|\bar S_{\mathrm{inp}}|,
|S_{\mathrm{ref}}|,(d+1)^{2\ell})$.
\end{theorem}

The proof closely follows that of
\cite[Theorem~3.2]{KSTV25}, with the additional bounded
weight $v(x,y)$ in the filtering objective.

\begin{proof}
Our parameters $\beta,\gamma,B,\Delta$ coincide with those in
\cite[Theorem~3.2]{KSTV25} at degree $2\ell$, accuracy
$\varepsilon/16$, and $R=2$.

The initial filtering in line~\ref{line:weighted-ipf-initial} of Algorithm~\ref{fig:weighted-ipf} and uniform VC concentration arguments give, with probability at least $1-\delta$,
an initial loss of at most
$\varepsilon|\bar S_{\mathrm{inp}}|/32$ unchanged clean examples
and the following inequality, simultaneously for every
degree-$2\ell$ polynomial $p$ and every $\tau\ge0$:
\begin{equation}\label{eq:weighted-ipf-tail-comparison}
\frac{
|\{x\in S_{\mathrm{cln}}:|p(x)|>\tau\}|
}{|\bar S_{\mathrm{inp}}|}
\le
\frac{
|\{x\in S_{\mathrm{ref}}:|p(x)|>\tau\}|
}{|S_{\mathrm{ref}}|}
+\frac{\Delta}{R}.
\end{equation}
The stated sample sizes satisfy the requirements of these
arguments.

\noindent\textbf{Existence of a deletion threshold.}
The weights satisfy
\begin{equation}\label{eq:weighted-ipf-weight-bound}
0\le v(x,y)\le\frac14.
\end{equation}
Suppose, for contradiction, that $\lambda^\star>\varepsilon/64$ but no threshold satisfies the deletion condition. Line~\ref{line:weighted-ipf-initial} of Algorithm~\ref{fig:weighted-ipf} ensures that $|p^\star(x)|\le B$ on the filtered sample.
The failed tail inequality, integrated over $[0,B]$, therefore
gives
\begin{align*}
\lambda^\star
=
\frac1{|\bar S_{\mathrm{inp}}|}
\sum_{(x,y)\in\bar S_{\mathrm{filt}}}
v(x,y)p^\star(x)
\le
\frac1{4|\bar S_{\mathrm{inp}}|}
\sum_{x\in S_{\mathrm{filt}}}|p^\star(x)|
&\le
\frac14\left(
\frac R{|S_{\mathrm{ref}}|}
\sum_{x\in S_{\mathrm{ref}}}|p^\star(x)|
+B\Delta
\right)
\\
&\le
\frac14(R\gamma+B\Delta)
=
\frac{\varepsilon}{64}.
\end{align*}
This contradicts $\lambda^\star>\varepsilon/64$, so there exists a threshold that satisfies the deletion condition.

\noindent\textbf{Clean versus corrupted deletions.}
The family $\mathcal P$ (line~\ref{line:poly-family}) and the threshold $B$ are unchanged
across calls, so line~\ref{line:weighted-ipf-initial} of Algorithm~\ref{fig:weighted-ipf} removes no additional examples after the first call. By Eq.~\eqref{eq:weighted-ipf-tail-comparison} and the deletion condition, each deleted tail contains at most a $1/R=1/2$
fraction of unchanged clean examples. Hence each tail deletion removes at least as many corrupted as unchanged clean examples.

All deletions are permanent, so the deleted tails are disjoint
across the entire sequence. Thus, tail deletions cannot increase
the number of clean deletions minus corrupted deletions.
Line~\ref{line:weighted-ipf-initial} of Algorithm~\ref{fig:weighted-ipf} removes at most $\varepsilon|\bar S_{\mathrm{inp}}|/32$ unchanged clean examples. Consequently, after every call, the filtered sample satisfies
\[
\frac{
|(S_{\mathrm{cln}}\cap S_{\mathrm{inp}})
\setminus S_{\mathrm{filt}}|}{|\bar S_{\mathrm{inp}}|}
-
\frac{|(S_{\mathrm{inp}}\setminus S_{\mathrm{cln}})
\setminus S_{\mathrm{filt}}|
}{|\bar S_{\mathrm{inp}}|}
\le
\frac{\varepsilon}{32}.
\]
\noindent\textbf{Weighted polynomial preservation.}
Fix a polynomial $p$ of degree at most $2\ell$ such that
\[
\E_{x\sim\mathcal D}|p(x)|\le\frac{\varepsilon}{128}.
\]
The fixed-polynomial concentration argument in \cite[proof of Theorem~3.2]{KSTV25} gives $p\in\mathcal P$ with probability at least $1-\delta$ over $S_{\mathrm{ref}}$. Here, the population threshold is $(\varepsilon/16)/(4R)=\varepsilon/128$ under our parameter substitution.

When the algorithm returns, $\lambda^\star\le\varepsilon/64$. Since $p\in\mathcal P$ implies $-p\in\mathcal P$, the maximization gives, for every $p\in\mathcal P$,
\[
\left|
\frac1{|\bar S_{\mathrm{inp}}|}
\sum_{(x,y)\in\bar S_{\mathrm{filt}}}
v(x,y)p(x)
\right|
\le\frac{\varepsilon}{64}.
\]
The reference sample and $\mathcal P$ are the same in every call, so this bound holds for the fixed polynomial $p$ in every adaptive call on the same event.

\noindent\textbf{Runtime.}
As in the proof of \cite[Theorem~3.2]{KSTV25},
line~\ref{line:weighted-ipf-initial} and line~\ref{line:maximization-step} of Algorithm~\ref{fig:weighted-ipf} reduce to convex optimization problems of polynomial size. The weights change only the linear objective. Every rejected maximization removes at least one example, so each call has at most $|\bar S_{\mathrm{inp},t}|$
rejections. The deletion threshold can be found from the sorted observed
values of $|p^\star(x)|$. This gives the stated runtime.
\end{proof}

\begin{lemma}[Empirical approximation]
\label{lemma:empirical-approximation}

Let $\mathcal{D}$ be an $A$-hypercontractive distribution and let $q^\star$ be a degree $\ell$ polynomial satisfying
$
\E_{x\sim \mathcal{D}}
\left[
(q^\star(x)-c^\star(x))^2
\right]
\le
\varepsilon^2 / 1536^2.
$
For a sufficiently large universal constant $C$, suppose
$
|\bar S_{\mathrm{inp}}|
\ge
\frac{(CA(d+1))^{4\ell}}{\varepsilon^3}
\log\left(\frac{C}{\delta}\right)
$
and
$
|S_{\mathrm{ref}}|
\ge
\frac{(CA(d+1))^{4\ell}}{\varepsilon^3}
(
\log\frac{C}{\delta}
)^{8\ell+1}$. Then, with probability at least $1-\delta$, 
\begin{enumerate}
    \item The polynomial belongs to the filtering family defined in Algorithm~\ref{fig:weighted-ipf}:
    \begin{equation}\label{eq:belongs}
    (q^\star)^2-1\in\mathcal P.
    \end{equation}
    
    \item The reference sample satisfies
    \begin{equation}\label{eq:reference-approximation}
    \frac1{|S_{\mathrm{ref}}|}
    \sum_{x\in S_{\mathrm{ref}}}
    |q^\star(x)-c^\star(x)|
    \le
    \frac{\varepsilon}{16}.
    \end{equation}

    \item The retained clean sample satisfies, for every $t\in[T]$,
    \begin{equation}\label{eq:clean-retained-approximation}
    \frac1{|\bar S_{\mathrm{inp}}|}
    \sum_{x\in S_{\mathrm{cln}}\cap S_{\mathrm{filt},t}}
    |q^\star(x)-c^\star(x)|
    \le
    \frac{\varepsilon}{16}.
    \end{equation}
\end{enumerate}
\end{lemma}

\begin{proof}

The proof consists of three parts corresponding to Eq.~\eqref{eq:belongs}, Eq.~\eqref{eq:reference-approximation}, and Eq.~\eqref{eq:clean-retained-approximation}.

\noindent\textbf{Part 1: Eq.~\eqref{eq:belongs}}. As in the proof of \cite[Theorem~4.2]{KSTV25}, Cauchy--Schwarz and $(c^\star)^2=1$ give
\[
\begin{aligned}
\E_{x\sim\mathcal D}\left|(q^\star(x))^2-1\right|
&\le
\left(
\E_{x\sim\mathcal D}[(q^\star(x)-c^\star(x))^2]
\right)^{1/2}
\left(
2+
\left(
\E_{x\sim\mathcal D}[(q^\star(x)-c^\star(x))^2]
\right)^{1/2}
\right)
\\
&\le
\frac{\varepsilon}{1536}
\left(2+\frac{\varepsilon}{1536}\right)
\le\frac{\varepsilon}{512}.
\end{aligned}
\]
By \Cref{thm:weighted-ipf}, $(q^\star)^2-1\in\mathcal P$. In particular,
\begin{equation}\label{eq:consequence}
\frac1{|S_{\mathrm{ref}}|}
\sum_{x\in S_{\mathrm{ref}}}\left|(q^\star(x))^2-1\right|
\le\frac{\varepsilon}{64}.
\end{equation}

\noindent\textbf{Part 2: Eq.~\eqref{eq:reference-approximation}}. Since $c^\star(x)q^\star(x)\le0$ implies $|q^\star(x)-c^\star(x)|\ge1$,
\[
\Pr_{x\sim\mathcal D}[c^\star(x)q^\star(x)\le0]
\le
\E_{x\sim\mathcal D}[(q^\star(x)-c^\star(x))^2]
\le\frac{\varepsilon^2}{1536^2}.
\]
Chernoff bounds on the reference and clean samples therefore give, with high probability,
\begin{equation}\label{eq:chernoff-bound}
\max\left(
\frac1{|S_{\mathrm{ref}}|}
\sum_{x\in S_{\mathrm{ref}}}
\mathbf1\{c^\star(x)q^\star(x)\le0\},
\;
\frac1{|\bar S_{\mathrm{inp}}|}
\sum_{x\in S_{\mathrm{cln}}}
\mathbf1\{c^\star(x)q^\star(x)\le0\}
\right)
\le\frac{\varepsilon}{128}.
\end{equation}
For every $x$, the following bound holds pointwise:
\begin{equation}\label{eq:l2-pointwise-transfer}
|q^\star(x)-c^\star(x)|
\le
\left|(q^\star(x))^2-1\right|
+
2\mathbf1\{c^\star(x)q^\star(x)\le0\}.
\end{equation}
Combining~\eqref{eq:l2-pointwise-transfer} with Eq.~\eqref{eq:consequence} yields
\[
\begin{aligned}
\frac1{|S_{\mathrm{ref}}|}
\sum_{x\in S_{\mathrm{ref}}}|q^\star(x)-c^\star(x)|
&\le
\frac1{|S_{\mathrm{ref}}|}
\sum_{x\in S_{\mathrm{ref}}}\left|(q^\star(x))^2-1\right|
+
\frac2{|S_{\mathrm{ref}}|}
\sum_{x\in S_{\mathrm{ref}}}
\mathbf1\{c^\star(x)q^\star(x)\le0\}
\\
&\le
\frac{\varepsilon}{64}+\frac{\varepsilon}{64}
\le\frac{\varepsilon}{16}.
\end{aligned}
\]

\noindent\textbf{Part 3: Eq.~\eqref{eq:clean-retained-approximation}}. Line~\ref{line:weighted-ipf-initial} of Algorithm~\ref{fig:weighted-ipf} and Eq.~\eqref{eq:weighted-ipf-tail-comparison}, integrated over $[0,B]$, give, simultaneously for every $t\in[T]$,
\begin{equation}\label{eq:clean-retained-polynomial-bound}
\begin{aligned}
\frac1{|\bar S_{\mathrm{inp}}|}
\sum_{x\in S_{\mathrm{cln}}\cap S_{\mathrm{filt},t}}
\left|(q^\star(x))^2-1\right|
&\le
\frac1{|S_{\mathrm{ref}}|}
\sum_{x\in S_{\mathrm{ref}}}\left|(q^\star(x))^2-1\right|
+\frac{B\Delta}{R}
\le
\frac{\varepsilon}{64}+\frac{\varepsilon}{64}
=
\frac{\varepsilon}{32}.
\end{aligned}
\end{equation}
By Eq.~\eqref{eq:chernoff-bound} and Eq.~\eqref{eq:l2-pointwise-transfer},
\[
\begin{aligned}
\frac1{|\bar S_{\mathrm{inp}}|}
\sum_{x\in S_{\mathrm{cln}}\cap S_{\mathrm{filt},t}}
|q^\star(x)-c^\star(x)|
&\le
\frac1{|\bar S_{\mathrm{inp}}|}
\sum_{x\in S_{\mathrm{cln}}\cap S_{\mathrm{filt},t}}
\left|(q^\star(x))^2-1\right|
+
\frac2{|\bar S_{\mathrm{inp}}|}
\sum_{x\in S_{\mathrm{cln}}}
\mathbf1\{c^\star(x)q^\star(x)\le0\}
\\
&\le
\frac{\varepsilon}{32}+\frac{\varepsilon}{64}
\le\frac{\varepsilon}{16}.
\end{aligned}
\]
The conclusion holds simultaneously for every $t\in[T]$ because the tail comparison is uniform and every retained clean sample is a subset of $S_{\mathrm{cln}}$.

We invoke the membership guarantee, the uniform tail comparison, and the two Chernoff bounds with failure probability $\delta/4$ each. A union bound completes the proof.
\end{proof}

\section{From randomized to deterministic hypotheses}\label{app:derandomization}

Our algorithms output randomized hypotheses. However, we can avoid randomness at prediction time and obtain a hypothesis that always gives the same label to the same point. We show that this is possible with the same $\eta+\varepsilon$ error guarantee under the Gaussian distribution, isotropic log-concave distributions, and, in sufficiently high dimensions, the uniform distribution on the Boolean hypercube.

\begin{lemma}[Deterministic hypotheses]
\label{lemma:deterministic-hypotheses}
Let $\varepsilon,\delta\in(0,1)$ and set $M:=\lceil 2\varepsilon^{-2}\log(2/\delta)\rceil$. Suppose that the assumptions of our learning theorem (either \Cref{thm:erm-main-guarantee} or \Cref{thm:main-guarantee-restated}) hold. If one of the following holds:
\begin{enumerate}
    \item $\mathcal D=\mathcal N(0,I_d)$, 
    \item $\mathcal D=\operatorname{Unif}(\{\pm1\}^d)$ for
    $d\ge\lceil\log_2 M\rceil$,
    \item $\mathcal D$ is any Isotropic Log-Concave distribution.
\end{enumerate}
then there exists a deterministic hypothesis $\tilde h$ such that, with probability at least $1-\delta$,
$$
\err_{\mathcal D}(\tilde h,c^\star)\le\eta+\varepsilon.
$$
\end{lemma}

\begin{proof}

For the ERM algorithm (\Cref{fig:erm-oracle-algorithm}), we use accuracy $\varepsilon/2$ and failure parameter $\delta/6$. For the polynomial algorithm (\Cref{fig:polynomial-algo}), we use accuracy $\varepsilon$ and failure parameter $\delta/2$.\footnote{The proof of \Cref{thm:main-guarantee-restated} gives error at most $\eta+23\varepsilon/64\le\eta+\varepsilon/2$.} Thus, with probability at least $1-\delta/2$, the output $h(x)\sim\operatorname{Rad}(\bar h_T(x))$ satisfies
\begin{equation}\label{eq:learning-guarantee}
\err_{\mathcal D}(h,c^\star)
\le\eta+\frac{\varepsilon}{2}.
\end{equation}
We first partition the domain into finitely many regions
$A_1,\ldots,A_N$, each of probability at most $1/M$.

For the Gaussian and isotropic log-concave cases,
the first coordinate has a log-concave density with mean zero
and variance one. Its density is at most $1$, and each tail
beyond distance $K\ge16$ has probability at most
$16e^{-K/16}$; see the proof of \cite[Lemma~6]{kalai2008agnostically}.
We set $K:=\lceil16\log(16M)\rceil$ and partition the first
coordinate into intervals of length $1/M$ inside $[-K,K)$,
together with the two remaining half-lines.
Each interval and each half-line has probability at most
$1/M$. This gives $N=2KM+2=O(M\log M)$ regions.

For the uniform distribution on the hypercube, we set
$b:=\lceil\log_2 M\rceil$ and partition the domain according
to its first $b$ coordinates. There are $N=2^b\le2M$
regions, each of probability $2^{-b}\le1/M$.

We now choose a Boolean hypothesis $\tilde h_j$
independently for each region $A_j$.
For the polynomial algorithm, we choose
$J_j\sim\operatorname{Unif}([T])$ and set
$\tilde h_j:=h_{J_j}$.
For the ERM algorithm, we also choose
$J_j\sim\operatorname{Unif}([T])$, then independently sample
$k$ concepts from the convex combination defining $\bar\mu_{J_j}$ and let
$\tilde h_j$ be their majority.
Here, $k$ is the majority parameter of the ERM algorithm.
In both cases, for every $x$,
$$
\mathbb E[\tilde h_j(x)]
=
\frac1T\sum_{t=1}^T h_t(x)
=
\bar h_T(x).
$$
We define $\tilde h(x):=\tilde h_j(x)$ for $x\in A_j$. All choices are made once and then kept fixed, so $\tilde h$ is a deterministic hypothesis. Let $p_j:=\Pr_{x\sim\mathcal D}[x\in A_j]$ and define
$$
Z_j:=
\Pr_{x\sim\mathcal D}
[x\in A_j,\ \tilde h_j(x)\ne c^\star(x)].
$$
The variables $Z_1,\ldots,Z_N$ are independent,
$0\le Z_j\le p_j$
$$
\sum_{j=1}^N Z_j=\err_{\mathcal D}(\tilde h,c^\star)
\qquad\text{and}\qquad
\mathbb E\left[\sum_{j=1}^N Z_j\right]
=\err_{\mathcal D}(h,c^\star).
$$
Since $p_j\le1/M$ and $\sum_j p_j=1$, we have
$\sum_j p_j^2\le1/M$.
Hoeffding's inequality gives
\begin{equation}\label{eq:hoeffding-app}
\Pr\left[
\err_{\mathcal D}(\tilde h,c^\star)
>
\err_{\mathcal D}(h,c^\star)+\frac{\varepsilon}{2}
\right]
\le
\exp\left(
-\frac{\varepsilon^2}{2\sum_{j=1}^N p_j^2}
\right)
\le
\exp\left(-\frac{M\varepsilon^2}{2}\right)
\le\frac{\delta}{2}.
\end{equation}
Combining Eq.~\eqref{eq:hoeffding-app} and the original learning guarantee in Eq.~\eqref{eq:learning-guarantee}, a union bound gives, with probability at least $1-\delta$,
$$
\err_{\mathcal D}(\tilde h,c^\star)\le\eta+\varepsilon.
$$
The partition has polynomially many regions, and each $\tilde h_j$ is a learned Boolean hypothesis or a majority of $k$ learned concepts. The construction therefore takes polynomial time and uses no additional samples or ERM calls.
\end{proof}

\section{Online Frank-Wolfe convergence guarantee}\label{app:D}

Here we provide a self-contained proof of \Cref{lemma:erm-online-frank-wolfe}. In fact, the analysis allows the stated additive error in each linear-optimization step. Specifically, the linear optimization subproblem in \Cref{lemma:erm-online-frank-wolfe} only requires a point $v_t\in K$ such that
$$
\left\langle
\nabla F_t(u_t),v_t
\right\rangle
\ge
\max_{v\in K}
\left\langle
\nabla F_t(u_t),v
\right\rangle
-
\frac{D_2^2\gamma_t}{100}.
$$
We closely follow \cite[Theorem~7.3, Lemma 7.4]{HazanOCO}, noting the differences in our slightly more general setup.

\subsection{
\texorpdfstring
{Proof of \Cref{lemma:erm-online-frank-wolfe}}
{Proof of the online Frank--Wolfe theorem}
}

\begin{proof}

\noindent If $T<4$, the Cauchy--Schwarz inequality gives, for every $u\in K$,
$$
\begin{aligned}
\sum_{t=1}^T
\langle g_t,u-u_t\rangle
&\le
\sum_{t=1}^T
\|g_t\|_2\|u-u_t\|_2
\le
D_2G_2T
\le
9D_2G_2T^{3/4}.
\end{aligned}
$$
For the rest of the proof, suppose that $T\geq4$.

For every $t\in[T]$, $u_t^\star \in \operatorname*{argmax}_{u\in K}F_t(u)$ exists because $K$ is compact. We then observe that $\nabla^2F_t(u)=-2I$, so $F_t$ is a $2$-strongly concave function. Moreover, since $F_t$ is quadratic, for every $u,v\in K$ and $\gamma\in[0,1]$,
\begin{equation}
\label{eq:erm-fw-quadratic-expansion}
F_t\bigl((1-\gamma)u+\gamma v\bigr)
=
F_t(u)
+
\gamma
\left\langle
\nabla F_t(u),v-u
\right\rangle
-
\gamma^2
\lVert v-u\rVert_2^2.
\end{equation}
The quadratic expansion gives the following bound:
\begin{align*}
F_t(u_t^\star)-F_t(u_{t+1})
&=
F_t(u_t^\star)-F_t(u_t)
-
\gamma_t
\left\langle
\nabla F_t(u_t),v_t-u_t
\right\rangle
+
\gamma_t^2
\left\lVert
v_t-u_t
\right\rVert_2^2
\tag{Eq.~\eqref{eq:erm-fw-quadratic-expansion}}
\\
&\le
F_t(u_t^\star)-F_t(u_t)
-
\gamma_t
\left\langle
\nabla F_t(u_t),u_t^\star-u_t
\right\rangle
+
\frac{D_2^2\gamma_t^2}{100}
+
\gamma_t^2
\left\lVert
v_t-u_t
\right\rVert_2^2
\tag{Eq.~\eqref{eq:approx-max}}
\\
&\le
(1-\gamma_t)
\left(
F_t(u_t^\star)-F_t(u_t)
\right)
+
\frac{D_2^2\gamma_t^2}{100}
+
\gamma_t^2
\left\lVert
v_t-u_t
\right\rVert_2^2
\tag{concavity of $F_t$}
\\
&\le
(1-\gamma_t)
\left(
F_t(u_t^\star)-F_t(u_t)
\right)
+
\frac{101D_2^2}{100}\gamma_t^2.
\tag{$\lVert v_t-u_t\rVert_2\le D_2$}
\end{align*}
As a result,
\begin{equation}
\label{eq:erm-fw-one-step}
F_t(u_t^\star)-F_t(u_{t+1})
\le
(1-\gamma_t)
\left(
F_t(u_t^\star)-F_t(u_t)
\right)
+
\frac{101D_2^2}{100}\gamma_t^2.
\end{equation}
Moreover, the sequence of objectives has a simple recursive structure:
\begin{equation}
\label{eq:recursion}
F_{t+1}(u)
=
\alpha\sum_{r=1}^{t-1}\langle g_r,u\rangle
-
\lVert u-u_1\rVert_2^2
+
\alpha\langle g_t,u\rangle
=
F_t(u)+\alpha\langle g_t,u\rangle.
\end{equation}
Now, we can use the established recursive relationship in Eq.~\eqref{eq:recursion} to prove the following bound:
\begin{align*}
F_{t+1}(u_{t+1}^\star)-F_{t+1}(u_{t+1})
&=
F_t(u_{t+1}^\star)-F_t(u_{t+1})
+
\alpha
\left\langle
g_t,u_{t+1}^\star-u_{t+1}
\right\rangle
\tag{Eq.~\eqref{eq:recursion}}
\\
&\le
F_t(u_t^\star)-F_t(u_{t+1})
+
\alpha
\left\langle
g_t,u_{t+1}^\star-u_{t+1}
\right\rangle
\tag{optimality of $u_t^\star$}
\\
&\le
F_t(u_t^\star)-F_t(u_{t+1})
+
\alpha G_2
\left\lVert
u_{t+1}^\star-u_{t+1}
\right\rVert_2
\tag{Cauchy--Schwarz}
\\
&\le
F_t(u_t^\star)-F_t(u_{t+1})
+
\alpha G_2
(
F_{t+1}(u_{t+1}^\star)-F_{t+1}(u_{t+1})
)^{1/2}
\tag{str. concavity}
\end{align*}
Since $\alpha G_2 = D_2/(2T^{3/4}) = D_2\gamma_T^{3/2}/(4\sqrt2)$, we move the square-root term to the left and complete the square to obtain
$$
\begin{aligned}
&
\left(
\sqrt{
F_{t+1}(u_{t+1}^\star)-F_{t+1}(u_{t+1})
}
-
\frac{D_2 \gamma_T^{3/2}}{8\sqrt{2}}
\right)^2
\le
F_t(u_t^\star)-F_t(u_{t+1})
+
\frac{D_2^2\gamma_T^3}{128}.
\end{aligned}
$$
Therefore,
\begin{equation}
\label{eq:erm-fw-gap-recursion}
\begin{aligned}
\sqrt{
F_{t+1}(u_{t+1}^\star)-F_{t+1}(u_{t+1})
}
&\le
\sqrt{
F_t(u_t^\star)-F_t(u_{t+1})
}
+
\frac{D_2 \gamma_T^{3/2}}{4\sqrt{2}}
\\
&\le
\sqrt{
F_t(u_t^\star)-F_t(u_{t+1})
}
+
\frac{D_2 \gamma_t^{3/2}}{4\sqrt{2}}.
\end{aligned}
\end{equation}
We claim that, for every $t\in[T]$,
\begin{equation}
\label{eq:erm-fw-tracking}
F_t(u_t^\star)-F_t(u_t)
\le
4 D_2^2 \gamma_t.
\end{equation}
We will prove Eq.~\eqref{eq:erm-fw-tracking} by induction.

\noindent\textbf{Base case: $t=1$.}
Since
$
F_1(u)=-\|u-u_1\|_2^2 \le0
$
for every $u\in K$, with equality at $u=u_1$, we have
$u_1^\star=u_1$. Therefore,
$
F_1(u_1^\star)-F_1(u_1)
=
0
\le
4D_2^2\gamma_1,
$
so Eq.~\eqref{eq:erm-fw-tracking} holds for $t=1$.

\noindent\textbf{Induction step:} Suppose that Eq.~\eqref{eq:erm-fw-tracking} holds for some
$t<T$. By Eq.~\eqref{eq:erm-fw-one-step},
$$
F_t(u_t^\star)-F_t(u_{t+1})
\le
4 D_2^2 \gamma_t - \frac{299 D_2^2}{100}\gamma_t^2.
$$
Consequently, Eq.~\eqref{eq:erm-fw-gap-recursion} gives
\begin{align*}
F_{t+1}(u_{t+1}^\star)-F_{t+1}(u_{t+1})
&\le
\left(
\sqrt{
4D_2^2\gamma_t-\frac{299}{100}D_2^2\gamma_t^2
}
+
\frac{D_2}{4\sqrt2}\gamma_t^{3/2}
\right)^2
\tag{Eq.~\eqref{eq:erm-fw-one-step} and~\eqref{eq:erm-fw-tracking}}
\\
&\le
4D_2^2\gamma_t-\frac{299}{100}D_2^2\gamma_t^2
+
\frac{D_2^2}{\sqrt2}\gamma_t^2
+
\frac{D_2^2}{32}\gamma_t^3
\tag{$\sqrt{4D_2^2\gamma_t-\frac{299}{100}D_2^2\gamma_t^2}
\le2D_2\sqrt{\gamma_t}$}
\\
&\le
4D_2^2\gamma_t-\frac{1}{2}D_2^2\gamma_t^2
\tag{$\gamma_t\le1$ and
$\frac1{\sqrt2}+\frac1{32}\le\frac{249}{100}$}
\\
&\le
4D_2^2\gamma_{t+1}.
\tag{$\gamma_{t+1}\ge\gamma_t-\gamma_t^2/8$}
\end{align*}
This proves Eq.~\eqref{eq:erm-fw-tracking}.

\noindent The $2$-strong concavity of $F_t$ and
Eq.~\eqref{eq:erm-fw-tracking} imply
\begin{equation}
\label{eq:erm-fw-distance}
\|u_t-u_t^\star\|_2^2
\le
F_t(u_t^\star)-F_t(u_t)
\le
4D_2^2\gamma_t.
\end{equation}
We now compare the sequence $(u_t)_{t=1}^T$ with an arbitrary fixed $u\in K$ through the exact maximizers $(u_t^\star)_{t=1}^T$:
\begin{equation}\label{eq:decomp}
\sum_{t=1}^T
\langle g_t,u-u_t\rangle
=
\sum_{t=1}^T
\langle g_t,u-u_t^\star\rangle
+
\sum_{t=1}^T
\langle g_t,u_t^\star-u_t\rangle.
\end{equation}
We next bound the first sum in Eq.~\eqref{eq:decomp}. Since $u_t^\star$ maximizes $F_t$,
$$
\begin{aligned}
u_t^\star
=
\operatorname*{argmax}_{u\in K}
\left(
\alpha\sum_{r=1}^{t-1}\langle g_r,u\rangle
-
\|u-u_1\|_2^2
\right)
=
\operatorname*{argmin}_{u\in K}
\left(
-\frac{\alpha}{2}
\sum_{r=1}^{t-1}\langle g_r,u\rangle
+
\frac{1}{2}\|u-u_1\|_2^2
\right).
\end{aligned}
$$
For the linear functions $u\mapsto-\langle g_t,u\rangle$, the gradients are $-g_t$. With learning rate $\alpha/2$ and quadratic regularizer $u\mapsto\frac{1}{2}\|u-u_1\|_2^2$, the update analyzed in \cite[Theorem~5.2]{HazanOCO} is therefore exactly
$u_t^\star$. Hence,
\begin{equation}\label{eq:regret-type-bound}
\sum_{t=1}^T
\left(
-\langle g_t,u_t^\star\rangle
+
\langle g_t,u\rangle
\right)
\le
2\left(\frac{\alpha}{2}\right)
\sum_{t=1}^T\|-g_t\|_2^2
+
\frac{
\frac{1}{2}\|u-u_1\|_2^2-\frac{1}{2}\|u_1-u_1\|_2^2
}{
\alpha/2
}.
\end{equation}
Therefore, the first sum in Eq.~\eqref{eq:decomp} can be bounded using Eq.~\eqref{eq:regret-type-bound}
\begin{equation}\label{eq:first-term}
\sum_{t=1}^T
\langle g_t,u-u_t^\star\rangle
\le
\alpha\sum_{t=1}^T\|g_t\|_2^2
+
\frac{\|u-u_1\|_2^2}{\alpha}
\le
\frac{D_2G_2}{2}T^{1/4}
+
2D_2G_2T^{3/4}
\le
\frac52D_2G_2T^{3/4}.
\end{equation}
It remains to bound the second sum in Eq.~\eqref{eq:decomp}. Since $\gamma_t=\min(1,2/\sqrt t)$, an integral comparison gives
$$
\sum_{t=1}^T\sqrt{\gamma_t}
=
4+\sqrt2\sum_{t=5}^T t^{-1/4}
\le
4+\sqrt2\int_4^T x^{-1/4}\,dx
\le
2T^{3/4}.
$$
Therefore, by the Cauchy--Schwarz inequality and
Eq.~\eqref{eq:erm-fw-distance},
\begin{equation}
\label{eq:second-term}
\begin{aligned}
\sum_{t=1}^T
\langle g_t,u_t^\star-u_t\rangle
\le
\sum_{t=1}^T
\|g_t\|_2\|u_t^\star-u_t\|_2
\le
2D_2G_2
\sum_{t=1}^T\sqrt{\gamma_t}
\le
4D_2G_2T^{3/4}.
\end{aligned}
\end{equation}
Combining Eq.~\eqref{eq:first-term}
and~\eqref{eq:second-term} with the initial
decomposition gives
$$
\begin{aligned}
\sum_{t=1}^T
\langle g_t,u-u_t\rangle
&\le
\left(
\frac52+4
\right)
D_2G_2T^{3/4}
=
\frac{13}{2}D_2G_2T^{3/4}
\le
9D_2G_2T^{3/4}.
\end{aligned}
$$

\end{proof}

\end{document}